\documentclass[sigconf]{acmart}

\AtBeginDocument{%
  \providecommand\BibTeX{{%
    \normalfont B\kern-0.5em{\scshape i\kern-0.25em b}\kern-0.8em\TeX}}}

\copyrightyear{2021} 
\acmYear{2021} 
\setcopyright{acmcopyright}\acmConference[CCS '21]{Proceedings of the 2021 ACM SIGSAC Conference on Computer and Communications Security}{November 15--19, 2021}{Virtual Event, Republic of Korea}
\acmBooktitle{Proceedings of the 2021 ACM SIGSAC Conference on Computer and Communications Security (CCS '21), November 15--19, 2021, Virtual Event, Republic of Korea}
\acmPrice{15.00}
\acmDOI{10.1145/3460120.3485668}
\acmISBN{978-1-4503-8454-4/21/11}

\usepackage{subfigure}
\usepackage{algorithm}
\usepackage{algpseudocode}
\usepackage{url}
\usepackage{xspace}
\usepackage{multirow}
\usepackage{color}

\usepackage{enumitem}
\usepackage{colortbl}

\definecolor{mygray}{gray}{.9}
\newcommand{\etal}{\textit{et al.}\xspace}
\newcommand{\ie}{\textit{i.e.}\xspace}
\newcommand{\eg}{\textit{e.g.}\xspace}

\newcommand{\etc}{\textit{etc.}\xspace}

\newcommand{\mypara}[1]{\smallskip\noindent\textbf{#1.} \xspace}

\newtheorem{observation}{Observation}
\newcommand{\mydp}{DP\xspace}
\newcommand{\myldp}{LDP\xspace}
\newcommand{\fo}{\ensuremath{\mathsf{FO}}\xspace}
\newcommand{\rr}{\ensuremath{\mathsf{RR}}\xspace}
\newcommand{\myvar}{\ensuremath{\mathsf{Var}}\xspace}

\newcommand{\grr}{\ensuremath{\mathsf{GRR}}\xspace}

\newcommand{\oue}{\ensuremath{\mathsf{OUE}}\xspace}
\newcommand{\olh}{\ensuremath{\mathsf{OLH}}\xspace}
\newcommand{\rappor}{\ensuremath{\mathsf{RAPPOR}}\xspace}

\newcommand{\myuni}{\ensuremath{\mathsf{Uni}}\xspace}
\newcommand{\myhio}{\ensuremath{\mathsf{HIO}}\xspace}
\newcommand{\mydht}{\ensuremath{\mathsf{DHT}}\xspace}
\newcommand{\mycalm}{\ensuremath{\mathsf{CALM}}\xspace}
\newcommand{\myHDG}{\ensuremath{\mathsf{HDG}}\xspace}
\newcommand{\myahead}{\ensuremath{\mathsf{AHEAD}}\xspace}
\newcommand{\PrivTree}{\ensuremath{\mathsf{PrivTree}}\xspace}
\newcommand{\hdmm}{\ensuremath{\mathsf{HDMM}}\xspace}
\newcommand{\PriView}{\ensuremath{\mathsf{PriView}}\xspace}
\newcommand{\myngram}{\ensuremath{\mathsf{Ngram}}\xspace}

\newcommand{\Cauchy}{\emph{Cauchy}\xspace}
\newcommand{\Zipf}{\emph{Zipf}\xspace}
\newcommand{\Laplacian}{\emph{Laplacian}\xspace}
\newcommand{\Gaussian}{\emph{Gaussian}\xspace}

\newcommand{\de}{\ensuremath{\mathsf{DE}}\xspace}
\newcommand{\lle}{\ensuremath{\mathsf{LLE}}\xspace}
\newcommand{\myErr}{\ensuremath{\mathsf{Err}}\xspace}

\begin{document}
\fancyhead{}

\title{AHEAD: Adaptive Hierarchical Decomposition for Range Query \\ under Local Differential Privacy}

\author{Linkang Du}
\affiliation{%
  \institution{Zhejiang University}
  \country{China}
}
\email{linkangd@gmail.com}

\author{Zhikun Zhang}
\affiliation{%
  \institution{CISPA Helmholtz Center for Information Security}
  \country{Germany}
}
\email{zhikun.zhang@cispa.de}

\author{Shaojie Bai}
\affiliation{%
  \institution{Zhejiang University}
  \country{China}
}
\email{white.shaojie@gmail.com}

\author{Changchang Liu}
\affiliation{%
  \institution{IBM Research}
  \country{USA}
}
\email{changchang.liu33@ibm.com}

\author{Shouling Ji}
\authornotemark[1]
\affiliation{%
  \institution{Zhejiang University \& Binjiang Institute of Zhejiang University}
  \country{China}
}
\email{sji@zju.edu.cn}

\author{Peng Cheng}
\authornote{Peng Cheng and Shouling Ji are the co-corresponding authors.}
\affiliation{%
  \institution{Zhejiang University}
  \country{China}
}
\email{saodiseng@gmail.com}

\author{Jiming Chen}
\affiliation{%
  \institution{Zhejiang University \& Zhejiang University of Technology}
  \country{China}
}
\email{cjm@zju.edu.cn}

\begin{abstract}
For protecting users' private data, local differential privacy (LDP) has been leveraged to provide the privacy-preserving range query, thus supporting further statistical analysis. However, existing LDP-based range query approaches are limited by their properties, \ie, 
collecting user data according to a pre-defined structure. 
These static frameworks would incur excessive noise added to the aggregated data especially in the low privacy budget setting. 
In this work, we propose an Adaptive Hierarchical Decomposition (\myahead) protocol, 
which adaptively and dynamically controls the built tree structure, so that the injected noise is well controlled for maintaining high utility. 
Furthermore, we derive a guideline for properly choosing parameters for \myahead so that the overall utility can be consistently competitive while rigorously satisfying LDP. 
Leveraging multiple real and synthetic datasets, we extensively show the effectiveness of \myahead in both low and high dimensional range query scenarios, as well as its advantages over the state-of-the-art methods. 
In addition, we provide a series of useful observations for deploying \myahead in practice. 
\end{abstract}

\ccsdesc[500]{Security and privacy~Privacy-preserving protocols}

\keywords{Differential Privacy; Range Query; Adaptive Decomposition}
\maketitle
\section{Introduction}
With the increasing incidents of data breaches, such as Facebook \cite{facebookbreach}, Marriott \cite{Marriottbreach}, Exactis \cite{Exactisbreach}, \etc, 
users' privacy has become a serious obstacle in many practical applications. 
As a promising countermeasure, differential privacy (\mydp) \cite{dwork2006differential, dwork2006calibrating} 
has been accepted as the \emph{de facto} standard for protecting data privacy in academia and industry \cite{he2020differential, dwork2014algorithmic, johnson2017practical, li2016differential, papernot2016semi, yang2018crowd, ji2014structural}, 
due to its rigorous theoretical guarantees and independence of attacker's background knowledge. 
DP in the centralized setting requires a \emph{trusted aggregator} that collects sensitive data from users and performs perturbation analysis, and then provides data services by answering queries or publishing synthetic data \cite{ge2020kamino, zhang2020privsyn}. 
  
When there is no trusted aggregator, DP in the centralized setting is no longer applicable and users are often reluctant to share their private data without protection. 
To address this obstacle, local differential privacy (LDP) \cite{raskhodnikova2008can, duchi2013local} is proposed, which allows individuals to encode and perturb their private data locally. 
In recent years, \myldp has been deployed by many well-known leading companies, 
including Google \cite{erlingsson2014rappor, fanti2016building}, Apple \cite{APPLE-DP-pdf} and Microsoft \cite{ding2017collecting}. 
For example, Google collects users' favorite homepages and Apple analyzes users' emoji preferences with LDP. 

Previous studies \cite{erlingsson2014rappor, bloom1970space, bassily2015local, wang2016private, wang2017locally} on LDP mainly focus on obtaining frequency distribution throughout the entire domain, \ie, frequency oracle (\fo) \cite{wang2019consistent}. 
However, in practice, people may be more interested in a range query, \ie, estimating the frequency in a certain range of a domain. 
For instance, 
supermarkets are interested to  know the proportion of their high-income customers, e.g., earning between 100K to 120K dollars annually, to make commercial policies. 
Furthermore, based on range query results, 
we can directly obtain other distribution features such as 
order statistics \cite{nissim2007smooth}. 

For range query, recent main-stream solutions can be divided into two categories by query dimension. 
For low($\leq 2$)-dimensional
query scenes, 
Wang \etal \cite{wang2019answering} proposed to hierarchically decompose the entire domain based on the complete $B$-ary tree structure and answer the range query by accumulating the frequency values, 
which was originally developed by Hay \etal \cite{hay2010boosting} in the centralized setting. 
Cormode \etal \cite{cormode2019answering} proposed to apply the discrete wavelet transformation (based on a full binary tree structure over the domain) to convert each user's private value to a Haar wavelet coefficient vector for perturbation 
and perform inverse transformation to get the query answer, 
as a generalization of \cite{xiao2010differential} under the centralized \mydp setting. 
For high($\geq 2$)-dimensional range query, 
Yang \etal \cite{yang2020answering} proposed to combine information from $1, 2$-dimensional grids, 
which was originally proposed by Qardaji \etal \cite{qardaji2013differentially} in the centralized setting, 
and leverage the weighted update strategy to estimate the high-dimensional range queries. 
However, the existing methods have several limitations.  
First, there exist sparse areas in the data domain of most real-world datasets.
For instance, 50-60 years old people account for a small ratio among the members of a football club.
Therefore, the nodes (cells) with small values in the complete tree (grid) are highly likely to be overwhelmed by the injected noises.
In addition, existing techniques are mainly designed for specific dimensional queries, \ie, \cite{wang2019answering, cormode2019answering} for $1, 2$-dim queries and \cite{yang2020answering} for high($\geq 2$)-dimensional queries. 
Although \cite{wang2019answering, cormode2019answering, yang2020answering} are not technically limited by query dimensions, they are less effective in the case of non-target dimensions. 
Since the dimensions of datasets are various in practice, the aggregator needs to combine the algorithms for different scenarios, 
thus limiting the adaptability and applicability of these algorithms. 

To suppress the excessive injected noises, \myahead provides a fine domain decomposition mechanism to accommodate the injected noise of nodes with various granularities in the tree. 
In order to enable \myahead to find the proper domain decomposition, 
we carefully analyze the error source of the query answer obtained by \myahead, 
and provide a guideline to obtain the decomposition. 
After \myahead completes the interaction with all users, there exist certain constraints on nodes' values, \eg,
the sum of the children's values is equal to their parent's value. 
Thus, a post-processing method is designed for \myahead to further boost the query accuracy. 
For high-dimensional queries, we compare two different expansion methods, \ie, Direct Estimation (\de) and Leveraging Low-dimensional Estimation (\lle), 
and show the advantage of \lle based on experimental results. 

To validate the effectiveness of \myahead, we use multiple real and synthetic datasets to show the consistent advantage of \myahead over the state-of-the-art methods. 
Specifically, on several real datasets, 
\myahead can achieve significantly smaller estimation errors as compared to previous works by up to two orders of magnitude. 
For low dimensional scenarios, we evaluate various combinations of essential parameters, \eg, privacy budget, domain size, user scale and distribution skewness, and then provide a comprehensive understanding of \myahead 
by considering 757 parameter combinations in order to guide its adoption in practice. 
For high dimensional scenarios, we investigate the query accuracy of \myahead in different data dimensions and attribute correlations, and show the characteristics of \myahead and competitors based on experimental results. 

In summary, the contributions of this paper are three-fold:
\begin{itemize}[leftmargin=*]
  \item We propose a dynamic algorithm for range query under LDP, which can adaptively determine the granularity of the domain composition. 
  As compared to the state-of-the-art techniques, \myahead can reduce the impact of the inserted noise to range queries for maintaining outstanding utility performance. 
  \item We theoretically derive the parameter settings (decomposition threshold and tree fanout) for the consistent high utility under rigorous LDP guarantees. 
  Furthermore, we extend our strategy to multi-dimensional scenarios.
  \item Through extensive experiments, we demonstrate the effectiveness of \myahead on multiple real-world datasets as well as its advantages over previous approaches in balancing the utility and privacy tradeoff. 
  In addition, we present six useful observations for deploying \myahead in practical use. 
\end{itemize}

\section{Background}
\label{Background}
\renewcommand\arraystretch{1.1}
\begin{table}[!t] \small
  \label{tab:my_label}
  \centering
  \caption{Summary of mathematical notations. }
  \vspace{-0.2cm}
    \begin{tabular}{c c}
      \toprule
         \bf{Notation} & \bf{Description}\\
      \midrule
         $N$ & The total number of users (user scale)\\
         \rowcolor{mygray}
         $D$ & Private attribute domain\\
         $B$ & Tree fanout\\
         \rowcolor{mygray}
         $m$ & The number of private attributes\\
         $\epsilon$ & Privacy budget\\
         \rowcolor{mygray}
         $\theta$ & Threshold for intervals decomposition\\
      \bottomrule
    \end{tabular}
    \vspace{-0.3cm}
\end{table}

\subsection{Local Differential Privacy}
In LDP, 
each user perturbs his/her private data $v$, through a perturbation mechanism $\Psi$, and then transmits $\Psi(v)$ to the aggregator while satisfying rigorous LDP guarantees defined as below. 
\begin{definition}{$\epsilon$$-$Local Differential Privacy ( $\epsilon$-LDP ) \cite{kasiviswanathan2011can}. } 
A perturbation function $\Psi(\cdot)$ satisfies $\epsilon$-LDP if and only if for $\epsilon > 0$ and all possible pairs of input $v_1$, $v_2 \in D$, we have
\begin{equation}
    \forall T \in \operatorname{Range}(\Psi): \Pr{\Psi\left(v_{1}\right) \in T} \leq e^{\epsilon} \Pr {\Psi\left(v_{2}\right) \in T},  \nonumber
\end{equation}
where $\operatorname{Range}(\Psi)$ denotes the set of all possible outputs of \xspace $\Psi$.
\end{definition}

\subsection{Frequency Oracle}
\label{Frequency Oracle}

The frequency oracle (\fo) protocol is used to estimate the frequency distribution $F$ across a private attribute, serving as a basic building block for general LDP tasks such as marginal release \cite{zhang2018calm} and range query \cite{wang2019answering, cormode2019answering}. 
Most \fo protocols consist of three steps: $\operatorname{Encoding}$, $\operatorname{Perturbation}$ and $\operatorname{Aggregation}$ \cite{wang2017locally}. We introduce two state-of-the-art \fo protocols in the following.
\subsubsection{Generalized Randomized Response (\grr)}

The \grr algorithm is a generalized version of random response \cite{warner1965randomized}. 

\mypara{Encoding} \grr directly perturbs on private value $v$, thus the encoded value $x_i$ equals to $v_i$ for user $i$. 

\mypara{Perturbation} User $i$ keeps $v_i$ with probability $p = \frac{e^{\epsilon}}{e^{\epsilon}+|D|+1}$ and randomly chooses $v_i^{\prime} \in D$ s.t. $v_i \neq v_i^{\prime}$  with probability $q=\frac{1}{e^{\epsilon}+|D|+1}$, then uploads $x_i^{\prime}$ to the server, where $x_i^{\prime} \coloneqq \operatorname{Perturb}(x_i)$. 

\mypara{Aggregation} 
The aggregator counts how many times $v$ is reported, denoted by $\operatorname{count}[v] = \sum_{i = 1}^{N} \mathbb{I}_{\{x^{\prime}=v\}}$. An unbiased estimation of the frequency of $v$ is  
$ \hat{f_v} = \frac{\operatorname{count}[v]- Nq}{N(p-q)}$. 

\mypara{Estimation Error} 
$\hat{f_v}$ is an unbiased estimation of the true frequency $f_v$ \cite{wang2019answering}. 
Therefore, the estimation error of \grr originates from the algorithm variance 
\begin{small}
  \begin{align}
    \myvar_{\grr(\epsilon)} = \frac{|D| - 2 + e^{\epsilon}}{N\left(e^{\epsilon}-1\right)^{2}}
    \label{GRRVAR}
    \end{align}  
\end{small}
\subsubsection{Optimized Unary Encoding (\oue)}
\label{myoue}
\oue \cite{wang2017locally}  is an optimization of the basic \rappor protocol in \cite{erlingsson2014rappor}. 

\mypara{Encoding} 
User $i$ encodes his/her private data into a one-hot binary vector, \ie, $ x_i = [0,\ 0\ \ldots,\ 1,\ \ldots,\ 0] $
of length $|D|$, where only the $v_i$-th position is $1$.

\mypara{Perturbation} User $i$ flips each bit of $x_i$ based on probabilities $p = \frac{1}{2}$ and $q = \frac{1}{e^\epsilon + 1}$ as below, 
while transmitting 1’s and 0’s differently. 
1's (resp., 0's) keeps with the probability of $p$ (resp., $1-q$) and flips to the reverse with the probability of $1-p$ (resp., $q$). 
Then user $i$ uploads $x_i^{\prime}$ to the server. 

\mypara{Aggregation} 
The aggregator collects $\{x_i^{\prime}\}_{i=1}^{N}$ uploaded by the users and counts the number of occurrences of 1 in each bit, \eg, for the $v$-th bit, $\operatorname{count}[v] = \sum_{i = 1}^{N} x_i^{\prime}[v]$. 
The $\operatorname{count}[v]$ needs to be corrected to obtain an unbiased estimation
$ \hat{f}[v] = \frac{\operatorname{count}[v] - Nq}{N(p - q)} $. 

\mypara{Estimation Error} 
It is proved in \cite{wang2017locally} that \oue has variance 
\begin{equation}
    \myvar_{\oue(\epsilon)} = \frac{4e^{\epsilon}}{N\left(e^{\epsilon}-1\right)^{2}}
    \label{OUEVAR}
\end{equation}  
Both \grr and \oue achieve unbiased estimation of frequency values. 
As shown in \autoref{GRRVAR} and \autoref{OUEVAR}, \oue has a variance that is independent of $|D|$. 
For smaller $|D|$ (such that $|D|-2<3e^\epsilon$), \grr is better; while \oue is superior for larger $|D|$. 
\section{Problem Definition and Existing Solutions}
\subsection{Range Query Problem}
\begin{figure}[h]
    \centering
    \includegraphics[width=0.5\hsize]{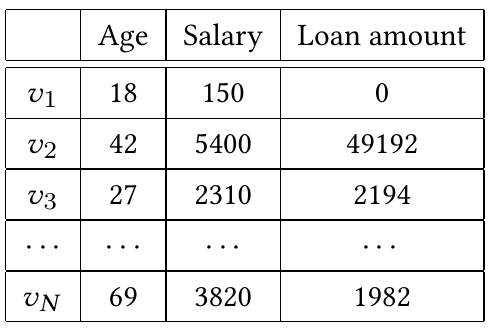}
    \vspace{-0.2cm}
    \caption{An example database containing $N$ users with three attributes: \emph{age}, \emph{salary} and \emph{loan amount}. }
    \label{Example dataset}
    \vspace{-0.3cm}
  \end{figure}

Assume there are $N$ users, where the $i$-th user has an $m$-dim ordinal record ${\bf{v}}^{m}_i = (v^{(1)}_i, v^{(2)}_i, \cdots, v^{(m)}_i)$, with $v_i^{(j)}$ representing the $j$-th private attribute value owned by user $i$. 
Denote the domain for the $j$-th item as $D_j$. Given a series of ranges $\alpha_j$, $\beta_j$ ( $j=1,2,\cdots,m$ ), an $m$-dim range query can be computed as 
\begin{small}
\begin{align}
    R_{\bigcap {[\alpha_j,\beta_j]}_{j=1}^{m}}=\frac{1}{N} \sum_{i=1}^{N} \mathbb{I}_{\bigcap{\{\alpha_j \leq v_i^{(j)} \leq \beta_j\}_{j=1}^{m}}}  \nonumber
\end{align}
\end{small}where $\mathbb{I}_\gamma$ is an indicator function that takes 1 if the predicate $\gamma$ is true and 0 otherwise.
\autoref{Example dataset} gives a running example of range query. 
For example, the proportion of people within 20 years to 40 years old constitutes a 1-dim range query, 
while the ratio of people within 20 years to 40 years old, with salary less than 5000, and with loan amount less than 20000 constitutes a 3-dim range query, 
where the three dimensions corresponding to age, salary and loan amount, respectively.

\subsection{Hierarchical-Interval Optimized (\myhio)}
\label{exsitingHIO}
Based on a $B$-ary tree, \myhio \cite{wang2019answering} hierarchically decomposes the entire domain into mutually disjoint subsets called \emph{intervals}. 
The root node represents the entire domain, and the leaf nodes represent the individual values. 
Nodes on the same layer represent intervals of the same granularity. 
Then, \myhio obtains the frequency estimations of nodes in each layer by the \oue \cite{wang2017locally} algorithm. 
When answering a range query, 
\myhio completely covers the query range by using the minimum number of intervals from different layers. 

For example, when the users' private attribute domain size $|D| = 8$ and tree fanout $B = 2$, the range query $[2,7]$ can be decomposed into intervals $[2,3] \cup [4,7]$. 
Then, \myhio adds the estimated frequency values of the two intervals above to get the answer of a range query. 
For general query with range length $r$, \myhio can answer it with at most $2(B-1)\log_B|D|$ intervals. 
Compared with directly using \fo mechanisms, \myhio can effectively reduce the number of intervals used when answering queries, thus substantially reducing the cumulative error caused by adding noisy frequency values of intervals within the range. 

However, \myhio has two weaknesses that limit its applicability in practice. 
1) \myhio inserts the same level of noise into the estimated frequencies of all the intervals. 
For nodes with small intervals, the perturbation noise often overwhelms the true frequency values thus degrading the utility of the entire algorithm. 
2) In the multi-dimensional scene, the number of tree layers increases exponentially with the number of dimensions. 
For high-dimensional scenarios, 
the query error increases extremely with the excessive small value nodes. 

\subsection{Discrete Haar Wavelet Transform (\mydht)}
\label{exsitingDHT}

\mydht \cite{cormode2019answering} imposes a full binary tree structure over the domain, and encodes the user private value $v$ into a set of \emph{Haar wavelet coefficients}.
The motivation underlying \mydht \cite{cormode2019answering} is that the calculation of a length-$r$ range query uses only a smaller number of estimated values in the Haar wavelet domain, comparing to apply \fo directly. 

\mydht also faces several limitations. 
1) Similar to \myhio, \mydht inserts the same level of noise into all estimated Haar wavelet coefficients. 
For some coefficients with low values, noise tends to skew the estimated coefficients causing query error to increase. 
2) It is mainly designed for 1-dim scenario, thus limiting its application in practice.

\subsection{Consistent Adaptive Local Marginal (\mycalm)}
\label{Consistent Adaptive Local Marginal}
\mycalm \cite{zhang2018calm} is a marginal release LDP protocol, which can construct the joint distribution of $m$ attributes with privacy protection guarantee. 
Instead of directly estimating all marginal tables, \mycalm strategically chooses the size and the number of marginal tables, based on which all the marginal tables can be reconstructed. 
We notice that \mycalm can be used to answer range queries. 
More specifically, to answer a multi-dimensional range query, \mycalm can sum up the reconstructed marginals included in the query. 

However, when the domain size $|D|$ is large, \mycalm needs to sum up extensive noisy marginals to answer a query, which is likely to inject a large amount of noise to the true answer. 

\subsection{Hybird-Dimensional Grids (\myHDG)}

\myHDG \cite{yang2020answering} is the state-of-the-art method to answer multi-dimensional range query under LDP. 
The main idea of \myHDG is to carefully bucketize the 2-dim domains of all attribute pairs into coarse 2-dim grids and then estimate the answer of a higher dimensional range query from the answers of the associated 2-dim range queries. 
To capture the fine-grained distribution information for users' data, \myHDG also introduces 1-dim grids to offer finer-grained distribution information on each attribute and combines information from 1-dim and 2-dim grids to answer range queries. 

\myHDG also faces several limitations. 
1) The equal granularity grids of \myHDG cannot handle various distributions of users' data. For skew-distributed datasets, whose data is concentrated in a small part of the whole domain, noise error or non-uniform error dominates in some grids thus degrading the utility. 
2) Using 1-dim grids may destroy the correlation between attributes. 

\subsection{Remarks}
\label{Remarks}
To overcome the limitations of the state-of-the-art low-dimensional mechanisms (\myhio, \mydht) and high-dimensional mechanisms (\mycalm, \myHDG),  
we aim to achieve the following two design goals: 
1) find a reasonable decomposition for the domain to avoid introducing excessive noise; 
2) the designed mechanism can be extended to multi-dimensional scenarios, with better query accuracy than existing algorithms. 
Motivated by these goals, 
we propose \myahead, which differs from the existing work in several major aspects: 
1) \myahead is an adaptive and dynamic algorithm, compared to the existing algorithms with static frameworks; 
2) \myahead reduces the impact of noise on small value nodes by merging intervals. 
3) The designed mechanism can migrate from 1-dim to multi-dimensional scenarios. 
Next, we will illustrate the motivation and design of \myahead in detail. 

\section{\myahead: Adaptive Hierarchical Decomposition}
\label{Adaptive Hierarchical Decomposition}
\subsection{Motivation and Overview}
\label{Motivation and Overview}
In this subsection, we use an example to illustrate the limitations of the existing algorithms and the rationality of \myahead. 
As shown in \autoref{Method Motivation}, the tables on the left show the intervals with corresponding real frequency values. 
For instance, the true frequency value of interval $[0,1]$ is $0$, meaning that there is no user data reside in this interval. 
Then, the middle tables display the frequency values of the intervals, separately estimated by different strategies. 
$\sigma^2$ represents the variance of the noise introduced in the perturbation process. 
The remaining part on the right shows the process of answering the query based on the estimated values. 

Firstly, we focus on the process of the baseline strategy, such as \myhio. 
The baseline strategy chooses to publish the estimated frequency of each interval. 
Each estimated value integrates a \emph{noise error} by the \fo mechanism to meet the LDP guarantees. 
For interval $[0,1]$, its true value is 0, meaning that the estimated values 
for these intervals are completely filled with noise. 
When answering a range query, such as $frequency([0, 5])$, the questioner wants to know the frequency value of interval $[0, 5]$. 
The answer of the baseline strategy is $0.5 + 3\sigma^2$. 
It is worthy noting that interval $[0,1]$ does not contribute to the query answer but bring the same degree of noise, which reduces the query accuracy of the existing algorithms. 

On the other hand, for the adaptive strategy, as if we know the true frequency values of the intervals, we can combine the intervals $n_0$ and $n_1$ and estimate a single value for $n_p$. 
When answering the same range query $frequency([0,5])$, the answer of the adaptive strategy is $0.5 + 2\sigma^2$, which reduces the noise error by 30\% compared to the baseline method. 
We attenuate the noise error for the intervals with small frequency values, and the adaptive strategy works better in this case. 

The combination of intervals will reduce the impact of noise on intervals with small frequency values. 
However, when a query falls within an interval, the answer has to be approximated by the assumption about the distribution within the interval. 
Making the \emph{uniform distribution assumption} is a dominant strategy \cite{ioannidis2003history}, where the value of each record in the interval is the same. 
When the assumption is not satisfied, it leads to a \emph{non-uniform error}. 
For instance, if the questioner wants to know $frequency([2, 3])$, the answer of the adaptive strategy is calculated from the frequency value of interval $[0,3]$, \ie, the half of the frequency value of interval $[0,3]$. 
Compared with the baseline strategy, the adaptive strategy reduces the noise error from $\sigma^2$ to $\frac{\sigma^2}{2}$, 
while it also brings non-uniform error $0.05 - \frac{0.05}{2} = 0.025$. 

Adaptivity reduces the noise error by merging intervals, while introducing the non-uniform error by assuming uniformity. 
Therefore, to reduce the overall query error, we aim to find the optimal domain decomposition through balancing these two errors. 
However, finding the optimal partition for 2-dim datasets is diffcult \cite{muthukrishnan1999rectangular},
which is even worse with privacy constraint. 
Inspired by the above example, 
we propose a multi-phase hierarchy based recursive partitioning strategy (detailed in \autoref{Workflow of ahead}) that seeks to balance the errors and address the limitations of the existing solutions. 

\begin{figure}[!t]
    \centering
    \includegraphics[width=0.9\hsize]{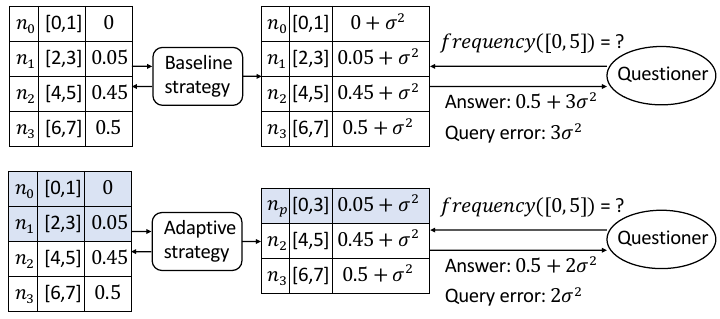}
    \vspace{-0.2cm}
    \caption{Baseline strategy vs. Adaptive strategy. 
    }
    \label{Method Motivation}
    \vspace{-0.3cm}
\end{figure} 
\subsection{Workflow of \myahead}
\label{Workflow of ahead}
\begin{figure*}[t]
    \centering
    \includegraphics[width=\hsize]{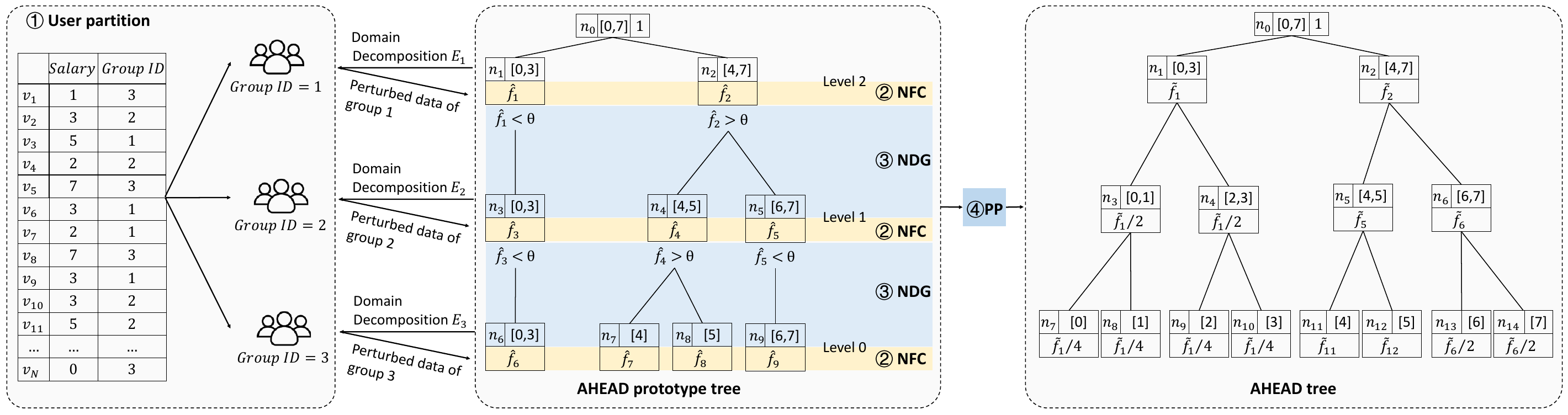}
    \vspace{-0.4cm}
    \caption{Workflow of \myahead. From left to right, the four steps in the \myahead algorithm, \ie, user partition, noisy frequency construction,  new decomposition generation and post-processing, are shown respectively. 
    \myahead answers the range queries based on the tree in the rightmost sub-figure.} 
    \label{AHEAD method}    
    \vspace{-0.4cm}
  \end{figure*}
In this subsection, we show the workflow of \myahead with an example as shown in \autoref{AHEAD method}. 
In this example, the aggregator wants to complete the range query task about the user’s salary based on \myahead. 
The salary data is bucketized into 8 ordinal levels, \ie, domain size $|D| = 8$. 
The tree fanout $B = 2$, meaning that each node of the \myahead tree has at most two child nodes. 
In each node of the \myahead prototype tree, $n_i$ represents the node index, $[a, b]$ represents the node's interval, and $\hat{f_i}$ represents the estimated frequency value. 
It is worth noting that \myahead adopts the \emph{sampling principle} \cite{cormode2019answering}, \ie, partitioning users into groups with each group using the full privacy budget). 
\emph{Sampling principle} can significantly reduce the overall error in local setting \cite{nguyen2016collecting, wang2017locally, wang2019locally} (refer to more details in \autoref{The rationality of sampling principle}). 
Below, we divide the workflow of \myahead into four steps and describe the steps in detail. 

\mypara{Step 1: User Partition (UP)} 
As shown in the left dashed box in \autoref{AHEAD method}, the aggregator determines the number of partitions $c$, 
where $c = \log_B|D|$ is set to ensure that users are assigned to each layer of the \myahead tree. 
The users randomly choose the group number in range $[1, 2, 3, \cdots, c]$. 
In addition, the users can also leverage their public information to select groups, 
such as the time of account registration, user ID, \etc
The partition process should ensure that each group is representative of the overall population with a similar number of users. 

\mypara{Step 2: Noisy Frequency Construction (NFC)} 
In the middle dashed box, the aggregator first establishes a root node $n_0$ representing the entire domain. 
After that, the aggregator performs the initial decomposition of the domain, 
\ie, dividing the entire domain into $B$ equal-sized intervals, 
then attaches the interval nodes to the root node $n_0$.
The children of the root node represents a way to divide the whole domain, denoted as domain decomposition $E_1$. 
The aggregator selects the first group of users and sends the decomposition $E_1$ and privacy budget $\epsilon$ to them. 
Each user in the first group projects his/her private value $v$ onto the intervals of $E_1$ and uploads the projected value of $v$ via \oue. 
After receiving users' reports, 
the server uses the aggregation algorithm to obtain the estimated frequency distribution $\hat{F_1}$, which represents the ratio of users falling within each node's interval. 

\mypara{Step 3: New Decomposition Generation (NDG)} 
The aggregator compares each frequency value of $\hat{F_1} = \left\{\hat{f_1}, \hat{f_2}\right\}$ with a threshold $\theta$ and decides whether to divide the corresponding interval of $E_1 = \left\{[0, 3], [4, 7]\right\}$ further. 
To be specific, since $\hat{f_2}$ is greater than the setting $\theta$, the corresponding interval $[4,7]$ in $E_1$ should be divided into $B$ equal-sized sub-intervals $[4,5]$, $[6,7]$. 
While the frequency value $\hat{f_1}$ of node $n_1$ is not greater than $\theta$, interval $[0,3]$ does not need further partition. 
For the new interval nodes, we attach them to the corresponding parent interval nodes. 
When all the elements in $\hat{F_1}$ are traversed completely, 
we can obtain a new set of intervals serving as decomposition $E_{2}$. 

Then, the aggregator sends decomposition $E_{2}$ to the second group of users and obtains the estimated frequency distribution $F_2$. 
The aggregator repeats the above steps until all user groups are applied and gets an \myahead prototype tree as shown in the middle dashed box of \autoref{AHEAD method}. 
Since the estimated frequency is less than the threshold $\theta$, \myahead will not decompose the intervals $[0,3]$ and $[6,7]$ in subsequent interactions. 
To guarantee LDP, $[0,3]$ and $[6,7]$ should be estimated by all groups of users. 

While constructing the prototype tree, 
\myahead estimates each layer separately, 
which does not consider the constraint of frequency values in the tree, 
\ie, the sum of the child nodes' frequency values is equal to that of their parent node. 
Therefore, in Step 4, 
we further boost the accuracy of \myahead by conducting \emph{non-negativity} and \emph{weighted averaging} between the nodes' estimations. 

\mypara{Step 4: Post-processing (PP)} 
The post-processing module contains two steps: \emph{non-negativity} and \emph{weighted averaging}. 

Firstly,  \myahead processes the nodes in the same layer by Norm-Sub \cite{wang2019consistent} to ensure that the estimated frequencies of nodes are non-negative and the sum of the frequencies is equal to 1. 
\myahead converts the negative value into 0 and calculates the total difference between the sum of positive values and 1. 
Next, each positive value subtracts the average difference, which is obtained by dividing the total difference by the number of positive estimated values. 
The \emph{non-negativity} process repeats until all values become non-negative. 

Then, from bottom to top, \myahead calculates the weighted average between non-leaf node $n$ and its children to update the estimated frequencies of $n$, \ie, reducing the added noise by fusing multiple estimations of $n$. 
For a non-root node $n$: 
\begin{small}
\begin{equation}
    \label{weighted average}
    \tilde{f}(n)=\left\{
    \begin{array}{rcl}
    \lambda_1\hat{f}(n) + \lambda_2 \sum_{u \in child(n)}\hat{f}(u), &   \mbox{if\ u\ is\ a\ leaf\ node} \\
    \lambda_1\hat{f}(n) + \lambda_2 \sum_{u \in child(n)}\tilde{f}(u),    & \mbox{o.w.}  \\
    \end{array}\right.
\end{equation}
\end{small}
The weights $\lambda_1$ and $\lambda_2$ are inversely proportional to the variance of the estimates, 
\ie, $\lambda_1 = \frac{\myvar_{child(n)}}{\myvar_{child(n)} + \myvar_{(n)}}$ and $\lambda_2 = \frac{\myvar_{(n)}}{\myvar_{child(n)} + \myvar_{(n)}}$, 
where $\myvar_{child(n)}$ represents the sum of node $n$'s children variances, and $\myvar_{n}$ indicates the variance of node $n$. 
$\tilde{f}$ indicates the post-processed version of $\hat{f}$ and will be used to answer queries. 
The weighted average process can minimize the magnitude of noise as shown in the following theorem. 
The proof of \autoref{theorem: weighted average} can be found in \autoref{The proof of weighted average}.
\begin{theorem}
\label{theorem: weighted average}
    Using \autoref{weighted average} to combine the frequencies of child nodes, the node $n$ can achieve the minimal updated variance.
\end{theorem} 

Finally, from top to bottom, 
\myahead decomposes the frequency value recursively under the uniform distribution assumption of the node's interval to obtain a complete tree (shown in the rightmost sub-figure of \autoref{AHEAD method}), 
which will be used to answer range queries. 

\subsection{Privacy and Utility Analysis}
\label{Privacy and Utility Analysis}
\mypara{Privacy Guarantee} 
\myahead is sequentially interactive \cite{joseph2019role, joseph2020exponential, duchi2013local, acharya2020interactive, kasiviswanathan2011can}, \ie, each user communicates once (Step 2 in \autoref{Workflow of ahead}) but the randomization depends on earlier user's messages (Step 3 in \autoref{Workflow of ahead}). 
Since the private data of each user is transmitted to the aggregator once via \oue with privacy budget $\epsilon$ (no other information of the users is leaked), 
we claim that \myahead rigorously satisfies $\epsilon$-LDP and 
the proof is deferred to \autoref{AHEAD satisfies LDP} due to the space limitation. 

\mypara{Error Analysis}
The overall error between the true query answer and the estimated answer originates from three sources of errors. 

\emph{Noise and Sampling Errors} originate from the \oue's perturbation and the user sampling processes. 
As shown in \autoref{myoue}, although \oue can get an unbiased estimation of the frequency values, there is still an estimation variance caused by perturbation.
In addition, \myahead divides users into $c$ groups and uses each user group to represent the frequency estimation from the entire population. 
Based on the analysis in \cite{yang2020answering}, 
the sampling error is a constant which is much smaller than the inserted noise. 
Since each user randomly chooses one of the $c$ groups to report private data, the population of each group approximates to $\frac{N}{c}$. 
By \autoref{OUEVAR}, 
the variance of perturbed noise $X$ is proportional to the number of groups $c$, \ie, $\sigma^2 = c \cdot \frac{4e^{\epsilon}}{N\left(e^{\epsilon}-1\right)^{2}}$.
Due to the threshold setting, some fine-grained intervals' frequency values may not be directly estimated. 
For the non-estimated intervals, their frequency values should be calculated from the larger intervals, \ie, the higher-level intervals (such as parent nodes) in the \myahead tree. 
If a non-estimated interval's frequency value is calculated from a larger interval whose size is $k$ times of the non-estimated interval's size, 
\myahead assigns $\frac{1}{k}$ of the large interval's value to the non-estimated interval. 
Thus, the noise error of the non-estimated interval can be viewed as originating from a random variable $\frac{X}{k}$. 

\emph{Non-uniform Error} arises from some intervals whose values are approximated by larger intervals' values in the \myahead tree. 
For a non-estimated interval $n$ whose true frequency value is $f_n$, 
the size of the larger interval is $k$ times that of $n$ ($k$ is the same as that in the noise and sampling errors part). 
During the calculation, it is assumed that the private values in the larger interval satisfy a uniform distribution \cite{qardaji2013differentially}. 
Thus, the assigned frequency value of the interval $n$ is $\frac{f_p}{k}$, where $f_p$ is the true frequency value of the larger interval. 
If the values in interval $p$ satisfy the uniform distribution, there is no non-uniform error, \ie, $f_n = \frac{f_p}{k}$. 
When the values in interval $p$ do not meet the uniform distribution assumption,  
the deviation of the interval $n$'s frequency value is $|f_n-\frac{f_p}{k}|$. 
Thus, the non-uniform error is influenced by the true distribution of the interval $p$. 
If the distribution is closed to the uniform distribution, the non-uniform error becomes small. 
Otherwise, the non-uniform error will increase, and the upper bound of the uniform error depends on the true frequency value $f_p$, \ie $|f_p-\frac{f_p}{k}|$. 
For the entire domain, when the frequencies are more uniformly distributed across nodes, \myahead behaves better owing to smaller non-uniform errors. 

\begin{algorithm}[!t]
    \caption{1-dim \myahead Tree Construction}
    \label{Construct 1-dim prototype tree}
	\begin{algorithmic}[1]
        \Require All users’ value set $V = \{v_1, v_2,\ldots,v_N \}$, attribute domain $D$, tree fanout $B$, privacy budget $\epsilon$, threshold $\theta$
        \Ensure \myahead Tree $T$
		\State $c = {\log_B}|D|$
		\State $//$ Step 1: User partition
        \State Randomly divide users into $c$ parts $\left\{ V_1, V_2, \ldots, V_c \right\}$
        \State Create the root node of tree $T$ with initial interval $e^0_0 = [1, |D|]$ and $T.\operatorname{node}(e^0_0).frequency = 1$
		\For {$i$ from 1 to $c$}
            \State $//$ Step 2: New decomposition generation
            \For {$j$, node in enumerate($T.\operatorname{node}(level=i-1)$)}
            \If {$\operatorname{node.frequency} > \theta$}
            \State Divide $\operatorname{interval} e^{i-1}_j$ into $B$ disjoint intervals $\{e^{i-1}_{j,k}\}$
            \For {$k$ from $1$ to $B$}
    		\State $\operatorname{node.add\_child}(e^{i-1}_{j,k})$
            \EndFor
            \Else
            \State $\operatorname{node.add\_child}(e^{i-1}_{j})$
            \EndIf    		
            \EndFor

            \State $//$ Step 3: Noisy frequency construction
            \State $F = \fo(V_i, T.\operatorname{node}(level=i).\operatorname{interval}, \epsilon)$
            \For {$k$, node in enumerate $T.\operatorname{node}(level=i)$}
    		\State $node.\operatorname{frequency} = F[k]$
            \EndFor
    
        \EndFor
        \State $//$ Step 4: Post-processing
        \State Run \autoref{alg:1-dim post_processing}
        \State Return $T$
	\end{algorithmic}
\end{algorithm}
\begin{algorithm}[!t]
	\caption{Post-processing}
	\label{alg:1-dim post_processing}
	\begin{algorithmic}[1]
        \Require \myahead tree $T$, tree fanout $B$
        \Ensure \myahead tree $T$
        \For {$i$ from $1$ to $c$}
		\State norm\_sub($T.\operatorname{node(level=i)}.\operatorname{frequency}$)
        \EndFor    
        \For {$j$ from $c-1$ to $1$}
            \For {\_, node in enumerate $T.\operatorname{node}(level=j)$}
            \State $f_1$ = node.frequency, $f_2$ = $\sum{\operatorname{node.children().frequency}}$
		    \State node.frequency = $\lambda_1$$f_1$ + $\lambda_2$$f_2$
            \EndFor 
        \EndFor    
        \For {$k$ from $1$ to $c$}
            \For {\_, node in enumerate $T.\operatorname{node}(level=k)$}
                    \If {node.children() == None}:
		            \State $\operatorname{node}.\operatorname{add\_children}()$
		            \State $\operatorname{node}.\operatorname{children()}.\operatorname{frequency} = \operatorname{node}.\operatorname{frequency}/B$
                    \EndIf
            \EndFor 
        \EndFor    
    \end{algorithmic}
\end{algorithm}

\subsection{Selection of $B$ and $\theta$}
\label{Selection parameters}
The most important parameters of \myahead are tree fanout $B$ and threshold $\theta$. 
Since \myahead has already rigorously satisfied LDP guarantees (recall \autoref{satisfy LDP} in \autoref{AHEAD satisfies LDP}), we aim to explore the settings of $B$ and $\theta$ so that the overall utility performance of \myahead can be maximized. 
Due to the partition strategy mentioned above and a large scale of users in actual scenarios ($N > 10^5$), 
we assume that each group has an equal number of users. 
Recalling the error analysis in \autoref{Privacy and Utility Analysis}, we focus on the noise error and non-uniform error, which dominate the overall estimation error. 

\mypara{Choosing $\theta$}
Intuitively, our goal in selecting the parameters of \myahead is to balance the two errors, so that \myahead can achieve an outstanding performance. 
    For a set of parameters, \ie, tree fanout $B$, privacy budget $\epsilon$, user scale $N$ and the number of groups $c$, the decomposition threshold $\theta$ setting follows the formula below. 
    \begin{equation}
        \theta = \sqrt{(B+1)\myvar}, 
        \label{theta setting1}
    \end{equation}    
    where \myvar is equal to $\frac{4e^{\epsilon}c}{N\left(e^{\epsilon}-1\right)^{2}}$, \ie, the variance of each estimated frequency value. 

The analysis to support \autoref{theta setting1} is as follows. 
Recalling the new decomposition generation step in \autoref{Workflow of ahead}, 
\myahead divides each interval separately by comparing the estimated value with the threshold. 
Therefore, our analysis can focus on one of the interval nodes of the \myahead tree. 
Suppose we have a node $n$ with a true frequency value $f$ and use $f_1$, $f_2$, $\cdots$, $f_B$ to denote the true frequency values of its children. 
Without loss of generality, one of $n$'s children frequency value is $\eta f$ and the sum of others is $(1-\eta)f$, where $\eta \in[0, 1]$. 
The parameter $\eta$ is determined by the distribution of the users' data. 
When the distribution is far away from the uniform distribution, 
$\eta$ closes to 0 or 1, which is the boundary of its value domain. 
Otherwise, $\eta$ closes to $\frac{1}{B}$. 
Here, we consider two different strategies mentioned in \autoref{Motivation and Overview} to estimate the frequency values of $n$'s children. 
Firstly, we use the baseline strategy used in \myhio to obtain their frequency values and the expected overall estimation error can be calculated below. 
\begin{align}
    \mathbb{E}\left[\myErr_1\right] & = \mathbb{E}\left[(\hat{f_1} - f_1)^2 + (\hat{f_2} - f_2)^2 + \cdots + (\hat{f}_B - f_B)^2\right]        \nonumber \\
        & = \mathbb{E}\left[(f_1 + X_1 - f_1)^2 +  \cdots + (f_B + X_B - f_B)^2\right] \nonumber \\
        & = \mathbb{E}\left[B \cdot X^2\right] = B\mathbb{E}\left[X^2\right] = B\myvar
    \label{method_1}
\end{align}    

In the derivation of \autoref{method_1}, \oue is conducted once to obtain users’ data distribution over $B$ intervals in each layer. 
$X_i$ represents the perturbed noise added to the frequency of the $i$-th sub-interval. 
For each layer, the number of users in \autoref{OUEVAR} is the same for all intervals, \ie, $\mathbb{E}[X_i^2] = O(c/N)$ for any $i$.
Leveraging the adaptive strategy used in \myahead, we estimate the frequency value of node $n$ and assign the average as the child nodes' values. Then, we obtain the expected estimation error. 
\begin{small}
    \begin{align}
        \mathbb{E}\left[\myErr_2\right] &= \mathbb{E}\left[ (\frac{\hat{f}}{B} - f_1)^2 + (\frac{\hat{f}}{B} - f_2)^2 + \cdots + (\frac{\hat{f}}{B} - f_B)^2 \right]  \nonumber \\
                    & = \mathbb{E}\left[ (\frac{f+X_1}{B} - \eta f)^2 + \cdots + (\frac{f+X_B}{B} - f_B)^2 \right] \label{method_2_1} \\
                    & = \mathbb{E}\left[ \eta^2 f^2 + \sum_{i=2}^B f_i^2 - \frac{f^2}{B}\right] + \mathbb{E}\left[\frac{X^2}{B} \right]  \nonumber \\
                    & \leq \eta^2 f^2 + (1-\eta)^2f^2 - \frac{f^2}{B} + \frac{1}{B}\myvar \label{method_2_6} \\
                    & \leq \frac{1}{B}((B-1)f^2+\myvar)
        \label{method_2}
    \end{align}    
\end{small}
In the derivation of \autoref{method_2}, the error of each child node contains noise error and non-uniform error. 
For instance, the squared error of the first child in \autoref{method_2_1} can be rewritten as $\left(\frac{X_{1}}{B}+\left(\frac{f}{B}-f_{1}\right)\right)^{2}$. 

We let \autoref{method_2} $<$ \autoref{method_1} to ensure that the adaptive strategy has a lower overall estimation error than that of the baseline strategy. Then we get the inequality as follows. 
\begin{align}
    f < \sqrt{(B+1)\myvar}
    \label{theta setting 2}
\end{align}    
We further calculate the frequency values of the nodes on node $n$'s subtree and the adaptive strategy always outranks the baseline strategy. 
Based on the analysis for \autoref{theta setting 2}, 
if the frequency value of a node meets \autoref{theta setting 2}, then it cannot be divided. 
Otherwise, it needs to be further divided. 
Therefore, we set the threshold $\theta=\sqrt{(B+1)\myvar}$ to guarantee that \myahead does not divide the nodes with small frequency values and reduces the estimation error than the baseline strategy. 

It is worth noting that $\theta$ can be less optimal after the post-processing step (step 4 in \autoref{Workflow of ahead}). However, post-processing is correlated with the tree structure, which is unknown when choosing $\theta$. Therefore, post-processing is hard to be incorporated in the theoretical analysis of $\theta$. Our current analysis of $\theta$ is independent of the tree structure (and the input data as well), making the derived $\theta$ generally applicable to any input data.
In the practical implementation of \myahead, we only have access to the estimated frequency value $\hat{f}$, \ie, the true frequency value $f$ with a random noise variable $X$. 
Since \oue is an unbiased protocol, the expected value of $\hat{f}$ is equal to $f$. 
Thus, we still use $\theta$ as the threshold in \myahead and provide a comprehensive validation of threshold choice in \autoref{Validation of Threshold Choice}. 

\mypara{Choosing $B$}
In general, $B$ is used to balance the tree height and the number of nodes required to answer the query. 
Previous studies \cite{wang2019answering, cormode2019answering} select the optimal fanout $B$ around $4$ when only considering noise and sampling errors. 
Different from \cite{wang2019answering, cormode2019answering}, \myahead introduces non-uniform error in the process of merging intervals. 
Compared to the $B$ choice of previous studies, we set $B=2$ considering non-uniformity and provide the analysis in the following. 

For a node $n$ with a true frequency $f$ ($f<\theta$), \myahead does not further decompose the interval of node $n$ due to the threshold setting, where the children of node $n$ can not directly get estimation frequency values in noisy frequency construction (Step 2 in \autoref{Workflow of ahead}). 
From \autoref{OUEVAR}, 
the variance of the perturbed noise $X$ on node $n$ is proportional to the number of groups $c$, \ie, $\sigma^2 = c \cdot \frac{4e^{\epsilon}}{N\left(e^{\epsilon}-1\right)^{2}}$. 
To obtain a complete tree for answering queries, 
\myahead assigns $\frac{1}{B}$ of node $n$'s estimation frequency value to its child nodes in post-processing (Step 4 in \autoref{Workflow of ahead}). 
Then, we can obtain the noise error of node $n$'s child as $\frac{X}{B}$. 
Since we have no prior knowledge of the data distribution, 
for non-uniform error, we consider the worst case, \ie, the non-uniform error is $f-\frac{f}{B}$. 
Considering $f$ should be close to the threshold, the expected estimation error can be expressed as 
\begin{small}
    \begin{align}
        \mathbb{E}\left[\myErr_3\right] &= \mathbb{E}\left[(\frac{X}{B}+(f-\frac{f}{B}))^2\right] \nonumber \\
        &=\frac{c \sigma ^2}{B} + (f-\frac{f}{B})^2 \nonumber \\ 
        &=\frac{c \sigma ^2}{B} + (\frac{B-1}{B})^2 (B+1) c \sigma^2 \nonumber \\ 
        &=(\sigma^2\ln{|D|})\frac{B+(B-1)^2(B+1)}{B^2\ln{B}},
        \label{B setting}
    \end{align}    
\end{small}where $|D|$ is the domain size of attribute, $\epsilon$ is the privacy budget, 
$N$ is the user scale and $c$ is the number of user groups. 
Letting the derivative of \autoref{B setting} to 0, we get $B=0.6$ and $B=2.2$. 
Since the tree fanout $B$ is an integer greater than 1 and the value of \autoref{B setting} is smaller at $B=2$ than that at $B=3$, 
we select $B = 2$ for \myahead and compare the query error with $B = 4$ \cite{wang2019answering, cormode2019answering}. 
We empirically validate the effectiveness of this parameter setting in \autoref{Evaluation}. 

\begin{figure*}[t]
    \centering
    \includegraphics[width=\hsize]{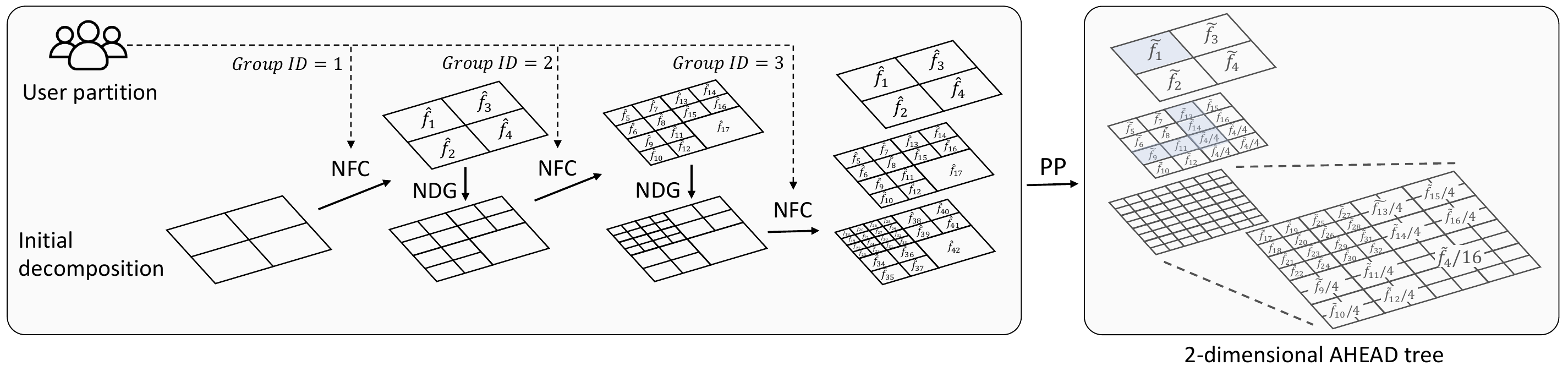}
    \vspace{-0.7cm}
    \caption{2-dim \myahead algorithm, 
    where the decomposition is implemented by decomposing both dimensions simultaneously.}
    \vspace{-0.4cm}
    \label{Multi-dimensional tree}
\end{figure*}

\subsection{Extension to Multi-dimensional Settings}
\label{Extension to Multi-dimensional Settings}
\mypara{2-dim Range Query}
Let us first look at the 2-dim scenario. 
Without loss of generality, we assume that all attributes have the same domain $D = \{1, 2, \cdots, d\}$, where $d$ is a power of the fanout $B$ (if not in real setting, we can simply add some dummy values to achieve it). 
The main difference between 1-dim and 2-dim \myahead is the decomposition process (lines 4, 9 of \autoref{Construct 1-dim prototype tree}). 
For 1-dim scenarios, \myahead hierarchically divides the entire domain, which is an interval $[1, |D|]$, into different sub-intervals with varying granularity. 
While for 2-dim scenarios, the entire domain becomes a square area $[1, |D|] \times [1, |D|]$. \myahead generates different granularity 2-dim grids to decompose the entire domain. 
The pseudo-code of 2-dim \myahead can be found in \autoref{The pseudo-code of 2-dim AHEAD}. 

An example of 2-dim \myahead is shown in \autoref{Multi-dimensional tree}. 
The domain size of the user private attribute is $8\times8$ and the tree fanout $B=4$. 
Similar to 1-dim scenarios, 2-dim \myahead also contains four steps to capture the users' data distribution. 

\begin{itemize}[leftmargin=*]
\item \mypara{Step 1: User Partition (UP)}
The users randomly choose the group number in range $[1, 2, 3, \cdots, c]$, where $c = \log_B|D|^2$. 

\item \mypara{Step 2: Noisy Frequency Construction (NFC)}
Then, the aggregator divides the entire domain into $B$ equal-size square areas and sends the initial decomposition of the domain to the first group of users. The users project their private data into the initial decomposition and reports their data through \oue. The server uses the aggregation algorithm to obtain the estimated frequency distribution, which represents the ratio of users falling within each sub-domain.

\item \mypara{Step 3: New Decomposition Generation (NDG)}
After that, the aggregator compares each frequency value with a threshold $\theta$ and decides whether to divide the corresponding sub-domain further. 
Repeating the NFC and the NDG processes, \myahead recursively decomposes the domain and constructs the \myahead prototype tree.

\item \mypara{Step 4: Post-processing (PP)}
Finally, \myahead conducts the non-negativity process within each layer and the weighted averaging process between two adjacent layers to further reduce the estimated error. Based on uniform distribution assumption, \myahead obtains a complete tree to answer queries. 
\end{itemize}
As shown in the right part of \autoref{Multi-dimensional tree}, 
In order to reduce the query error, which increases with the number of nodes used, \myahead prefers to use coarse-grained nodes to answer a query. 
For example, \myahead answers a 2-dim query $[0,5]\times[0,5]$. 
\myahead searches from the top to the bottom layer of the tree and calculates the sum of the fully covered sub-domains' frequencies. 
The query completely covers the area of $[0,3]\times[0,3]$ of the top layer and five areas $[0,1]\times[4,5]$, $[2,3]\times[4,5]$, $[4,5]\times[4,5]$, $[4,5]\times[0,1]$, $[4,5]\times[2,3]$ of the middle layer (these areas are highlighted in blue), thus the query answer is $(\tilde{f}_1 + \tilde{f}_9 + \tilde{f}_{11} + \tilde{f}_{13} + \tilde{f}_{14} + \tilde{f}_{4}/4)$.

For the 2-dim range query, we set $B = 2^2$, \ie, the $B$ of each dimension is $2$. 
As a comparison, we also provide the results of $B = 4^2$ in the experiment. Since the derivation of $\theta$ does not involve dimension changes, we still select $\theta$ according to \autoref{theta setting1}. 

\mypara{High-dimensional Range Query}
\myahead can be extended to higher dimensions in two ways. 
\begin{itemize}[leftmargin=*]
\item Direct Estimation (\de). Based on a tree with fanout $B = 2^m$, \myahead decomposes $m$ dimensions simultaneously. 
For instance, \myahead treats the 3-dim domain as a cube, and then aggregates the frequencies of sub-cubes with different granularities to answer queries. 
With the threshold setting of \autoref{theta setting1}, \myahead can well control the overall estimation errors of sub-domains. 
However, the number of leaf nodes increases exponentially with dimension, which makes the query answer process extremely time-consuming on high-dimensional datasets. 

\item Leveraging Low-dimensional Estimation (\lle). 
To solve the sub-domain explosion caused by the rise of data dimension, 
\lle combines the attributes in pairs, and then constructs a 2-dim \myahead tree for each attribute pair. 
When answering an $m$-dim query, \lle constructs a query set with the associated $2^m$ queries (recall Algorithm \ref{Estimating Answer of m-dimensional Range Query} in \autoref{Extending AHEAD to high-dimensional range query}). 
Then, taking the 2-dim frequency as constraints, 
\lle estimates the frequency values of all the $2^m$ queries by the maximum entropy optimization. For self-containment, we include its description and pseudo-code in \autoref{Extending AHEAD to high-dimensional range query}. 
Previous methods also used the maximum entropy optimization to effectively extend the low-dimensional mechanism to high-dimensional scenarios. 
\PriView \cite{qardaji2014priview} proposes to construct low-dimensional views, \ie, 2, 3-dim marginal tables, and apply maximum entropy optimization to reconstruct higher-way marginals from views. 
In this way, \PriView constructs marginal tables for $m$-dim data with high data utility while satisfying DP. 
\mycalm \cite{zhang2018calm} migrates the idea of \PriView to LDP, which also achieves high accuracy for the problem of marginal release under LDP. 
\myHDG \cite{yang2020answering} groups all attributes in pairs and estimates the frequency distribution of user data on each attribute pair (2-dim grid). 
Then, \myHDG obtains the answer of a higher dimensional range query from the answers of the associated 2-dim range queries. 
\end{itemize}

The implementation of the \de method is relatively simple, and the number of user partitions will not increase with the increase of the dimension. 
However, the \myahead tree with \de might be very large in high-dimensional scenarios, 
making the tree construction and query answering process time-consuming. 
Compared to \de, the \lle method combines the attributes in pairs, and then constructs a 2-dim \myahead tree for each attribute pair, which makes the scale of each tree not too large. 
We empirically show the performance of these two strategies and discuss how to choose between them in \autoref{Effectiveness of ahead for high-dimensional Range Query}. 

\subsection{Discussion}

\myahead is similar to \PrivTree \cite{zhang2016privtree}, \ie, a general approach for hierarchical decomposition on private data, where \PrivTree also 1) generates tree $T$ by recursively splitting a root node $n_0$ whose sub-domain covers the entire data domain $D$, and 2) decides whether a node $n$ should be decomposed based on a noisy frequency of $n$. 
However, \myahead differs from \PrivTree in several important aspects. 
It is worth noting that these differences are mainly due to the fact that these two methods work under different privacy requirements. 
That is, \PrivTree works in a centralized setting of differential privacy, while \myahead works in a local setting. 
\begin{itemize}[leftmargin=*]
\item \PrivTree can directly access all the information in the server, where \PrivTree operates on the dataset, adds noise, and then derives the answers. \myahead only has access to the noisy data uploaded by users, and then aggregates the reports to answer queries. 
\item In \PrivTree, each frequency is estimated using the information of all users, with a split privacy budget. While in the local setting, partitioning users into groups and estimating the frequency with an entire privacy budget can obtain a higher data utility \cite{zhang2018calm, cormode2019answering, wang2019answering}. Therefore, in \myahead, each frequency is estimated by only a subset of users, with an entire privacy budget.  
\item Based on the above two differences, besides \emph{noise errors} in \PrivTree, \myahead further considers the \emph{sampling} and \emph{non-uniform errors}. 
\item In \PrivTree, the tree fanout $B$ is not considered in error analysis. In comparison, \myahead derives the $B$ setting through the analysis of \emph{noise errors} and \emph{non-uniform errors}. 
\item For high-dimensional scenes, \PrivTree adopts direct estimation (\de) to extend low-dimensional strategies, while \myahead leverages a more efficient way, \ie, leveraging low-dimensional estimation (\lle). 
\end{itemize}

\section{Evaluation}
\label{Evaluation}
To validate the effectiveness of \myahead, we evaluate its performance on multiple real-world datasets and compare \myahead with the state-of-the-art methods such as \myhio \cite{wang2019answering} (1-dim query), \mydht \cite{cormode2019answering} (1-dim query), \mycalm \cite{zhang2018calm} (1, 2-dim query) and \myHDG \cite{yang2020answering} (2-dim and high-dimensional query) in balancing utility and privacy. 
We also provide a detailed complexity analysis of the algorithms used in our evaluation in \autoref{Complexity Analysis}. 
\subsection{Experimental Setup}
\begin{table}[!t]
    \centering
    \caption{Summary of datasets.}
    \vspace{-0.3cm}
    \setlength{\tabcolsep}{0.8mm}{
        \resizebox{0.40\textwidth}{!}{
          \begin{tabular}{c c c c c}
            \toprule
            \bf{Dataset}  & \bf{Distribution} & \bf{Scale}  & \bf{Field} & \bf{Type}    \\  
            \midrule
            Salaries      &      --         & 148,654       &   employee salary & real \\   
            \rowcolor{mygray}
            BlackFriday   &      --         & 537,577       &   shopping    & real \\
            Loan          &      --         & 2,260,668    &   online loan  & real \\             
            \rowcolor{mygray} 
            Financial    &      --         & 6,362,620    & fraud detection   &  synthetic\\
            Cauchy       &      Cauchy     &  --          & --              &  synthetic               \\    
            \rowcolor{mygray} 
            Zipf         &      Zipf (power-law)     &  --          &  --   &  synthetic    \\
            Gaussian     &      Gaussian  & --   &  --      &  synthetic \\
            \rowcolor{mygray} 
            Laplacian      &    Laplacian   & --  &  --    &  synthetic\\  
          \bottomrule
          \end{tabular}}}
          \vspace{-0.4cm}
          \label{realDatasets}
\end{table}
\mypara{Environment} 
All our evaluations are conducted on a PC with Intel Xeon Platinum 8269@2.5GHz and 32GB memory. 

\mypara{Datasets}
In our experiments, we use 3 real-world datasets and 5 synthetic datasets to evaluate the performance of \myahead, 
and \autoref{realDatasets} provides an overview of all datasets. 
The detailed information about the four datasets, \ie, Salaries, BlackFriday, Loan and Financial, is demonstrated in \autoref{Dataset Description}. 

We generate 1-dim datasets by sampling data from the \Cauchy ($x_0 = 0$, $\gamma = 1$), \Zipf ($\alpha = 1.1$), \Gaussian ($\mu = 0$, $\sigma^2 = 1$) and  \Laplacian ($\mu = 0$, $\lambda^2 = 1/2$) distribution respectively, as prior works did \cite{ye2019privkv, cormode2019answering, yang2020answering}. 
The multi-dimensional datasets are synthesized from multivariate \Gaussian and \Laplacian distribution with mean $0$, standard deviation $1$ \cite{yang2020answering}. 

\mypara{Metrics}
To quantify the performance of \myahead, 
we use the MSE (mean square error) widely used in literature \cite{ye2019privkv, cormode2019answering, wang2019consistent} to measure the deviation between estimated and true values. 
For each experimental setting, we compute the MSE of 200 query results, and then compute the mean and std among 20 repetitions. 
In addition, we also provide a 95\% confidence interval to reflect the deviations between MSEs. 

\mypara{Competitors}
For a fair comparison, the \myhio \cite{wang2019answering}, \mydht \cite{cormode2019answering}, \mycalm \cite{zhang2018calm} and \myHDG \cite{yang2020answering} methods are applied with the same parameter settings as in the original papers. 
We also plot the MSE of the Uniform method (denoted as \myuni) in 1, 2-dim scenes, which always obtains the query answer from a uniform distribution. Clearly, if the performance of one method is worse than \myuni, the query answer from that method is meaningless. For high-dimensional range query, we take \myHDG as the baseline method. 

\subsection{Evaluation for 1-dim Range Query}
\label{Evaluation for 1-dim Range Query}
\begin{figure*}[ht]
    \centering
    \subfigure[Loan, $|D|=256$, vary $\epsilon$]{
    \includegraphics[width=0.23\hsize]{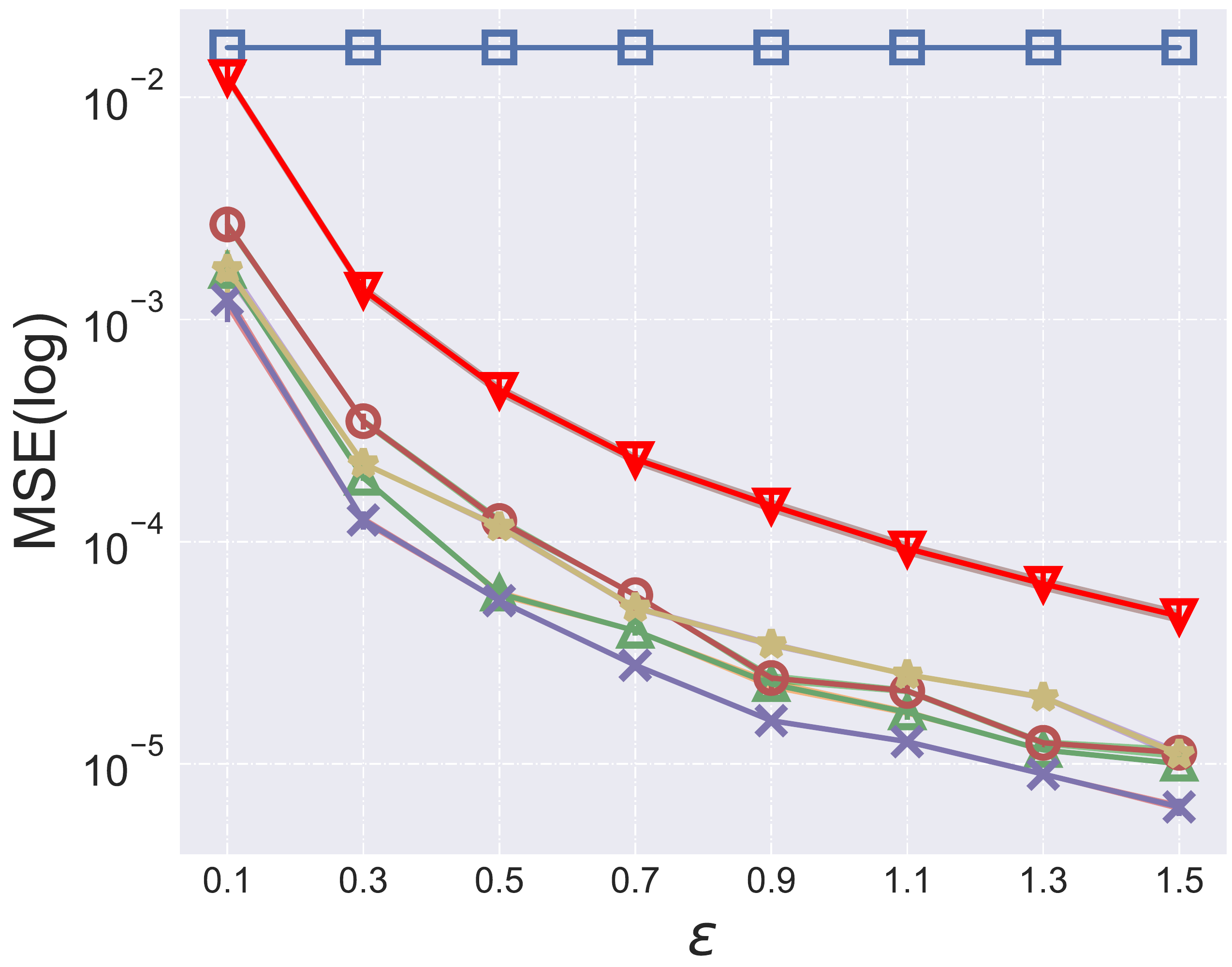}
    }
    \subfigure[Financial, $|D|=512$,  vary $\epsilon$]{
    \includegraphics[width=0.23\hsize]{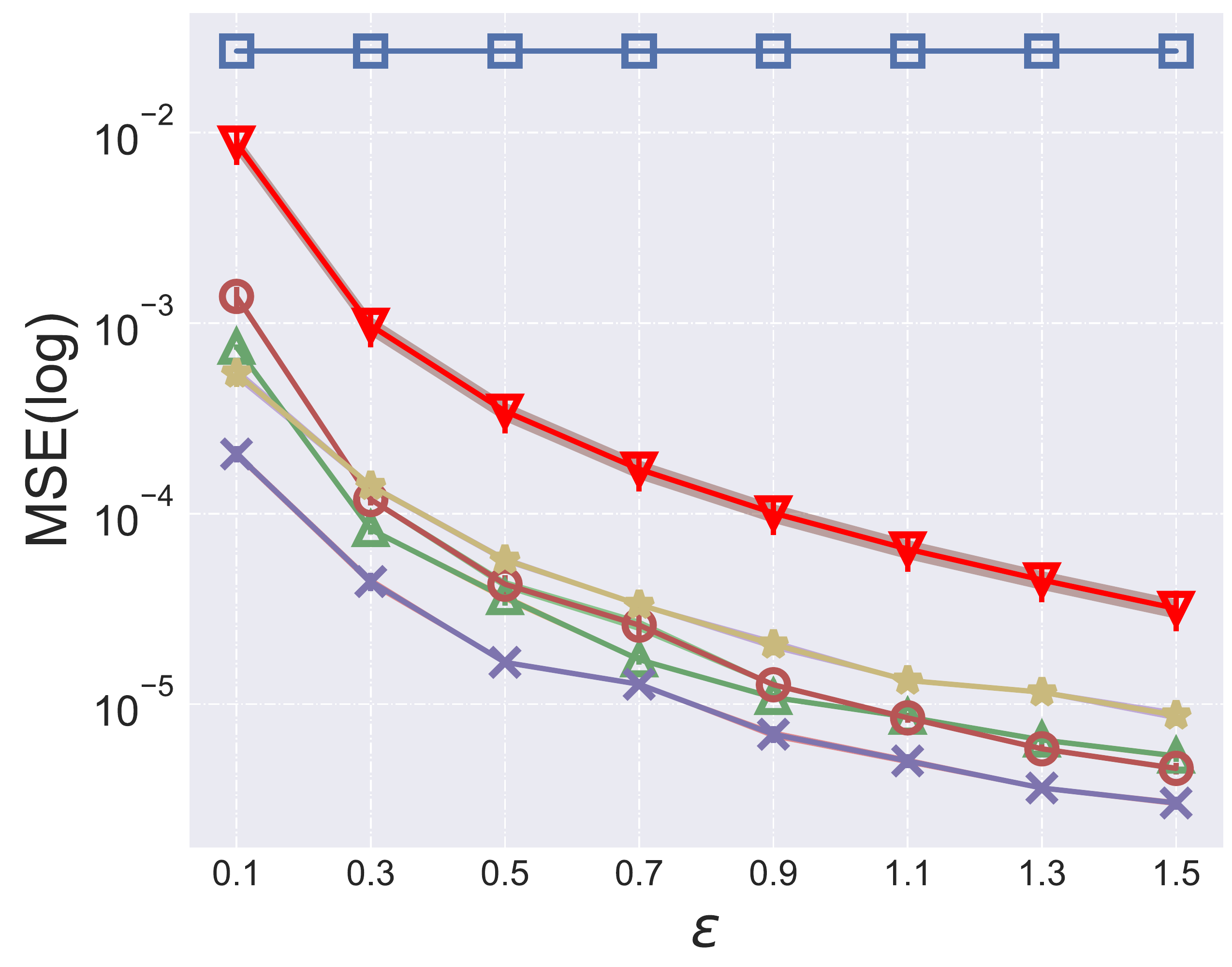}
    }
    \subfigure[BlackFriday, $|D|=1028$,  vary $\epsilon$]{
    \includegraphics[width=0.23\hsize]{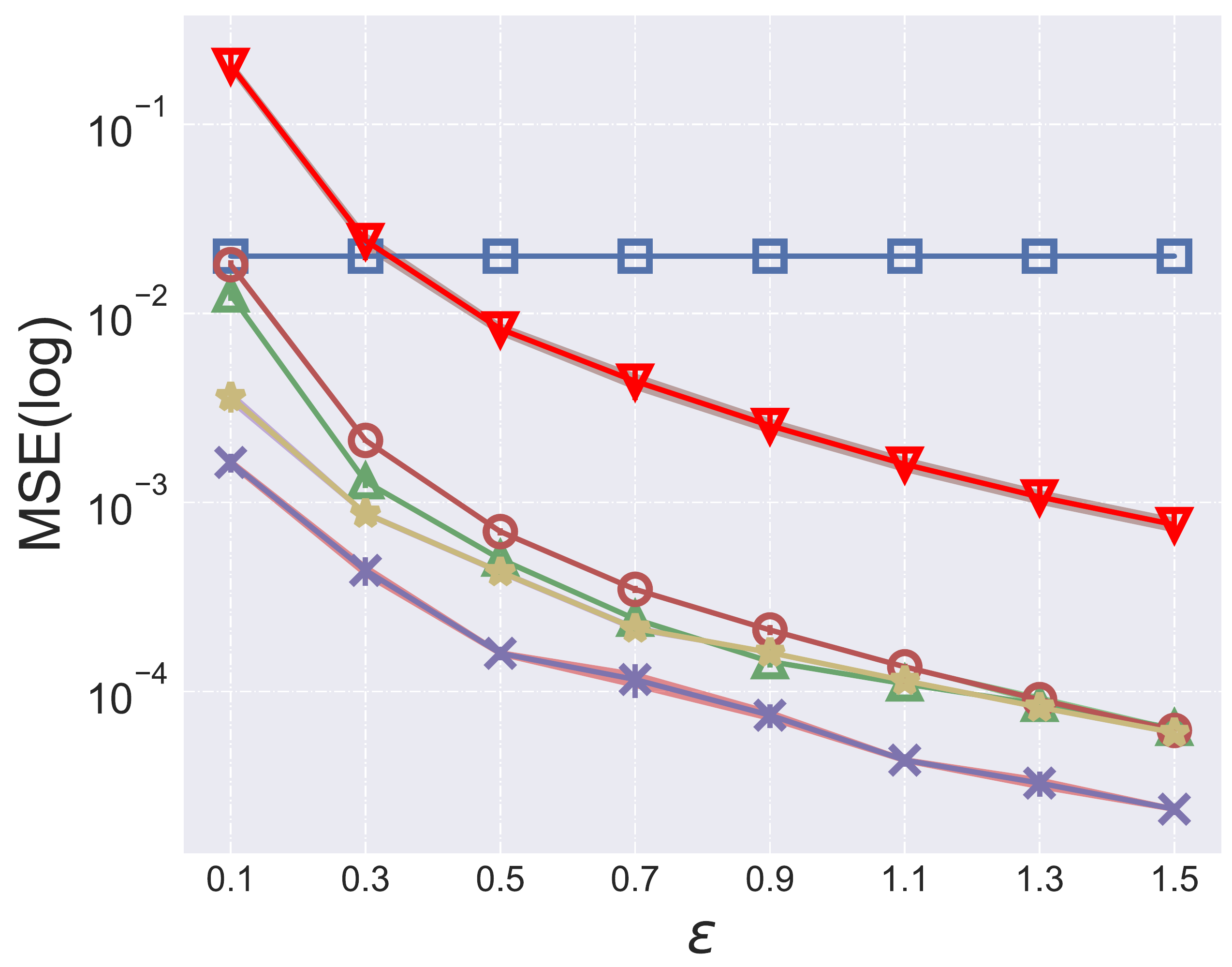}
    }
    \subfigure[Salaries, $|D|=2048$,  vary $\epsilon$]{
    \label{overall compare Salaries}
    \includegraphics[width=0.23\hsize]{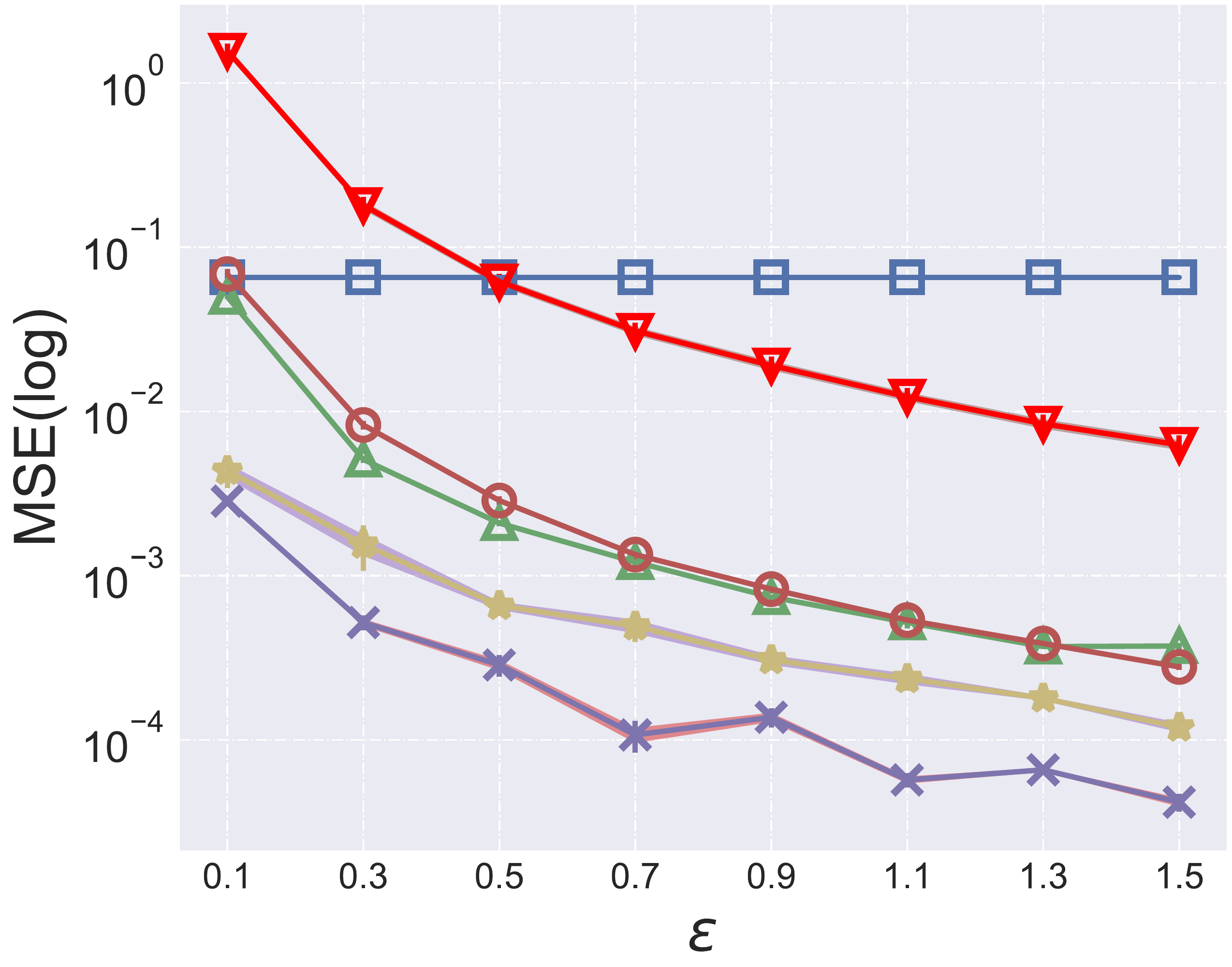}
    }
    \\[-1ex]
    \subfigure{
    \includegraphics[width=0.65\hsize]{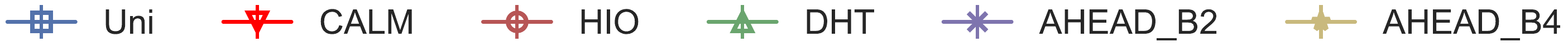} 
    }
    \vspace{-0.6cm}
    \caption{The MSE of different methods with privacy budget varying from 0.1 to 1.5. The results are shown in log scale. }
    \vspace{-0.4cm}
    \label{overall compare}
\end{figure*}

We evaluate the effectiveness of \myahead under various values of privacy budgets. 
Specifically, we consider privacy budget $\epsilon$ varying from 0.1 (representing high privacy protection) to 1.5 (representing low privacy protection) \cite{zhang2018calm, cormode2019answering}. 
\autoref{overall compare} illustrates the MSE comparison between \myahead and existing algorithms \cite{cormode2019answering, wang2019answering, zhang2018calm} for 1-dim queries. 
For each plot, we vary the privacy budget $\epsilon$ on the X-axis and show the corresponding MSE on the Y-axis. 
Each histogram in the plots shows the average MSE of 20 repeated experiments with the error bar representing the standard deviation. 

From \autoref{overall compare}, we have the following observations.  
1) The MSE of \myahead is smaller than its counterparts throughout the experiment datasets. 
Recalling \autoref{theta setting 2} in \autoref{Selection parameters}, 
for the nodes with frequency values less than the threshold, 
\myahead with the adaptive strategy has a smaller overall estimation error than the counterparts that employ the baseline strategy. 
2) \myahead obtains different performance over various datasets. 
For the Salaries dataset, 
\myahead achieves significant advantages over previous methods across the entire range of privacy budget, 
where the MSE of \myahead is almost one order of magnitude smaller than state-of-the-art methods. 
For the Loan dataset, \myahead is slightly better than \mydht. 
We speculate that this is mainly because these two datasets have quite different distributions, such as more sub-domains with small frequency values (less than the threshold) on the Salaries dataset compared to the Loan dataset, which causes \myahead to behave differently on these two datasets. 

In addition, 
the performance of \myahead varies under different datasets, 
which have various user scales, attribute domain sizes and data distributions. 
Therefore, we will conduct a comprehensive experiment in \autoref{comprehensive experiment} to further analyze the performance of \myahead under different experimental parameters and conclude important observations for its practical adoption. 

\subsection{Evaluation for 2-dim Range Query}

\begin{figure*}[h!]
    \centering
    \subfigure[2-dim \Laplacian, $|D|=256^2$, vary $\epsilon$]{\label{fig:Two_Laplace08-Set_10_7-Domain_8_8}\includegraphics[width=0.24\hsize]{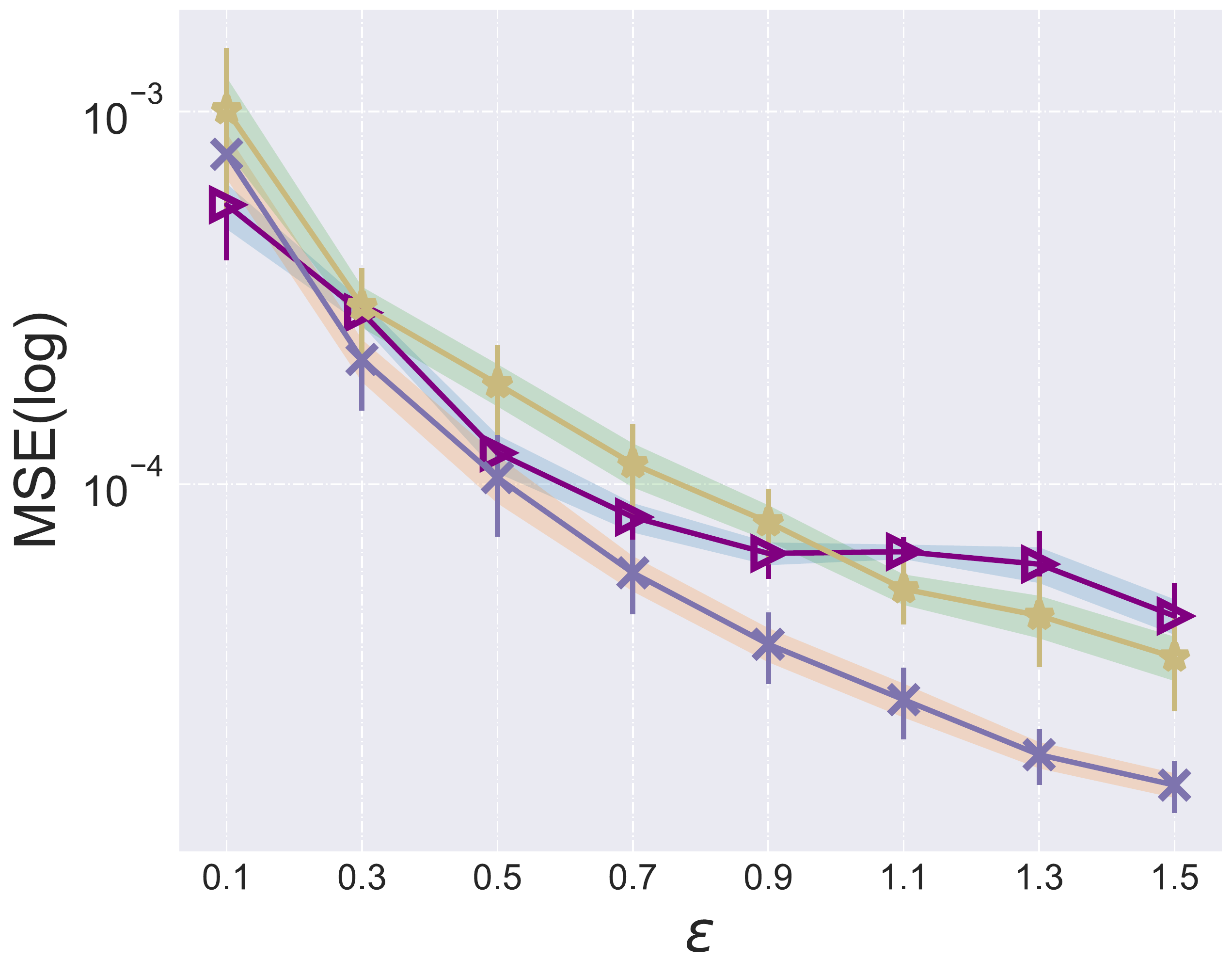}}
    \subfigure[2-dim \Laplacian, $|D|=1024^2$, vary $\epsilon$]{\label{fig:Two_Laplace08-Set_10_7-Domain_10_10}\includegraphics[width=0.24\hsize]{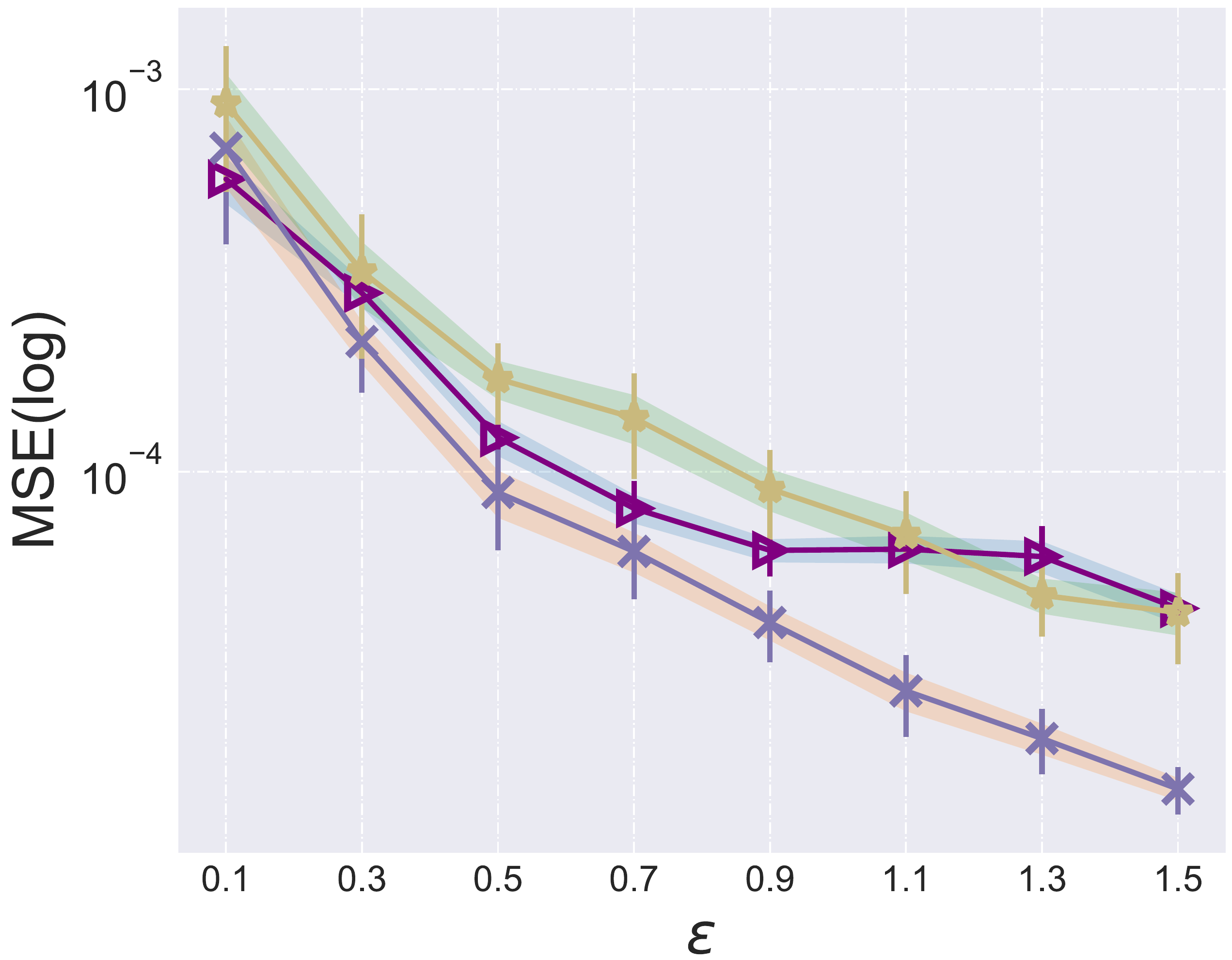}}
    \subfigure[2-dim \Gaussian, $|D|=256^2$, vary $\epsilon$]{\label{fig:Two_Normal08-Set_10_7-Domain_8_8}\includegraphics[width=0.24\hsize]{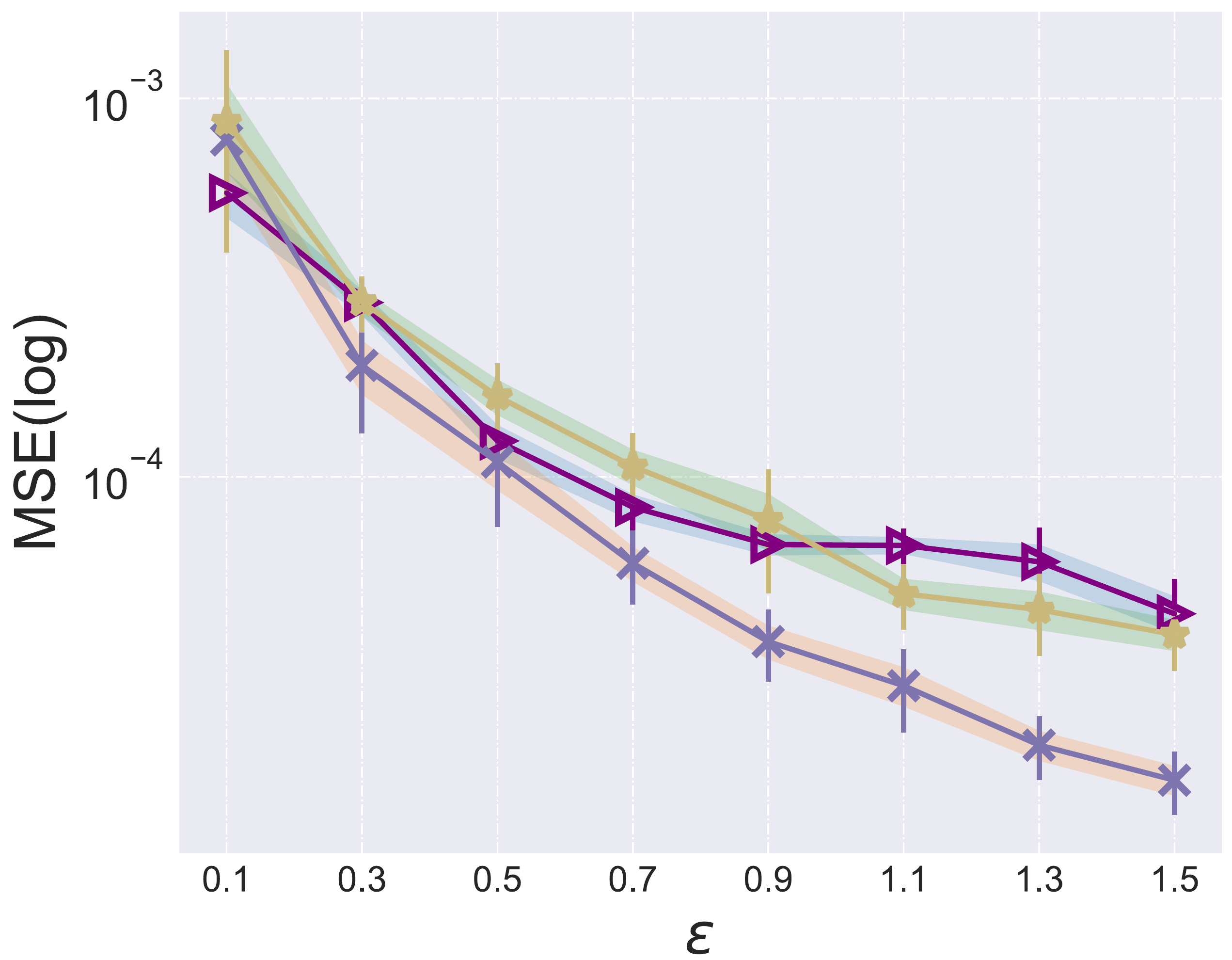}}
    \subfigure[2-dim \Gaussian, $|D|=1024^2$, vary $\epsilon$]{\label{fig:Two_Normal08-Set_10_7-Domain_10_10}\includegraphics[width=0.24\hsize]{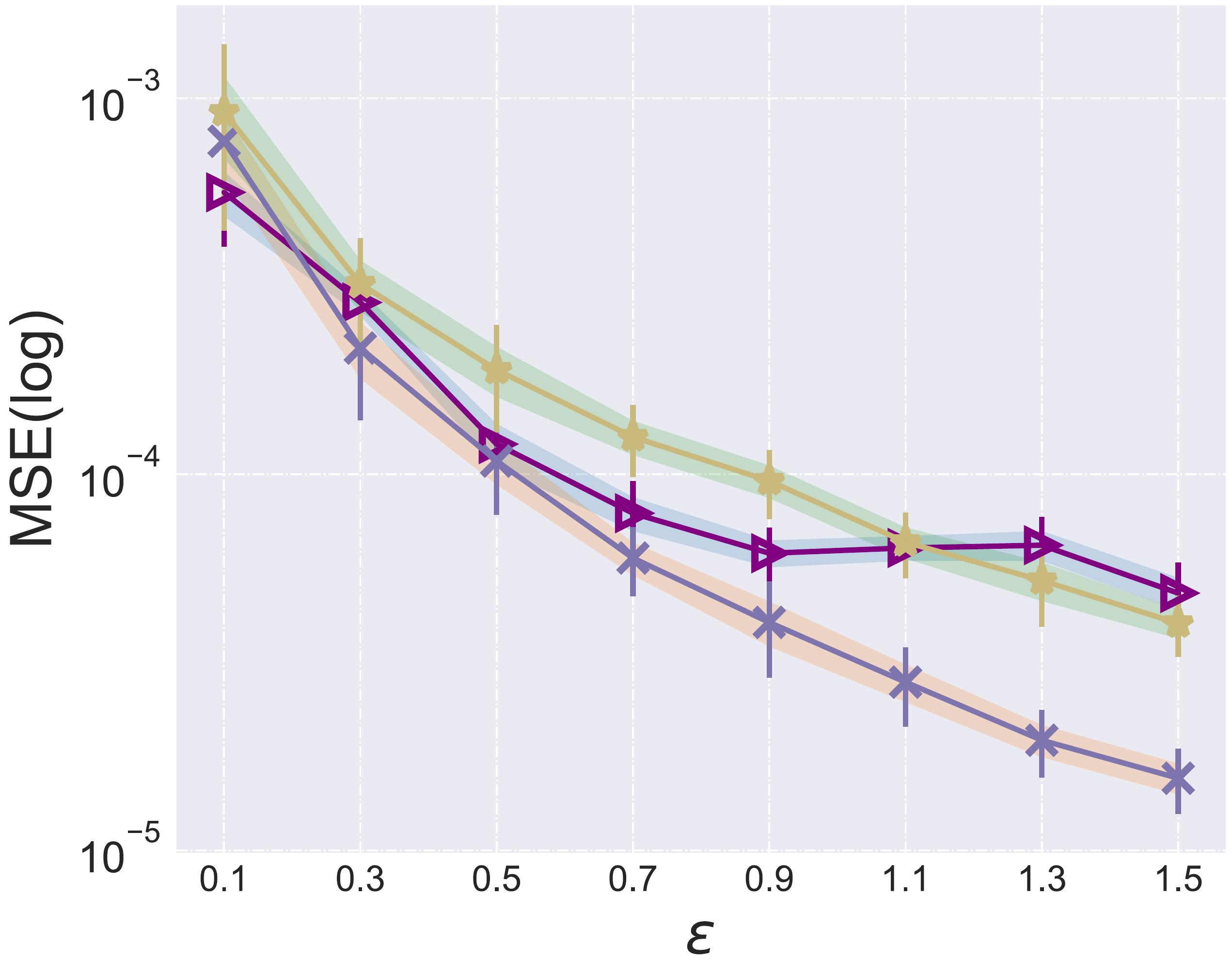}}\\[-1.8ex]
    \subfigure{\includegraphics[width=0.4\textwidth]{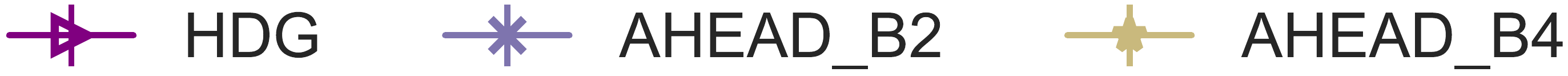}} 
    \vspace{-0.6cm}
    \caption{Comparison of different methods on 2-dim \Laplacian and \Gaussian datasets under various privacy budgets. We only plot the methods that are scalable in each setting. \myHDG is a baseline method. The results are shown in log scale.}
    \vspace{-0.5cm}
    \label{2-dim overall compare}
\end{figure*}

\begin{figure*}[h!]
    \centering
    \subfigure[2-dim \Laplacian, $|D|=256^2$, vary $r$]{\label{fig:Two_Laplace08-Set_10_7-Domain_8_8}\includegraphics[width=0.24\hsize]{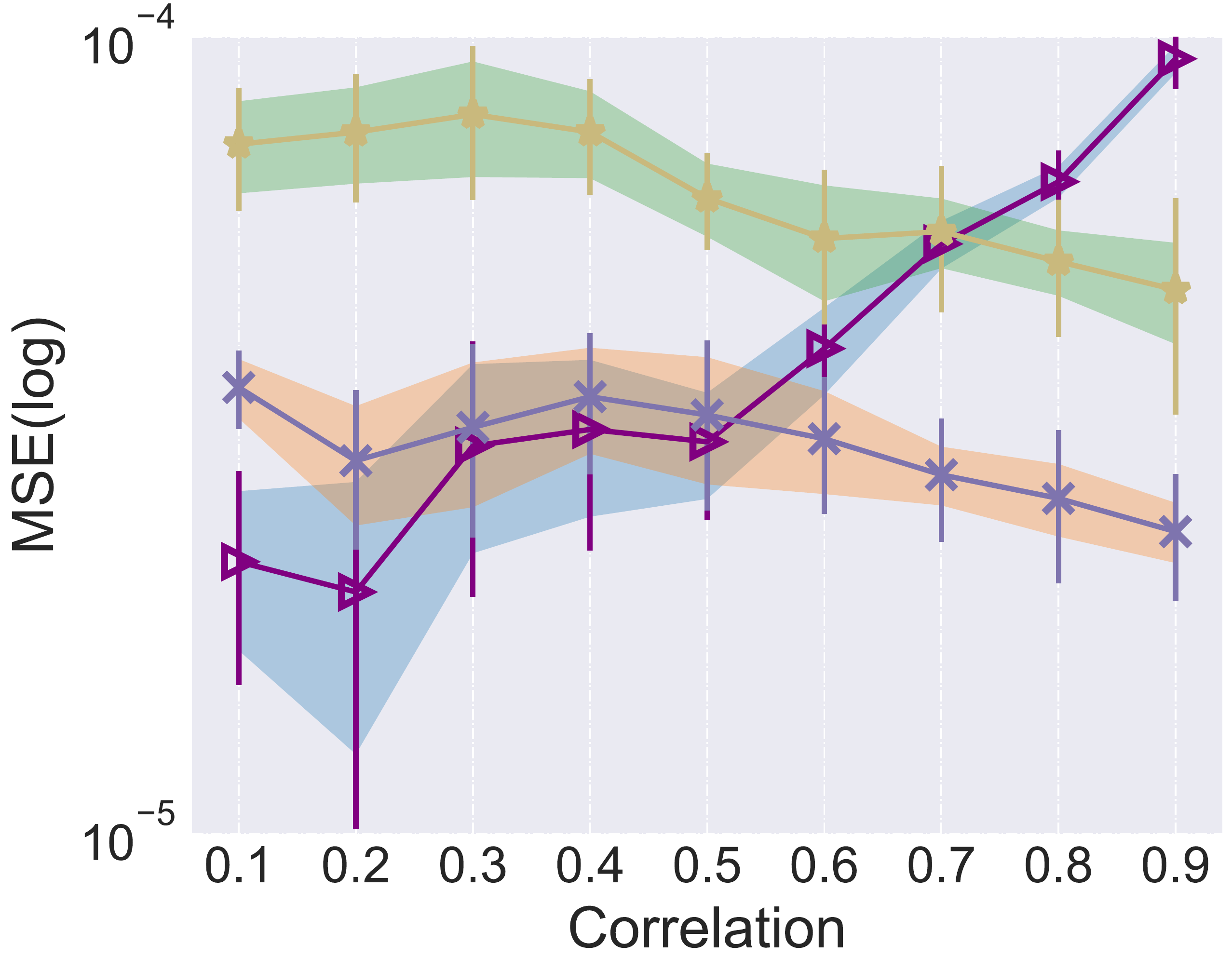}}
    \subfigure[2-dim \Laplacian, $|D|=1024^2$, vary $r$]{\label{fig:Two_Laplace09-Set_10_7-Domain_8_8}\includegraphics[width=0.24\hsize]{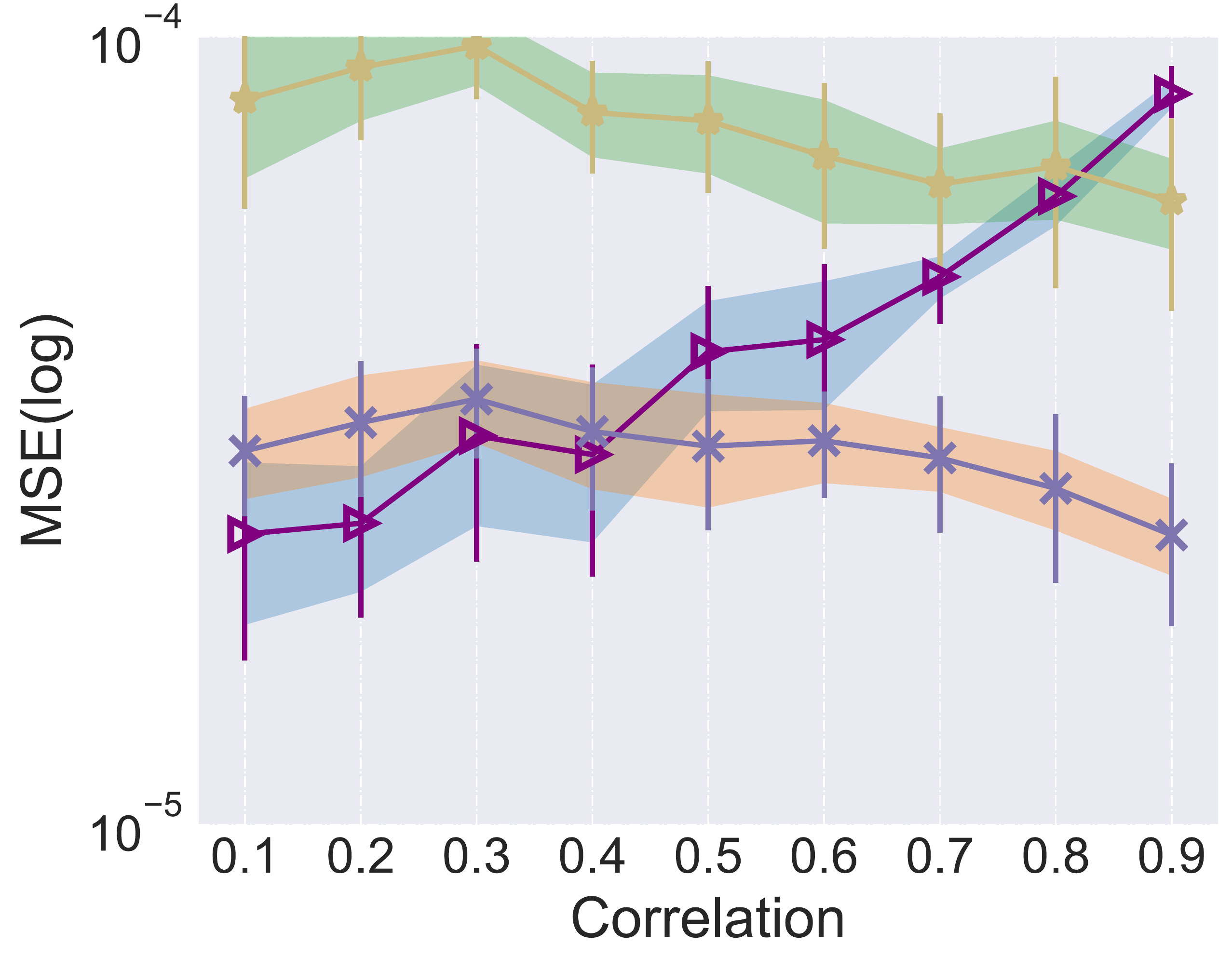}}
    \subfigure[2-dim \Gaussian, $|D|=256^2$, vary $r$]{\label{fig:Two_Laplace08-Set_10_7-Domain_10_10}\includegraphics[width=0.24\hsize]{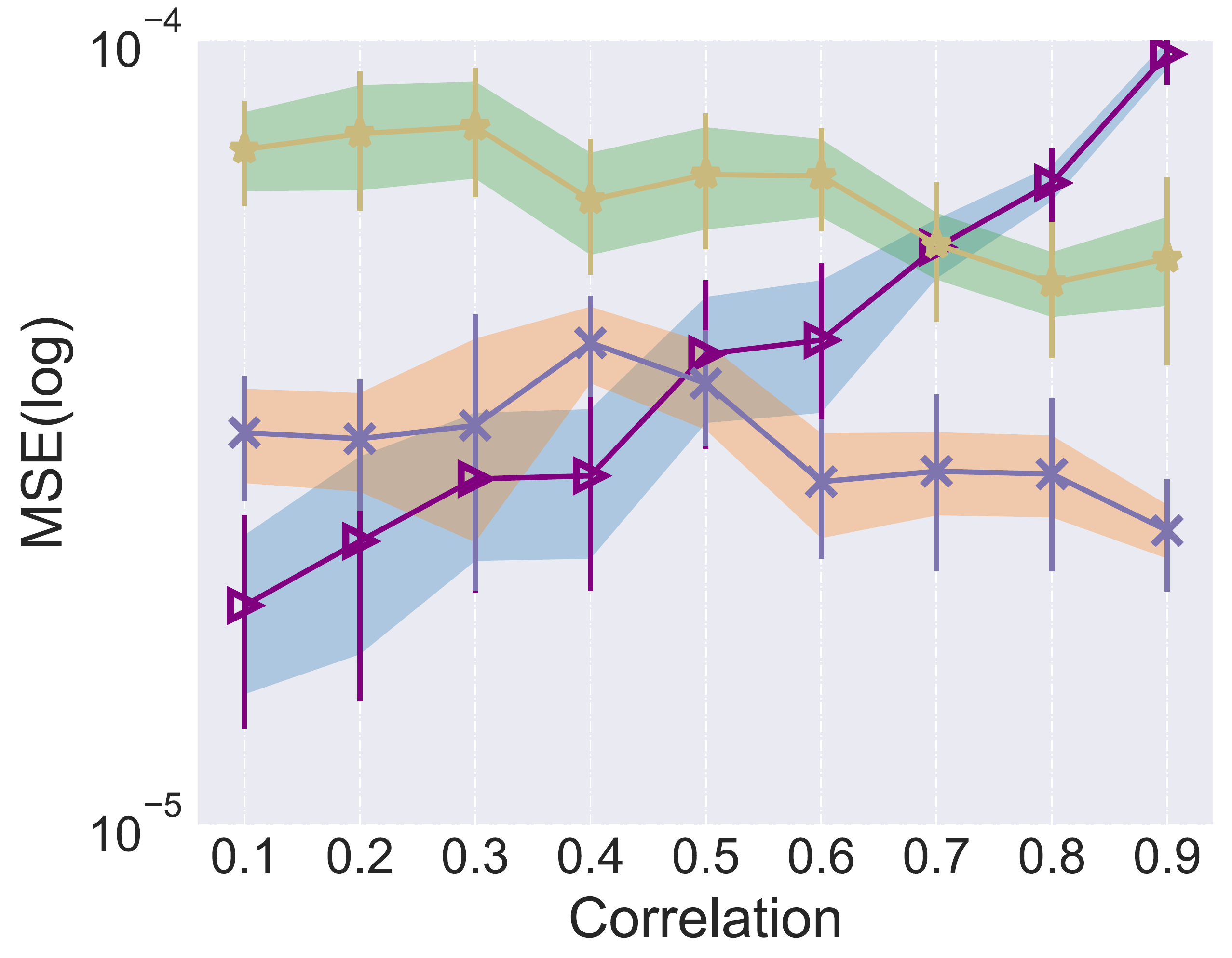}}
    \subfigure[2-dim \Gaussian, $|D|=1024^2$, vary $r$]{\label{fig:Two_Laplace09-Set_10_7-Domain_10_10}\includegraphics[width=0.24\hsize]{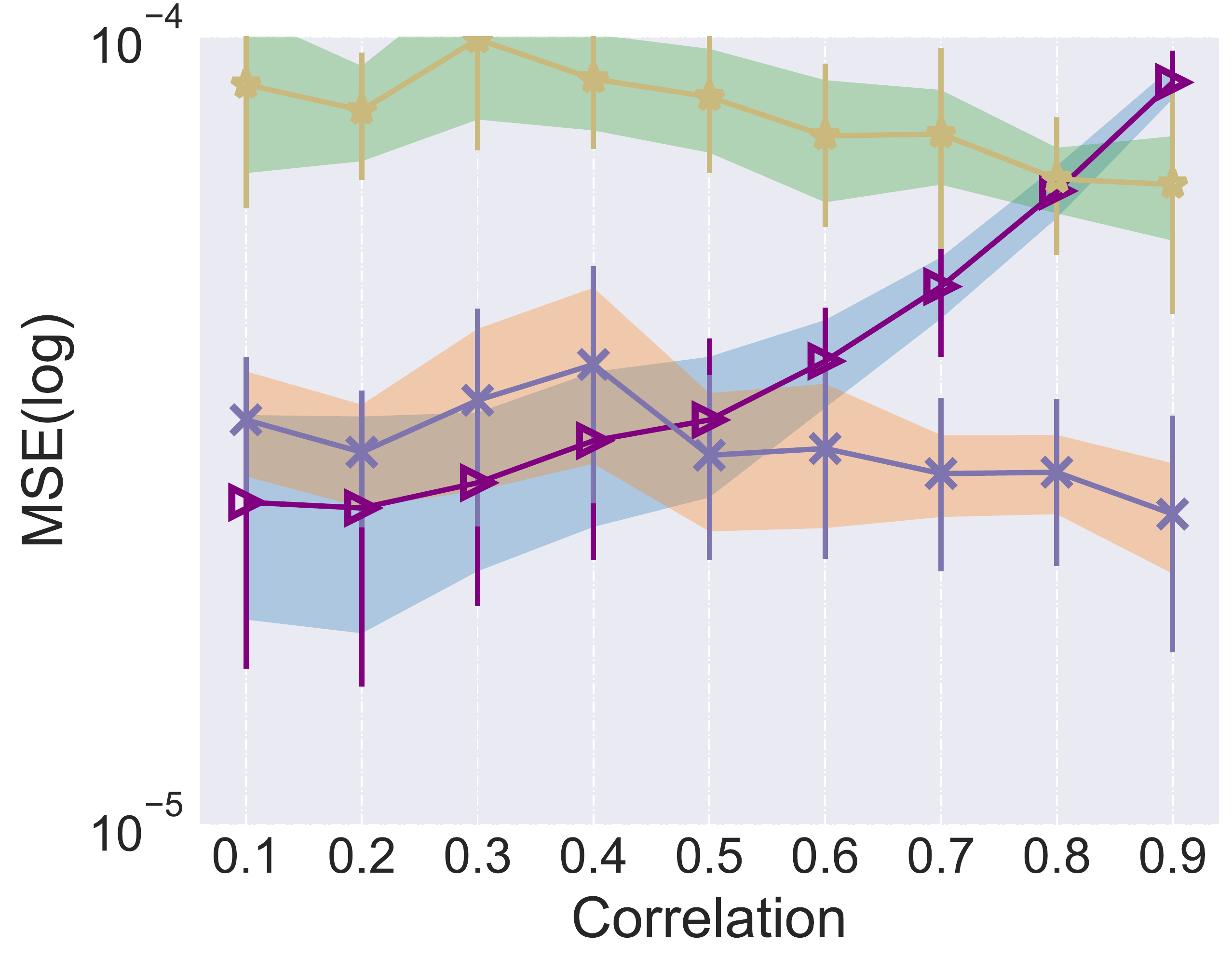}}\\[-1.5ex]
    \subfigure{\includegraphics[width=0.4\textwidth]{./figures_others/comparison_2D_legend.pdf.pdf}}  
    \vspace{-0.6cm}
    \caption{Comparison of different methods on 2-dim \Laplacian and \Gaussian datasets under various attribute correlation. The results are shown in log scale.}
    \vspace{-0.5cm}
    \label{2-dim correlation}
\end{figure*}

\mypara{MSE under Various Privacy Budget}
We evaluate the impact of varying privacy budget on the algorithms under 2-dim queries. 
Here, we focus on the synthetic datasets since they can reflect the algorithms' performance on the standard distribution and facilitate the adjustment of the dataset parameters. 
Each dataset contains $10^7$ records, sampling from 2-dim \Laplacian and \Gaussian distributions respectively, with two domain sizes as $256 \times 256$ and $1024 \times 1024$. 

\autoref{2-dim overall compare} shows the results over 
2-dim \Laplacian and \Gaussian distributions respectively. 
We adopt correlation coefficient $r = 0.8$ the same with \cite{yang2020answering} between the two attributes. 
The MSEs of \myuni and \mycalm are out of scale, thus their results are omitted here. 

Based on the results, we have the following observations. 
1) \myahead outperforms \myHDG throughout the privacy budget settings. 
\myHDG leverages coarse 2-dim grids to bucketize the 2-dim domain, and captures the correlation information between two attributes. 
For a 2-dim $1024 \times 1024$ domain, the finest granularity of \myHDG is $8 \times 8$ in experiments. 
Using such coarse-grained grid results in the loss of some fine-grained data correlation. 
\myahead establishes multiple granularity decompositions for sub-domains with different frequency values. 
For a sub-domain with high frequency value, which has a great impact on query answers, 
\myahead uses fine-grained decomposition to estimate the frequency distribution within the sub-domain. 
Thus, the data correlation will be captured more accurately. 
2) \myahead is robust to the changes in domain size. 
On one hand, 
a larger domain size means that users need to be divided into more groups. 
For example, \myahead partitions users into 8 groups for domain size $256 \times 256$, and 10 groups for domain size $1024 \times 1024$. 
There are fewer users in each group for a larger domain size, which will increase the added noise in the frequency values. 
On the other hand, the sub-domains whose frequencies are smaller than the threshold are not further divided from the layer where they first appear. 
These sub-domains will be estimated multiple times by user reports from different groups. 
After the weighted average process in post-processing, 
the noise errors of these sub-domains will become $\frac{1}{\beta}$ of the original noise errors, 
where $\beta$ is the number of the estimation times. 
Comparing two adjacent subplots in \autoref{2-dim overall compare}, the impact of the threshold setting is greater than that of domain size changes. Thus, the MSE of \myahead almost does not vary across different domain sizes, 
and we select diverse domain sizes to further verify this fact in \autoref{Impact of Domain Size}. 

\mypara{MSE under Various Attribute Correlation}
Then, we evaluate the impact of different attribute correlations on query errors as shown in \autoref{2-dim correlation}. 
For each subplot, we vary the correlation coefficient $r$ from 0.1 (weakly correlated) to 0.9 (strongly correlated) with a fixed privacy budget $\epsilon = 1.1$. 
From the results, we have following findings.  
1) The MSE of \myahead almost does not change with different correlations. 
\myahead decomposes both dimensions simultaneously, thus better protecting the correlation of the data. 
2) The data utility of \myHDG changes significantly with the correlation of attributes, and becomes worse with a stronger correlation. 
\myHDG combines finer-grained 1-dim grid to estimate the frequency distribution. 
If the correlation between two attributes is not very strong, 
the fine-grained 1-dim grid can complement the deficiencies of the coarse-grained 2-dim grid. 
Otherwise, the supplementary 1-dim information may be counterproductive. 
It is interesting to find that correlation $r = 0.5$ seems to be the intersection of \myHDG and \myahead, 
which may guide the aggregator to select the better algorithm based on the correlation of attributes.
For high-dimensional scenarios, we also evaluate the impact of different attribute correlations on query errors as shown in \autoref{Impact of Attribute Correlation}, where \myahead reacts similarly as 2-dim scenes. 

\subsection{Evaluation for High-dimensional Range Query}
\label{Effectiveness of ahead for high-dimensional Range Query}

\begin{figure*}[h!]
    \centering
    \subfigure[3-dim \Laplacian, $N=10^6$, vary $\epsilon$]{\label{fig:Three_Laplace08-Set_10_6-Domain_6_6_6}\includegraphics[width=0.24\hsize]{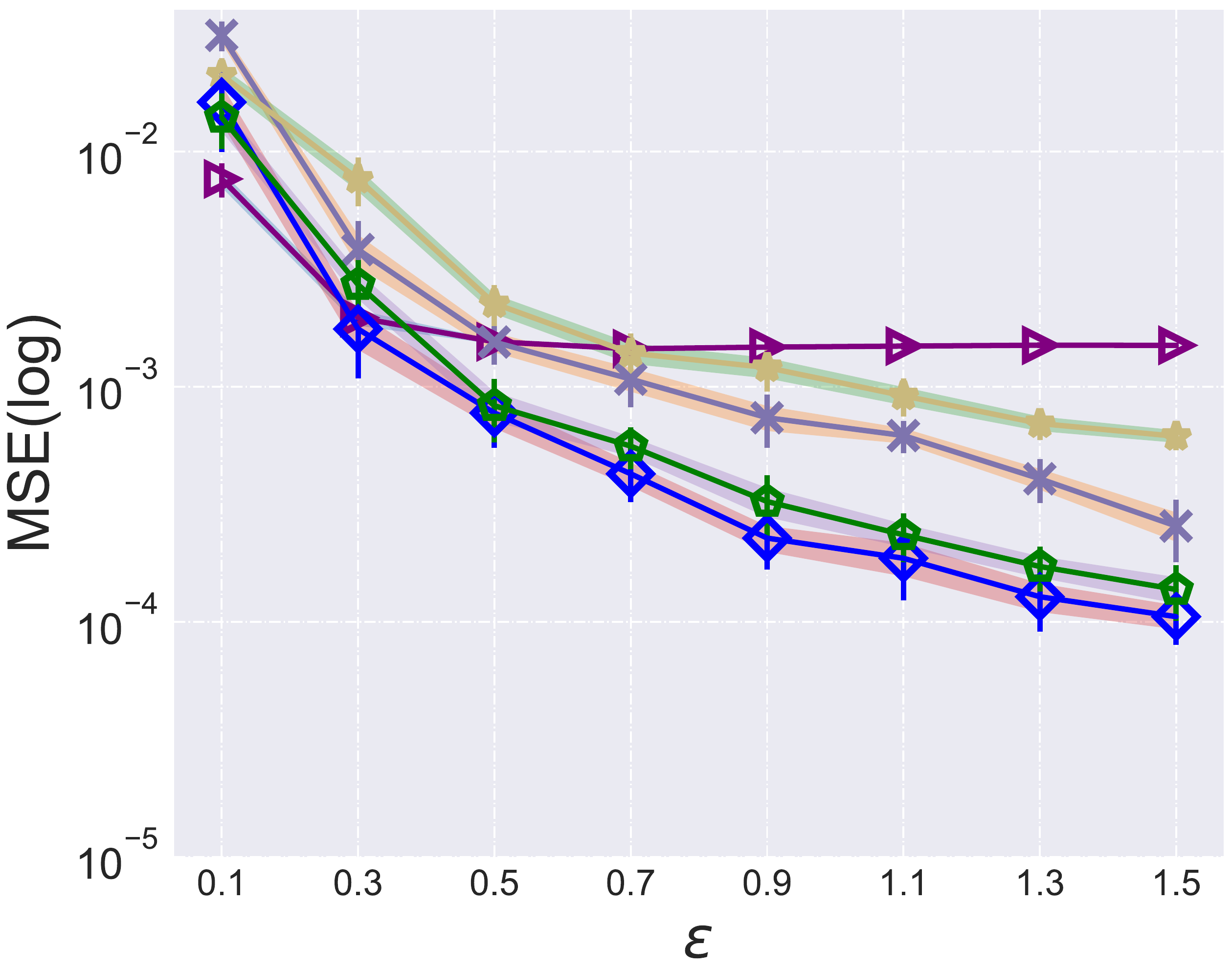}}
    \subfigure[3-dim \Laplacian, $N=10^7$, vary $\epsilon$]{\label{fig:Three_Laplace08-Set_10_7-Domain_6_6_6}\includegraphics[width=0.24\hsize]{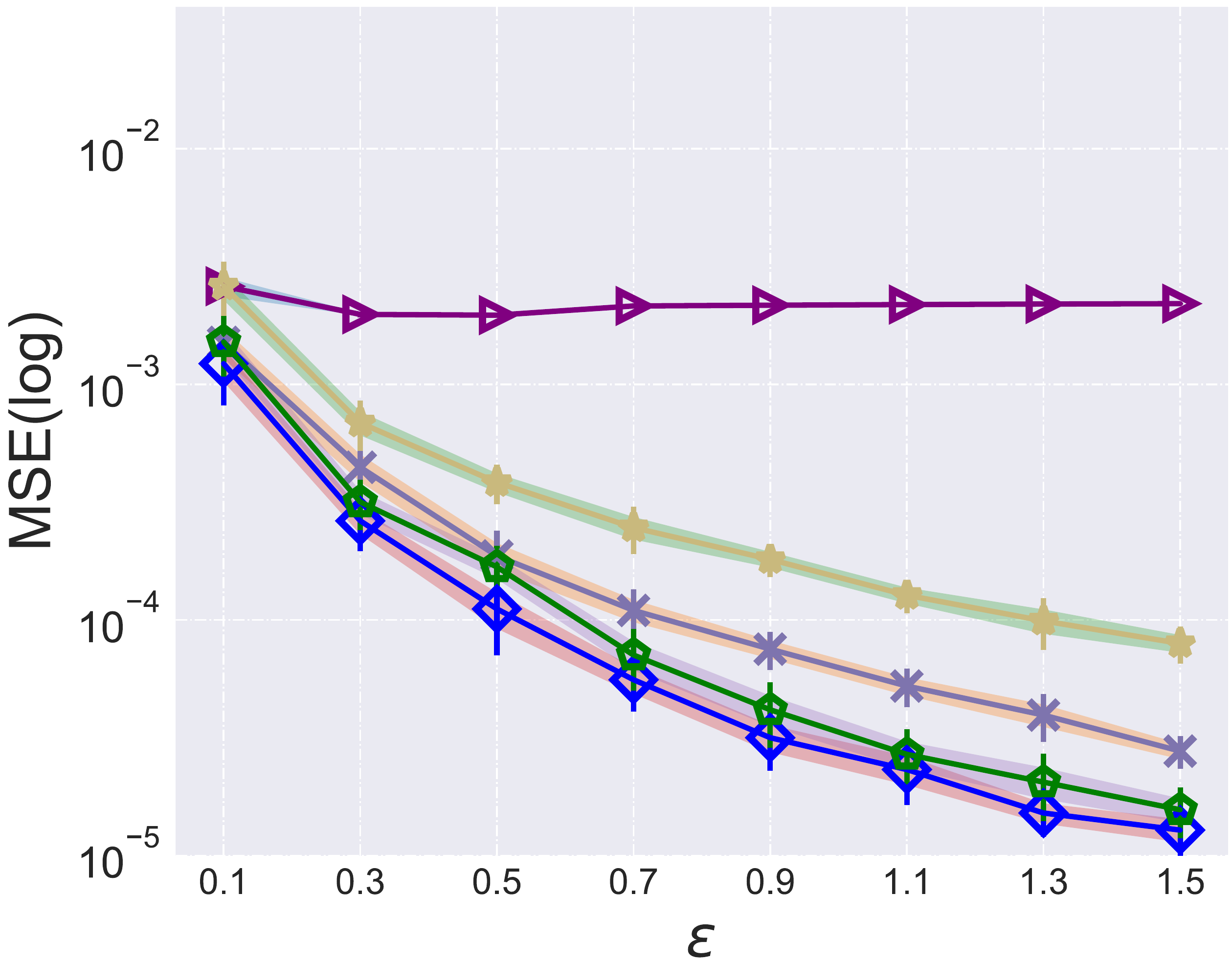}}
    \subfigure[3-dim \Gaussian, $N=10^6$, vary $\epsilon$]{\label{fig:Three_Normal08-Set_10_6-Domain_6_6_6}\includegraphics[width=0.24\hsize]{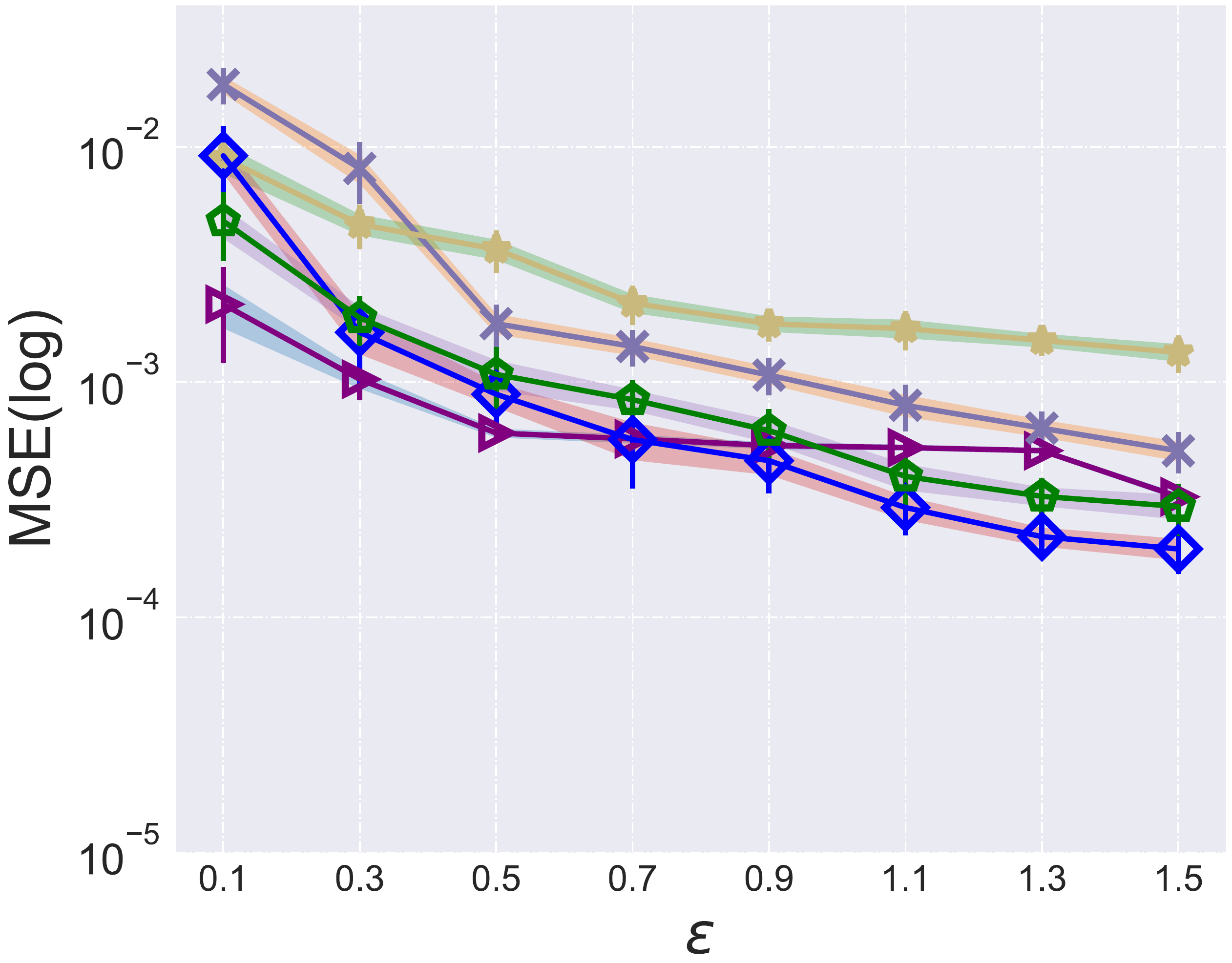}}
    \subfigure[3-dim \Gaussian, $N=10^7$, vary $\epsilon$]{\label{fig:Three_Normal08-Set_10_7-Domain_6_6_6}\includegraphics[width=0.24\hsize]{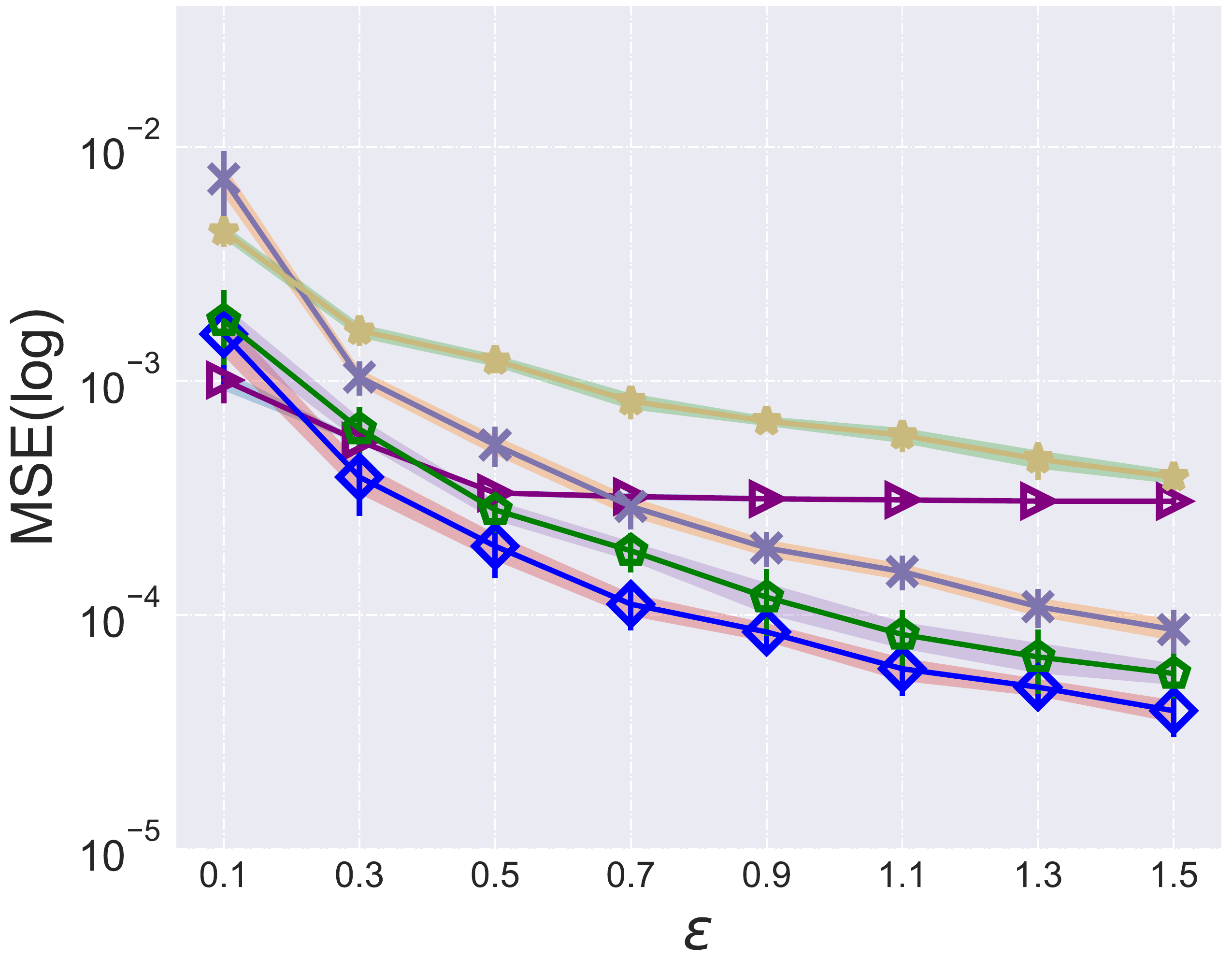}}\\[-2ex]
    \subfigure[5-dim \Laplacian, $N=10^6$, vary $\epsilon$]{\label{fig:Five_Laplace08-Set_10_6-Domain_6_6_6_6_6}\includegraphics[width=0.24\hsize]{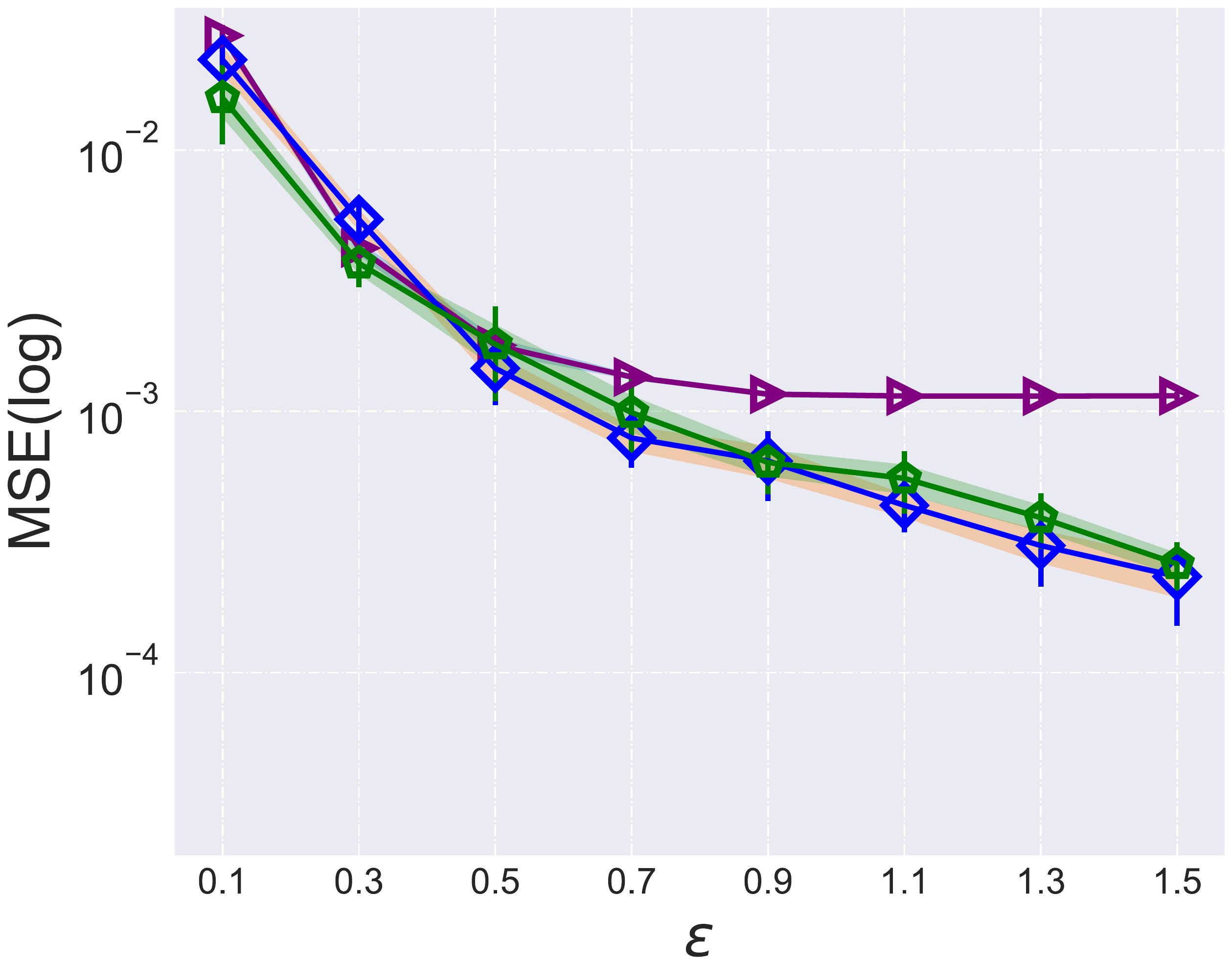}}
    \subfigure[5-dim \Laplacian,$N=10^7$, vary $\epsilon$]{\label{fig:Five_Laplace08-Set_10_7-Domain_6_6_6_6_6}\includegraphics[width=0.24\hsize]{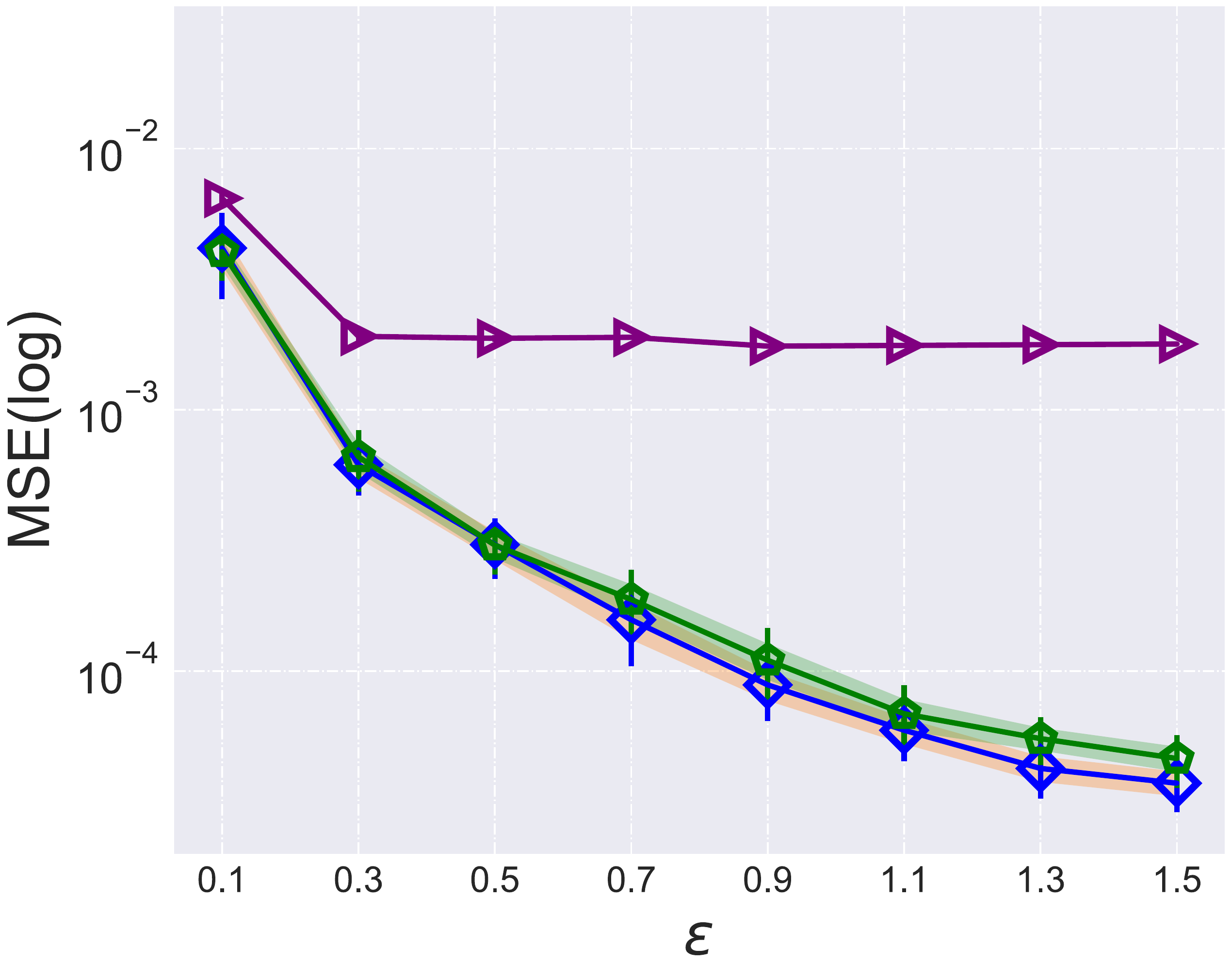}}
    \subfigure[5-dim \Gaussian, $N=10^6$, vary $\epsilon$]{\label{fig:Five_Normal08-Set_10_6-Domain_6_6_6_6_6}\includegraphics[width=0.24\hsize]{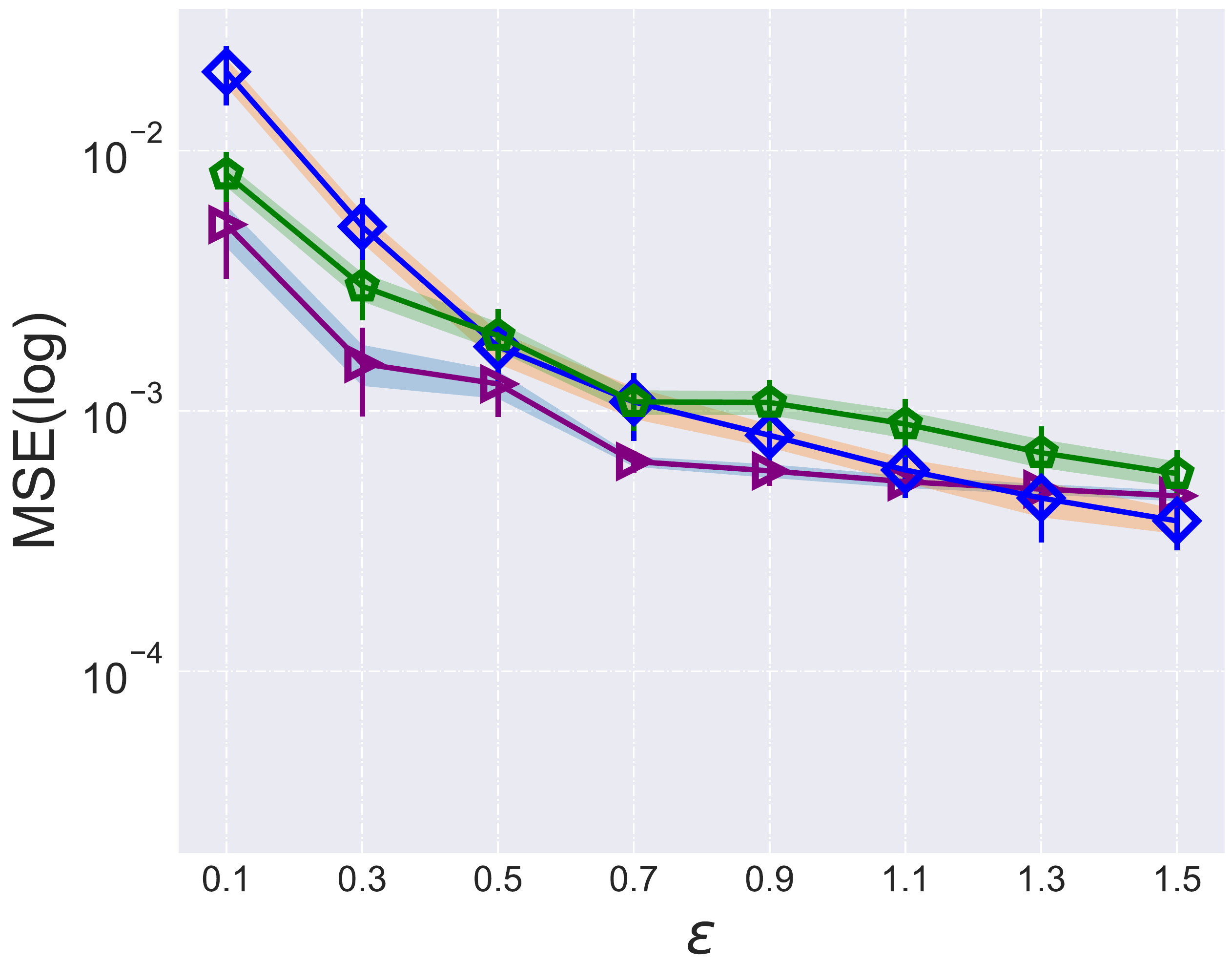}}
    \subfigure[5-dim \Gaussian, $N=10^7$, vary $\epsilon$]{\label{fig:Five_Normal08-Set_10_7-Domain_6_6_6_6_6}\includegraphics[width=0.24\hsize]{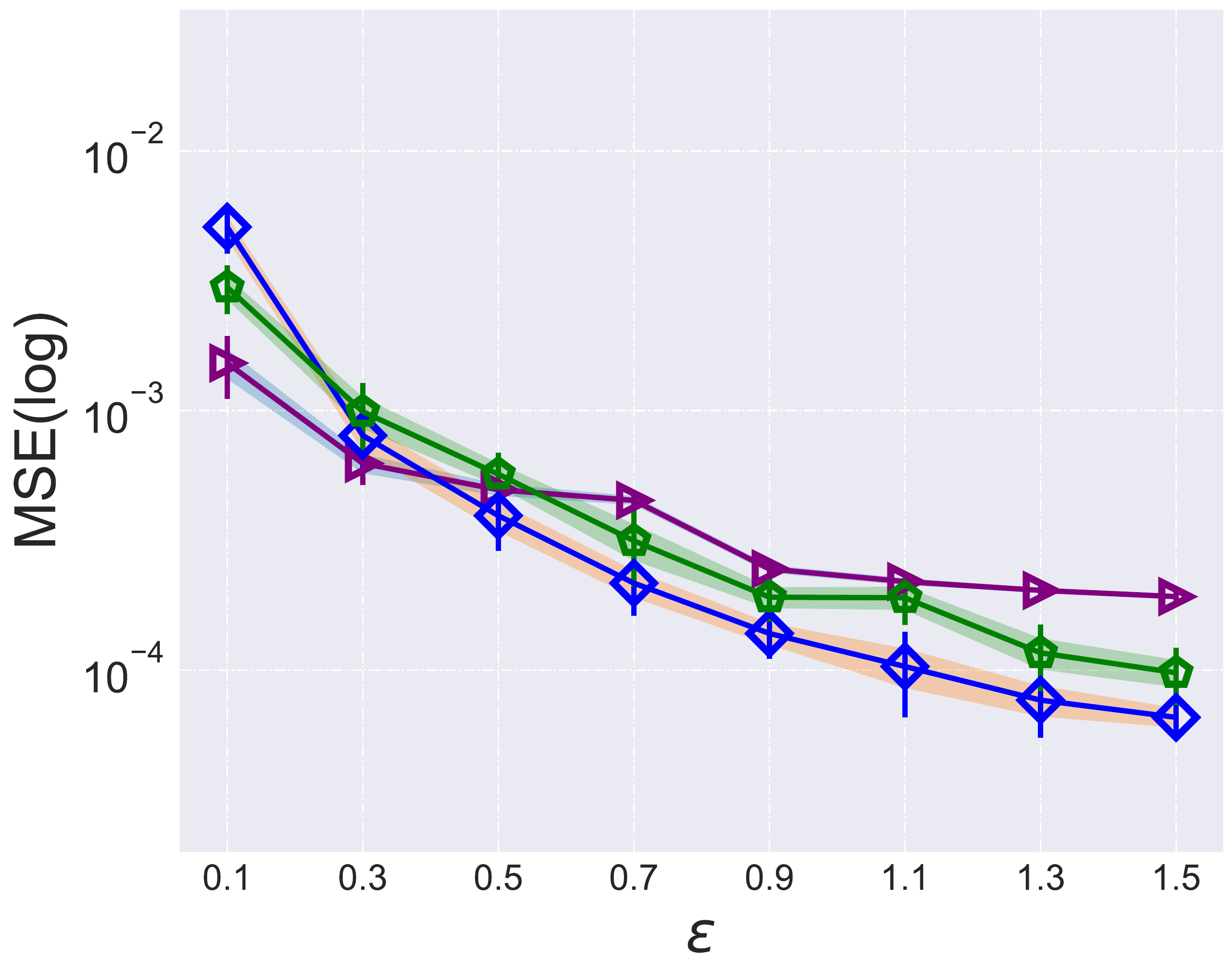}}\\[-1ex]
    \subfigure{\includegraphics[width=0.8\textwidth]{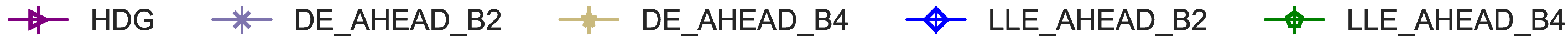}}  
    \vspace{-0.5cm}
    \caption{Comparison of different methods on high-dimensional \Laplacian and \Gaussian datasets under various privacy budgets. \de and \lle respectively represent two high-dimensional expansion methods, \ie, ``direct estimation'' and ``leveraging low-dimensional estimation''. \myHDG is a baseline method. The results are shown in log scale.}
    \vspace{-0.6cm}
    \label{3,5-dim overall compare}
\end{figure*}

\begin{figure*}[!h]
    \centering
    \subfigure[\Laplacian, $N=10^6$, $|D|= 64^m$]{
    \includegraphics[width=0.235\hsize]{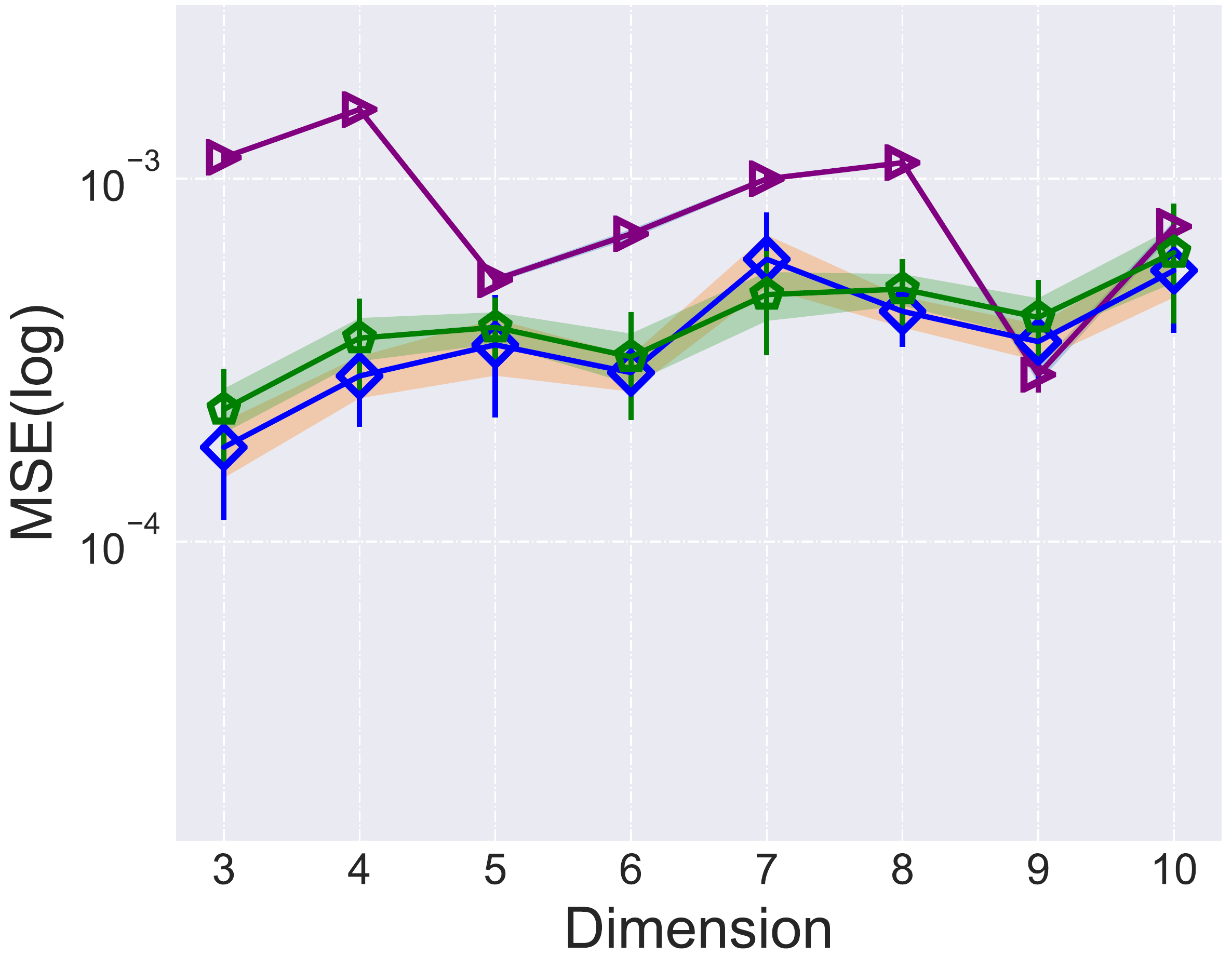}
    \label{fig: Attribute-Ten_Laplace08-Set_10_6}
    }
    \subfigure[\Laplacian, $N=10^7$, $|D|= 64^m$]{
    \includegraphics[width=0.23\hsize]{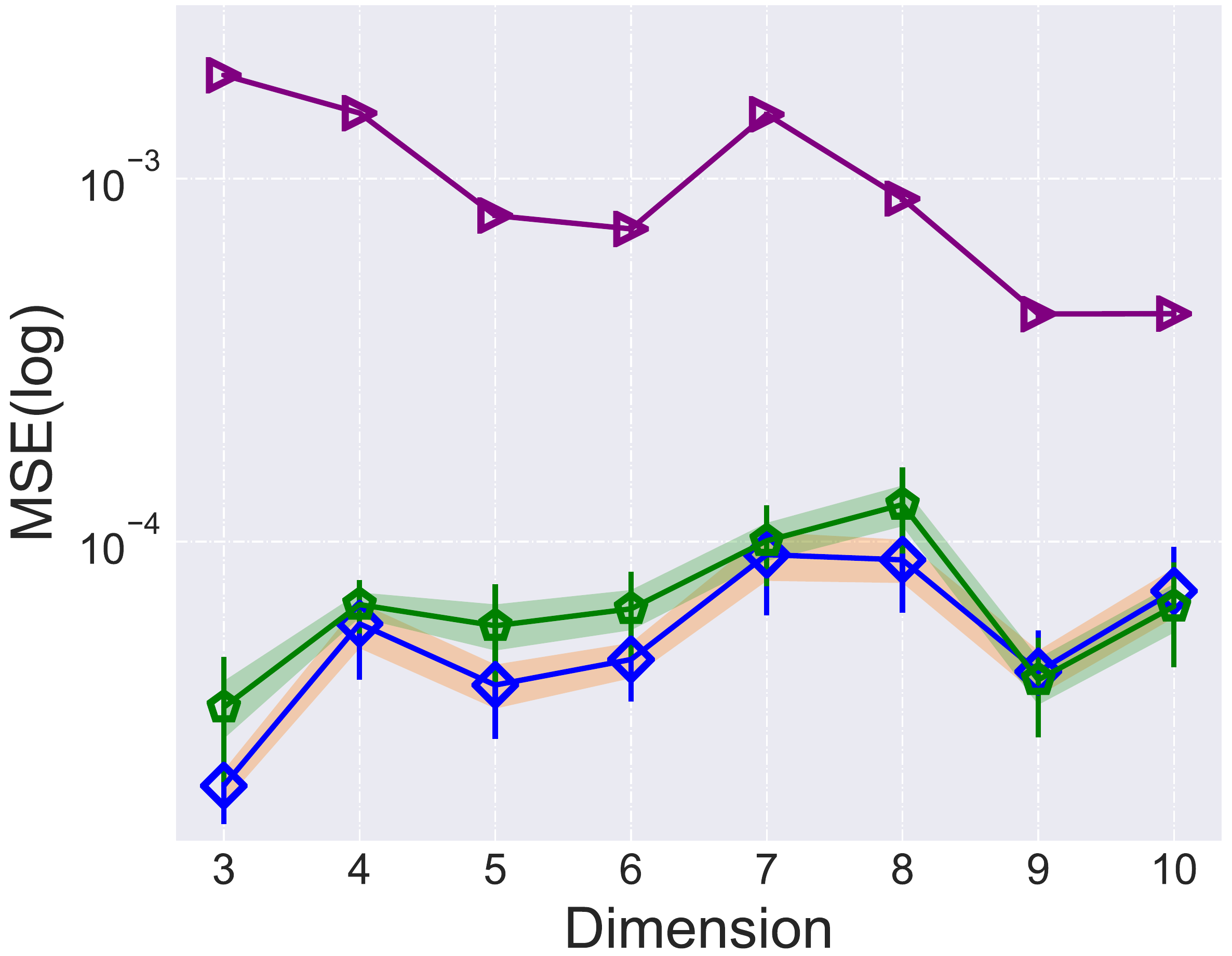}
    \label{fig: Attribute-Ten_Laplace08-Set_10_7}
    }
    \subfigure[\Gaussian, $N=10^6$, $|D|= 64^m$]{
    \includegraphics[width=0.235\hsize]{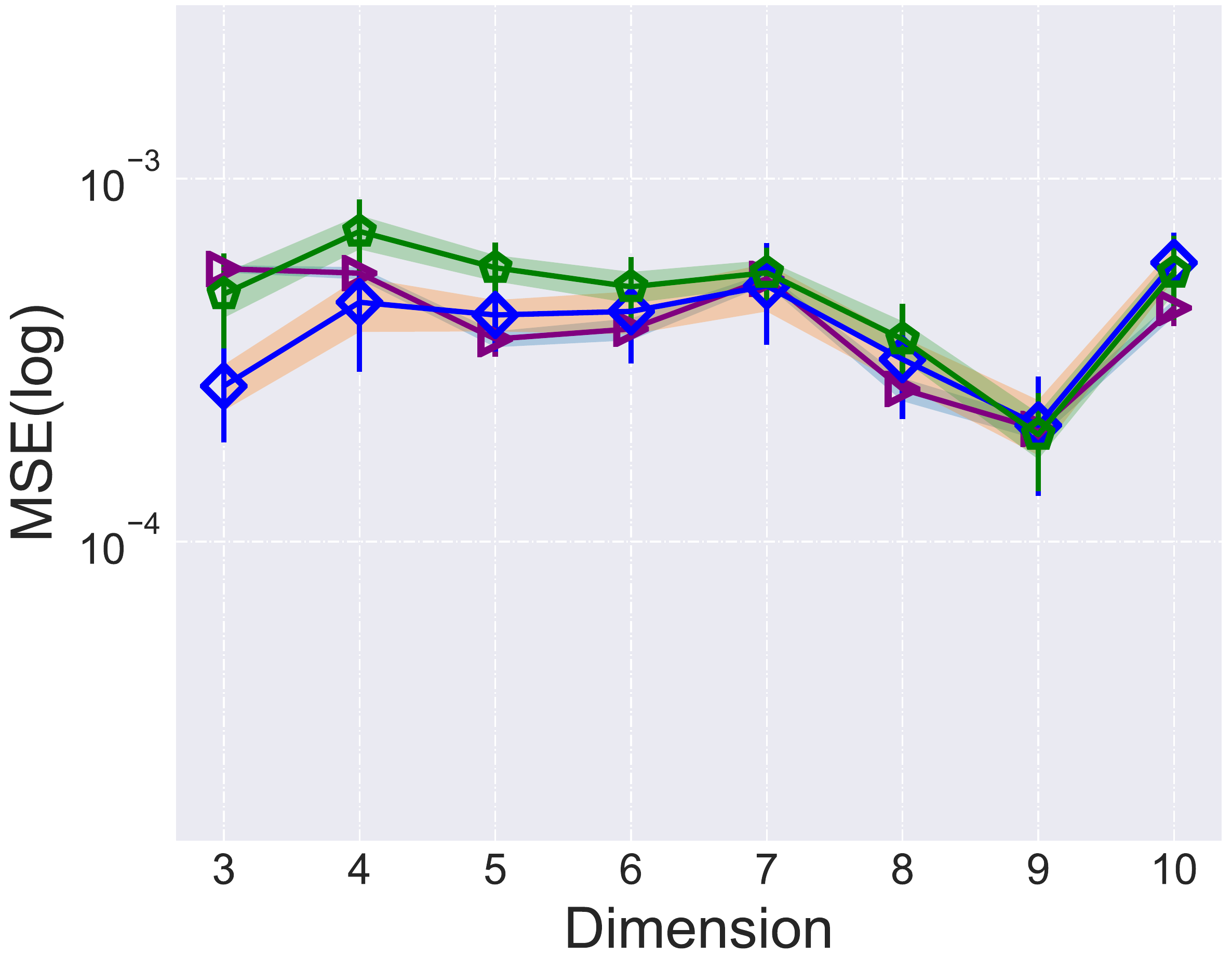}
    \label{fig: Attribute-Ten_Normal08-Set_10_6}
    }
    \subfigure[\Gaussian, $N=10^7$, $|D|= 64^m$]{
    \includegraphics[width=0.235\hsize]{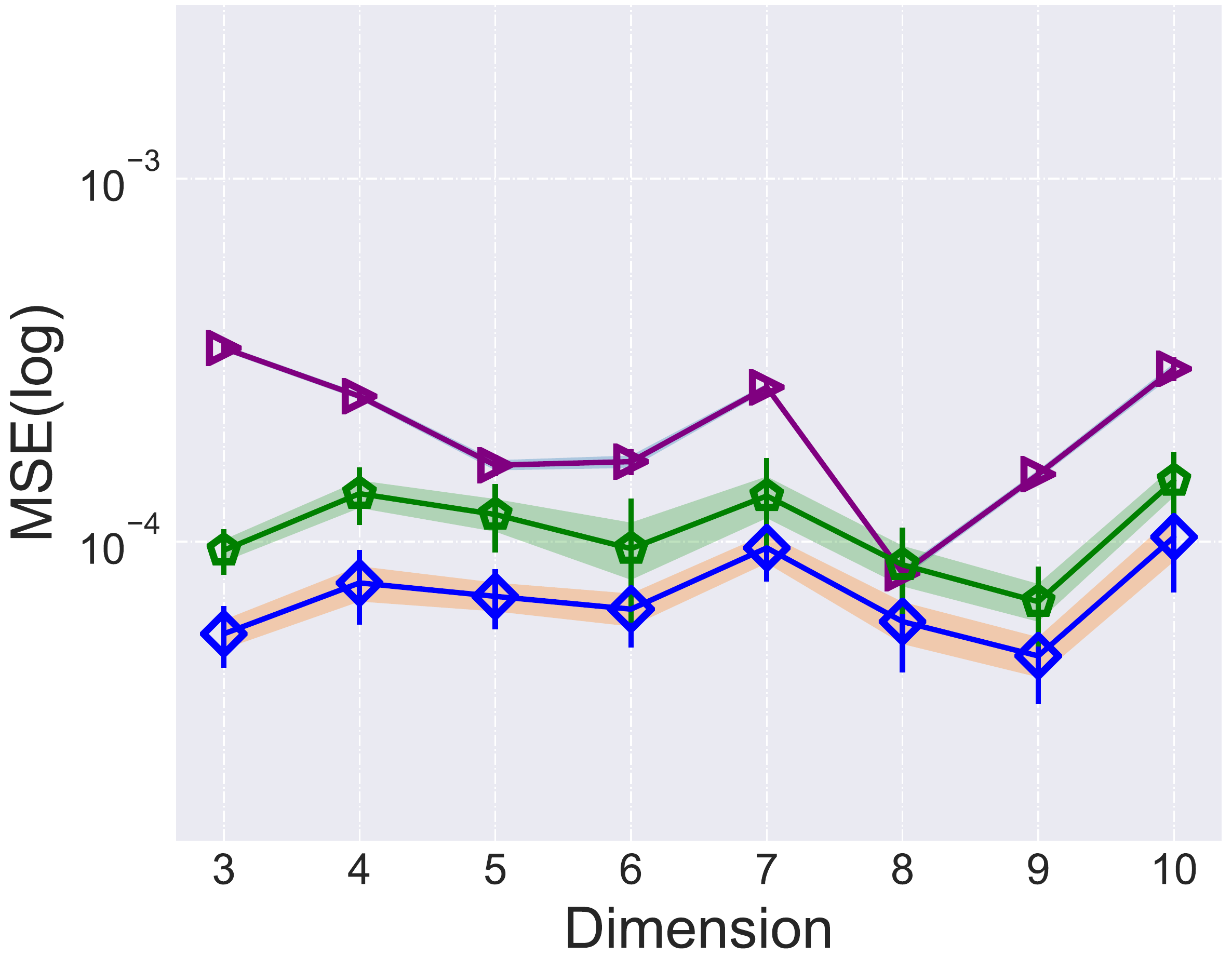}
    \label{fig: Attribute-Ten_Normal08-Set_10_7}
    }\\[-1ex]
    \subfigure{\includegraphics[width=0.45\textwidth]{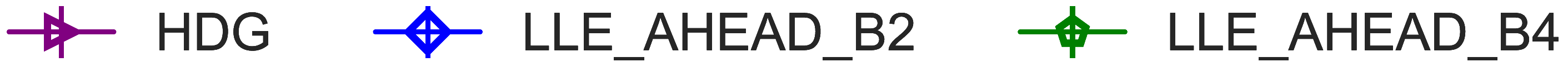}}  
    \vspace{-0.5cm}
    \caption{The MSE of different methods when varying data dimension $m$ with $\epsilon = 1.1$. The results are shown in log scale. }
    \vspace{-0.4cm}
    \label{fig: varying attribute dimension}
\end{figure*}

In this subsection, we evaluate \myahead on high-dimensional private datasets. 
As observed in \autoref{2-dim overall compare} and \autoref{2-dim correlation} (where the domain size is $256 \times 256$ and $1024 \times 1024$), MSEs of \myHDG and \myahead are not sensitive to domain size. 
Therefore, we fix the domain size $|D|=2^6$ and conduct the experiments on synthetic datasets with $10^6$ and $10^7$ records, 
which are sampled from high-dimensional \Laplacian and \Gaussian distributions, respectively. 
We refer the readers to \autoref{High-dimensional Range Query on Real datasets} for the evaluation \myahead on high-dimensional real-world datasets due to space limitation. 

\mypara{MSE under Two Expansion Methods}
Recalling the two extension ways in \autoref{Extension to Multi-dimensional Settings}, \ie, \de and \lle, we compare the MSEs of the two expansion methods under the 3-dim range query. 
For queries higher than 3 dimensions, 
we only consider \lle since the \de method is too time-consuming. 

From \autoref{3,5-dim overall compare}, we observe that \myahead with \lle obtains lower MSEs than \de. 
The sub-domains obtained by \de are equilateral high-dimensional cubes. 
\myahead tends to use the underlying nodes, inducing more nodes used in the answering process. 
For example, when answering query $[1,1]\times[1,8]\times[1,8]$ in a 3-dim dataset with domain $[1,8]\times[1,8]\times[1,8]$, 
\myahead selects the 64 leaf nodes at the bottom to answer the query instead of using higher-level nodes with larger sub-domains, 
which causes error accumulation in query answer. 
Although \lle needs to divide users and gather different attribute combination records from each user group, 
\ie, holding fewer records to estimate the frequency for each layer of \myahead trees compared to \de, 
the maximum entropy optimization step of \lle incorporates more information, 
thus resulting in smaller MSEs. 

\mypara{MSE under Various Privacy Budget}
\autoref{3,5-dim overall compare} shows the results over 3-dim and 5-dim datasets, respectively. 
We adopt correlation coefficient $r = 0.8$ to quantify the correlation between attributes, following the same setting as \cite{yang2020answering}. 

Based on the results, we have the following observations which are consistent with our analysis in \autoref{Extension to Multi-dimensional Settings}. 
1) \myahead is robust to the changes in data distributions. 
Recalling the derivation of parameters $\theta$ and $B$ in Section 4.4, \myahead does not have specific requirements for the distribution of users' data. 
When the entire domain is adaptively divided, the same $\theta$ ensures that the frequencies of intervals have similar overall errors. 
Therefore, from \autoref{fig:Three_Laplace08-Set_10_7-Domain_6_6_6}, \autoref{fig:Three_Normal08-Set_10_7-Domain_6_6_6} (or \autoref{fig:Five_Laplace08-Set_10_7-Domain_6_6_6_6_6}, \autoref{fig:Five_Normal08-Set_10_7-Domain_6_6_6_6_6}), 
\myahead behaves similarly on different datasets. 
Under the same dataset parameters, \myHDG uses the same granularity 1, 2-dim grids to aggregate user records, 
then answering queries based on the uniform assumption. 
Therefore, on the \Laplacian distribution with more uneven frequency, the MSEs of \myHDG are larger than those on the \Gaussian distribution. 
2) For high-dimensional queries, the performance of \myahead is more dependent on the scale of user records. 
Compared with \myHDG which only needs to divide users into different attribute combinations, \myahead needs to further part users into different layers in 2-dim trees (recall \autoref{Associated 2-dim AHEAD Tree Construction} in \autoref{Extending AHEAD to high-dimensional range query}). 
For a 3-dim dataset with attribute domain size $|D|=64$, \myahead randomly divides the users into $C_3^2 \cdot 6$ groups, while \myHDG separates users into $C_3^2 + C_3^1$ groups, where the number of user records in each group of \myahead is half of that of \myHDG, \ie, doubling noise error of \myahead. 
Therefore, from Figure \ref{fig:Three_Laplace08-Set_10_6-Domain_6_6_6} and \ref{fig:Three_Laplace08-Set_10_7-Domain_6_6_6}, 
we know that the superiority of \myahead will decrease with fewer user records. 

\mypara{MSE under Various Data Dimension}
We evaluate the performance of \myahead and \myHDG with varying dimensions of data (from 3-dim to 10-dim), which are sampled from multivariate \Laplacian and \Gaussian distribution with correlation coefficient $r = 0.8$. 

From \autoref{fig: varying attribute dimension}, \myahead is robust to the changes in data dimension. 
Recalling \autoref{Extension to Multi-dimensional Settings}, with the data dimension $m$ increasing, \lle needs to divide users into more groups due to more attribute combinations, \ie, holding fewer user records to estimate the frequency for each 2-dim \myahead tree. 
For instance, when $m = 3$, the number of attribute combination is $C_3^2 = 3$, and when $m = 10$, the number of attribute combination is $C_{10}^2 = 45$. 
Since \lle constructs a query set with associated $2^m$ queries (recall Algorithm \ref{Estimating Answer of m-dimensional Range Query} in \autoref{Extending AHEAD to high-dimensional range query}), 
the maximum entropy optimization step of \lle incorporates more query information with higher dimension, thus resulting in almost consistent MSEs with dimension increasing.

\section{Discussion}
\mypara{Highlights of \myahead}
1) Through dynamically building the tree structure, \myahead addresses the limitations of the state-of-the-art LDP methods, which significantly enhances the query accuracy and can motivate the development of future LDP based privacy-preserving frameworks. 
2) To overcome the hindrance in finding the optimal partition \cite{muthukrishnan1999rectangular}, 
\myahead decomposes the domain by leveraging the tree fanout $B$ and the threshold $\theta$, which are theoretically derived under rigorous LDP guarantees. 
From the experimental results, 
these parameter settings work 
well both in low and high dimensional scenarios. 
3) By conducting an in-depth analysis of \myahead under various privacy budget, domain size, user scale, distribution skewness, data dimension and attribute correlation, 
we conclude some useful observations for adopting \myahead (recall the observations in \autoref{Practical Deployment of AHEAD}). 

\mypara{Limitations and Future Work}
Below, we discuss the limitations of \myahead and promising directions for further improvements. 
1) For high(>2)-dimensional range query, \myahead is sensitive to the scale of user records. 
\myahead needs to divide users 2 times, \ie, partition into different attribute combinations and different layers in 2-dim trees. 
Thus, when the number of group user is large, \myahead needs sufficient user records to ensure the accuracy of the frequency estimation of each 2-dim layer. 
2) Compared to \myHDG using only two grids, \ie, finer-grained 1-dim and sparser-grained 2-dim grids, 
\myahead generates 2-dim intervals with various granularities to decompose the entire domain. 
Therefore, in practice, \myahead requires longer frequency value searching time and larger memory usage compared to \myHDG. 
A layer fusion strategy can be designed for each 2-dim \myahead tree to compress tree height. 
3) \myahead is a completely dynamic framework, 
which uses the threshold to determine the division of each sub-domain. 
In the cases of large domain size, each sub-domain needs to be compared with the threshold, which unavoidably increases the computational overhead. 
Since the higher-level node represents a relatively large sub-domain (whose frequency value is generally greater than the threshold), a `static' and `dynamic' hybrid tree structure can be potentially adopted.
For instance, the top few layers can use the static framework, while the remaining layers can leverage the threshold to decompose the sub-domains in line with the dynamic framework. 

\section{Related Work}
\label{Related Work}
\mypara{Frequency Estimation under LDP}
The notion of local differential privacy (LDP) was introduced in \cite{kasiviswanathan2011can}. 
Duchi \etal \cite{duchi2013local} systematically analyzed the LDP algorithm and gave LDP's theoretical upper bound based on information theory. 
For LDP scenarios, one of the upmost basic tasks is frequency estimation of user values. 
Erlingsson \etal \cite{erlingsson2014rappor} proposed \rappor, which is the first practical example of frequency estimation. 
The algorithm uses the Bloom filter \cite{bloom1970space} to encode the private data, 
and then leverages a random response (\rr) method \cite{warner1965randomized} to perturb the encoded data. 
After that, several mechanisms \cite{warner1965randomized, bassily2017practical, bassily2015local, duchi2013local} were also proposed for frequency estimation under LDP. 
Wang \etal \cite{wang2017locally} compared the estimation variance of different algorithms 
and gave algorithm recommendations under different domain sizes. 
Wang \etal \cite{wang2020set} proposed a wheel mechanism with a same variance as \oue \cite{wang2017locally}. 

\mypara{Marginal Release under LDP}
Marginal release are widely studied under LDP. 
Kulkarni \etal \cite{cormode2018marginal} proposed to apply the Fourier transformation method and Ren \etal \cite{ren2018textsf} proposed to apply the expectation maximization method. 
The state-of-the-art method \mycalm \cite{zhang2018calm}, strategically choose sets of attributes and adaptively choose a randomization algorithm to reduce the noise effect. 
We notice that the methods for marginal release can be also used to answer range queries. Thus, we have also analysed  \mycalm's performance in \autoref{Consistent Adaptive Local Marginal}. 

\mypara{Range Query under DP}
The range query problem has been studied extensively in the centralized setting including works based on hierarchy \cite{cormode2012differentially, hay2010boosting, li2010optimizing, qardaji2013understanding}, 
coarsened domain \cite{li2014data, qardaji2013differentially, xiao2012dpcube, xu2013differentially, zhang2014towards}, 
Wavelet or Fourier transformation \cite{xiao2010differential, acs2012differentially}, 
or publishing a synthetic dataset \cite{hardt2012simple}. 
Hay \etal \cite{hay2016principled} proposed DPBench to evaluate algorithms for answering 1-dim and 2-dim range queries. 
McKenna \etal \cite{mckenna2018optimizing} presented \hdmm to answer workloads of predicate counting queries. 
Both \myngram \cite{chen2012differentially} and \myahead adopt the adaptive tree strategies and leverage the decomposition threshold to balance utility and privacy. 
However, they are different in terms of the DP setting and targeted data type. 
1) \myngram is designed and analyzed for DP scenarios, while \myahead is orientated to the local setting. 
2) \myngram publishes the sequential data with DP guarantee, while \myahead aggregates one/multi-dimensional ordinal attributes satisfying LDP. 

\mypara{Range Query under LDP}
For 1-dim scenarios, 
Cormode \etal \cite{cormode2019answering} applied the Haar wavelet tranform to the LDP setting and proposed \mydht. 
Wang \etal \cite{wang2019answering} leveraged the idea of hierarchical intervals and presented \myhio mainly for answering 1-dim and 2-dim range queries. 
For answering 2-dim and high-dimensional range queries, Wang \etal \cite{yang2020answering} designed the state-of-the-art method \myHDG, which is inspired by the Adaptive Grids approach \cite{qardaji2013differentially} under DP. 
There are some works that utilized an relaxation of $\epsilon$-LDP \cite{xiang2019linear} or leveraged the properties of workload \cite{ edmonds2020power} to achieve significant gains in utility of query answer. 

Besides the above problems, DP (LDP) has been also used for other data analysis tasks, such as heavy hitters \cite{wang2018locally, wang2018privtrie, qin2016heavy, bassily2017practical}, location \cite{andres2013geo, chen2016private}, graph data \cite{qin2017generating, sun2019analyzing, ji2015your, ji2015secgraph}, key-value data \cite{ye2019privkv, gu2020pckv}, evolving data \cite{joseph2018local}, machine learning \cite{abadi2016deep, jordon2018pate, wang2020privacy}. 
Their problem definitions are different from ours, thus not suitable for comparison. 

\section{Conclusion}
\label{conclusion}

In this work, we propose a novel LDP protocol for the one and multi-dimensional range query problem, by leveraging adaptive hierarchical decomposition. 
Our method satisfies rigorous LDP guarantees while achieving advantageous utility performance with the theoretically-derived parameters. 
Through theoretical analysis as well as extensive experimental evaluation, 
we show the effectiveness of \myahead in balancing utility and privacy for range queries 
and its significant advantages over the state-of-the-art methods. 
Furthermore, by studying various parameter settings, 
we conclude several important observations for adopting \myahead in practice. 
Our source code is available on GitHub at \url{https://github.com/link-zju/ccs21-AHEAD}.
\section*{acknowledgement}
We would like to thank Lisa M. Tolles from Sheridan Communications and the anonymous reviewers for their helpful comments. 
This work was partly supported by NSFC under Grants 62088101, 61833015, 61772466, U1836202, the Zhejiang Provincial Natural Science Foundation for Distinguished Young Scholars under No. LR19F020003, and the Fundamental Research Funds for the Central Universities (Zhejiang University NGICS Platform). 
Changchang Liu was partly sponsored by the Combat Capabilities Development Command Army Research Laboratory and was accomplished under Cooperative Agreement Number W911NF-13-2-0045 (ARL Cyber Security CRA).

\bibliographystyle{ACM-Reference-Format}
\bibliography{sample-base}

\appendix
\section{The rationality of sampling principle}
\label{The rationality of sampling principle}
\myahead can adopt two strategies to utilize the privacy budget: the privacy budget splitting strategy and the user partition strategy. 
More specifically, the privacy budget splitting strategy divides the whole privacy budget $\epsilon$ into $c$ pieces and estimates the frequency distribution with all users’ reports under privacy budget $\epsilon/c$. The user partition strategy randomly assigns the users into $c$ groups and uses the whole privacy budget to obtain the frequencies from each group of users. For the same node, the variances of the two strategies are recorded as $\myvar_1$ and $\myvar_2$. From Equation 2, we have 
$\myvar_{1}=\frac{4 e^{\frac{\epsilon}{c}}}{N\left(e^{\frac{\epsilon}{c}}-1\right)^{2}}$ and 
$\myvar_{2}=\frac{c}{N} \frac{4 e^{\epsilon}}{\left(e^{\epsilon}-1\right)^{2}}$. 
Since $\epsilon$ is greater than 0 and $c$ is a positive integer greater than 1, we know $\myvar_1 > \myvar_2$. 
Therefore, under the same setting, the user partition strategy has a less noise error than the privacy budget splitting strategy. 

For inconsistency of different levels, previous studies \cite{nguyen2016collecting, wang2017locally, wang2019locally} have proven that the error of user partition is small. In addition, when answering queries, \myahead uses the combination with the least number of nodes, so the inconsistency will not cause conflicts when answering queries. 

\section{\myahead satisfies $\epsilon$-LDP}
\label{AHEAD satisfies LDP}
\myahead can answer range query while satisfying rigorous LDP guarantees as shown in the following theorem. 
\begin{theorem}
    \myahead satisfies $\epsilon$-LDP.
    \label{satisfy LDP}
\end{theorem}
\begin{proof}
    In Step 1, the users randomly select their group ID in $[1,2, \cdots, c]$ without privacy budget consumption. 
    In Step 2 and Step 3, \myahead sequentially interacts with users, and each user produces a single output. 
    In Step 4, \myahead does not touch users' private data, thus incurring no additional privacy budget. 
    Therefore, if the interaction with users (Step 2 and Step 3) meets the $\epsilon$-LDP, \myahead satisfies $\epsilon$-LDP \cite{joseph2020exponential}. 
    
    In each interactive round, \myahead constructs the noisy frequency based on the \oue protocol with privacy budget $\epsilon$ in Step 2. 
    For any pair of possible values $v_1, v_2 \in D$ belonging to the same user in the group $g$, a noisy binary vector $O$ is every potential output in the range of \oue. 
  \begin{small}
  \begin{align}\label{6}
  \frac{\Pr{O|v_1, g}}{\Pr{O|v_2, g}} &= \frac{\Pi_{i=1}^{\ell}\Pr{O[i]|v_1, g}}{\Pi_{i=1}^{\ell}\Pr{O[i]|v_2, g}} \nonumber \\
                        &= \frac{\Pi_{i=1}^{\ell}\Pr{O[i]|v_1} \cdot \Pr g}{\Pi_{i=1}^{\ell}\Pr{O[i]|v_2} \cdot \Pr g} \nonumber \\
                                          &\leq{\frac{\Pr{O[v_1]=1|v_1}\Pr{O[v_2]=0|v_1}}{\Pr{O[v_1]=1|v_2}\Pr{O[v_2]=0|v_2}}} \nonumber \\
                                          &= \frac{p}{q}\cdot\frac{1-q}{1-p} \nonumber \\
                                          &= \frac{1/2}{1/(1+e^\epsilon)}\cdot\frac{1-1/(1+e^\epsilon)}{1-1/2}= e^\epsilon, 
  \end{align}
  \end{small}
  where $\ell$ is the length of $O$. 
  The $p$ and $q$ are the flipping probabilities of the \oue protocol. When $p = \frac{1}{2}$ and $q = \frac{1}{(e^\epsilon +1)}$, \oue can obtain the minimal variance \cite{wang2017locally}. 
  From \autoref{6}, Step 2 satisfies $\epsilon$-LDP. 
  Since Step 3 processes uploaded noisy data, i.e., not using the users’ original private data, there is no additional privacy budget consumption. 
  Because of the above, the overall process satisfies $\epsilon$-LDP. 
\end{proof}

\section{Proof of Theorem 4.1}
\label{The proof of weighted average}
    Using \autoref{weighted average} to combine the frequencies of child nodes, the node $n$ can achieve the minimal updated variance.
\begin{proof}
According to \cite{wanglocally}, when $D$ is large and $\epsilon$ is not too large, 
each interval's estimated frequency value $\hat{f}$ is approximate to $f + X$, where $f$ is the true frequency value and $X$ is a random variable following \Gaussian distribution $\mathcal{N}(0, \myvar_{\oue})$. 
If $u$ is a leaf node, without loss of generality, we assume that 
\begin{align}
\hat{f}(n) = f(n) + \mathcal{N}(0, \myvar_{(n)}) \nonumber
\end{align}
\begin{align}
\sum_{u \in child(n)} \hat{f}(u) = \sum_{u \in child(n)}f(u) + \mathcal{N}(0, \myvar_{(u)}) \nonumber
\end{align}
The interval of node $n$ equals to the combination of its children intervals, thus 
\begin{align}
f(n) = \sum_{u \in child(n)}f(u) \nonumber
\end{align}
Assuming the used coefficients for weighted average are $\lambda_1$ and $\lambda_2$, the weighted average between $f(n)$ and $\sum_{u \in child(n)}f(u)$ is an unbiased estimation of node $n$'s frequency with $\lambda_1 + \lambda_2 = 1$. 
\begin{align}
\tilde{f}(n)=\lambda_1\hat{f}(n) + \lambda_2\sum_{u \in child(n)}\hat{f}(u) \nonumber
\end{align}
The variance of $\tilde{f}(n)$ is equal to $\lambda_1^2\myvar_{(n)} + \lambda_2^2\myvar_{child(n)}$. When 
\begin{align}
\lambda_1 = \frac{\myvar_{child(n)}}{\myvar_{child(n)} + \myvar_{(n)}} \nonumber,~~~~
\lambda_2 = \frac{\myvar_{(n)}}{\myvar_{child(n)} + \myvar_{(n)}}, \nonumber
\end{align}
we can minimize the variance of $\tilde{f}(n)$. 
After the weighted average, the variances still fit \Gaussian distribution. Thus the above proof still holds when $u$ is not a leaf node. 
\end{proof}

\section{The pseudo-code of 2-dim \myahead}
\label{The pseudo-code of 2-dim AHEAD}
As shown in \autoref{Construct 2-dim prototype tree} and \autoref{alg:2-dim post_processing}, the main difference between 1-dim and 2-dim \myahead is the decomposition process (lines 4, 9 of \autoref{Construct 2-dim prototype tree}). 

\floatname{algorithm}{Algorithm}
\renewcommand{\algorithmicrequire}{\textbf{Input:}}
\renewcommand{\algorithmicensure}{\textbf{Output:}}

{\color{red}
\begin{algorithm}[!h]
    \caption{2-dim \myahead Tree Construction}
    \label{Construct 2-dim prototype tree}
	\begin{algorithmic}[1]
        \Require All users’ value set $V = \{v_1, v_2,\ldots,v_N \}$, attribute domain $D$, tree fanout $B$, privacy budget $\epsilon$, threshold $\theta$
        \Ensure \myahead Tree $T$
		\State $c = {\log_B}|D|$
        \State $//$ Step 1: User partition
        \State Randomly divide users into $c$ parts $\left\{ V_1, V_2, \ldots, V_c \right\}$
        \State Create the root node of tree $T$ with initial interval $e^0_0 = [1, |D|] \times [1, |D|]$ and $T.\operatorname{node}(e^0_0).frequency = 1$
		\For {$i$ from 1 to $c$}
            \State $//$ Step 2: New decomposition generation
            \For {$j$, node in enumerate($T.\operatorname{node}(level=i-1)$)}
            \If {$\operatorname{node.frequency} > \theta$}
            \State Divide $\operatorname{interval} e^{i-1}_j$ into $B$ disjoint intervals $\{e^{i-1}_{j,k}\}$
            \For {$k$ from $1$ to $B$}
    		\State $\operatorname{node.add\_child}(e^{i-1}_{j,k})$
            \EndFor
            \Else
            \State $\operatorname{node.add\_child}(e^{i-1}_{j})$
            \EndIf    		
            \EndFor

            \State $//$ Step 3: Noisy frequency construction
            \State $F = \fo(V_i, T.\operatorname{node}(level=i).\operatorname{interval}, \epsilon)$
            \For {$k$, node in enumerate $T.\operatorname{node}(level=i)$}
    		\State $node.\operatorname{frequency} = F[k]$
            \EndFor
        \EndFor
        \State $//$ Step 4: Post-processing
        \State Run \autoref{alg:2-dim post_processing}
        \State return $T$
	\end{algorithmic}
\end{algorithm}

\begin{algorithm}[!h]
	\caption{Post-processing}
	\label{alg:2-dim post_processing}
	\begin{algorithmic}[1]
        \Require \myahead tree $T$, tree fanout $B$
        \Ensure \myahead tree $T$
        \For {$i$ from $1$ to $c$}
		\State norm\_sub($T.\operatorname{node(level=i)}.\operatorname{frequency}$)
        \EndFor    
        \For {$j$ from $c-1$ to $1$}
            \For {\_, node in enumerate $T.\operatorname{node}(level=j)$}
            \State $f_1$ = node.frequency, $f_2$ = $\sum{\operatorname{node.children().frequency}}$
		    \State node.frequency = $\lambda_1$$f_1$ + $\lambda_2$$f_2$
            \EndFor 
        \EndFor    
        \For {$k$ from $1$ to $c$}
            \For {\_, node in enumerate $T.\operatorname{node}(level=k)$}
                    \If {node.children() == None}:
		            \State $\operatorname{node}.\operatorname{add\_children}()$
		            \State $\operatorname{node}.\operatorname{children()}.\operatorname{frequency} = \operatorname{node}.\operatorname{frequency}/B$
                    \EndIf
            \EndFor 
        \EndFor    
        
    \end{algorithmic}
\end{algorithm}
}

\section{Complexity Analysis} 
\label{Complexity Analysis}
Here, we provides a detailed complexity analysis of the algorithms used in our evaluation. 
For ease of exposition, we assume that all attributes have the same domain $D$. 

\begin{table}[h]
\centering
\caption{Comparison of complexity for different methods. The table lists the server-side computation, server-side storage, and client-server commutation. }
    \begin{tabular}{|c|c|c|c|c|}
    \hline
    \multicolumn{2}{|l|}{}   & Time & Space & Comm \\ \hline
    \multirow{4}{*}{1-dim} & \myahead &  $O(log_{|D|}2\cdot N \cdot|D|)$   & $O(|D|)$    &  $O(|D|)$          \\ \cline{2-5} 
                         & \mycalm  &  $O(N\cdot|D|)$                    & $O(|D|)$    &  $O(|D|)$          \\ \cline{2-5} 
                         & \myhio   &  $O(log_{|D|}5\cdot N \cdot|D|)$   & $O(|D|)$    &  $O(log(|D|))$     \\ \cline{2-5} 
                         & \mydht   &  $O(N+|D|^3)$                      & $O(|D|^2)$  &  $O(log(|D|))$     \\ \hline
    \multicolumn{5}{|l|}{}                             \\ \hline
    \multirow{2}{*}{2-dim} & \myahead & $O(log_{|D|}2\cdot N \cdot|D|^2)$  & $O(|D|^2)$  &  $O(|D|^2)$        \\ \cline{2-5} 
                         & \myHDG   & $O(N)$  & $O(N^\frac{1}{2})$       & $O(log|D|)$                      \\ \hline
    \end{tabular}
\end{table}

\mypara{Time Complexity}
For 1-dim scenarios, the computation of \myahead is mainly from processing users' reports, 
which takes $O(\frac{N}{log_2^{|D|}}\cdot2^i)$ for the $i$-th group of users, 
\ie, $O(log_{|D|}2\cdot N \cdot|D|)$ in total. 
The similar arguments also hold for the hierarchy-based \myhio method, \ie, $O(log_{|D|}5\cdot N \cdot|D|)$. Because the domain size $D$ is usually large in the range query scene, 
\mycalm adopts \oue to aggregate the users' data, which takes $O(|D|)$ for each user, 
and $O(N \cdot |D|)$ in total. For \mydht, the running time is dominated by the sum operation and inverse transformation. 
Specifically, the method first sums the reports of users with the same index, which should count all of them for each user, \ie, $O(N)$. 
After that, \mydht also needs an inverse transform process to produce the estimated frequency, \ie, $O(|D|^3)$. 
For 2-dim scenarios, the domain size changes from $|D|$ to $|D| \times |D|$.
\myahead selects fanout $B = 4$ and the computation becomes $O(log_{|D|}4\cdot N \cdot|D|^2)$. 
To estimate the user frequency distribution, \myHDG should evaluate hash functions for each report from users, \ie, $O(N)$ in total. 

\mypara{Space Complexity}
We measure the storage needed except that occupied by the inputs and outputs, which is the same amount of storage for all methods. 
For 1-dim scenarios, \myahead and \myhio needs to maintain the hierarchical tree structure, requiring $O(|D|)$ storage. The \mycalm method sums length-$|D|$ vector reported by each user to calculate the frequency distribution, which also needs $O(|D|)$ storage. Due to the inverse transform process, for \mydht, a storage to count all possible intermediate results are needed, \ie, $O(|D|^2)$. 
For 2-dim scenarios, the domain becomes $D \times D$, thus \myahead require $O(|D|^2)$ storage to maintain the hierarchical tree structure. 
The storage of \myHDG depends on the granularity of 2-dim grids. 
Base on the analysis in \cite{yang2020answering}, \myHDG needs $O(N^\frac{1}{2})$ storage.

\mypara{Communication Complexity}
Since \myhio, \mydht and \myHDG uses \olh \cite{wang2017locally} as \fo, the report by \olh is only one bit, plus the index, which can be represented by $log|D|$ bits or $log|D|^2$ bits for 2-dim \myHDG. 
The \mycalm method adopts \oue, where each user should report a length-$|D|$ binary vector, 
requiring $O(|D|)$ bits. 
\myahead is dominated by distributing the domain decomposition to each user. Considering the worst case, \myahead needs $O(|D|^2)$. 

\section{Extending \myahead to high-dimensional range query} 
\label{Extending AHEAD to high-dimensional range query}

\begin{algorithm}[t]
    \caption{Associated 2-dim AHEAD Tree Construction}
    \label{Associated 2-dim AHEAD Tree Construction}
	\begin{algorithmic}[1]
        \Require All users’ value set $V = \{v_1, v_2,\ldots,v_N \}$, private attribute dimensions $m$, private attributes set $A$, attribute domain $D$, tree fanout $B$, privacy budget $\epsilon$, threshold $\theta$
        \Ensure \myahead Forest $\{T\}$
		\State $c = C^2_m$
        \State Randomly divide users into $c$ parts $\left\{ V_1, V_2, \ldots, V_c \right\}$
        \State $//$ Step 1: Building Block Construction
        \For {$k, \{a_x, a_y\}_{x \neq y}$ in enumerate(pairwise attributes)}
        \State $\{T\}.\operatorname{add}(\operatorname{2dim\_AHEAD\_tree}(V_k[a_x, a_y], D, B, \epsilon, \theta))$
        \EndFor
        \State $//$ Step 2: Consistency on Attributes
        \For {$\ell$ from 1 to $T$.height}
            \For {attribute $a_x$ in enumerate($A$)}
            \State  make $\{T.\operatorname{node}(level=\ell).\operatorname{frequency}\}$ consistent on $a_x$
            \EndFor
        \EndFor
    \State return $\{T\}$
	\end{algorithmic}
\end{algorithm}

\begin{algorithm}[t]
    \caption{Estimating Answer of $m$-dim Range Query}
    \label{Estimating Answer of m-dimensional Range Query}
	\begin{algorithmic}[1]
        \Require \myahead forest $\{T\}$, $m$-dim range query $R_{\bigcap {[\alpha_j,\beta_j]}_{j=1}^{m}}$
        \Ensure answer of $m$-dim range query
        \State $//$ generate a set of $m$-dim range queries
        \State $Q(q)=\left\{\wedge\left(a_{j},\left[\alpha_{j}, \beta_{j}\right] \text { or }\overline{\left[\alpha_{j}, \beta_{j}\right]}\right) \mid a_{j} \in A\right\}$
        \State $//$ associated 2-dim queries’ answers
        \For {$\_, \{a_j, a_k\}_{j \neq k}$ in enumerate(pairwise attributes)}
        \State $Q(q^{(x,y)})=\left\{\wedge\left(a_{j},\left[\alpha_{j}, \beta_{j}\right] \text { or }\overline{\left[\alpha_{j}, \beta_{j}\right]}\right) \mid a_{j} \in (a_x, a_y)\right\}$
        \State $f_{q^{(x,y)}}(Q(q^{(x,y)})) = T^{(x,y)}.\operatorname{frequency}( Q(q^{(x,y)}))$
        \EndFor
        \State $//$ Step 3: Maximum Entropy Optimization
        \State $\text { maximize }-\sum_{g \in Q(q)} f_{q}(g) \cdot \log \left(f_{q}(g)\right)$ 
    \State return $f_q(R_{\bigcap {[\alpha_j,\beta_j]}_{j=1}^{m}})$
	\end{algorithmic}
\end{algorithm}

The extending process contains three steps as follows. 
\begin{itemize}[leftmargin=*]
    \item \mypara{Step 1: Building Block Construction} \myahead groups all attributes in pairs to form $C_m^2$ 2-dim attribute pairs. 
    Then, \myahead estimates the frequency distributions for the 2-dim attribute pairs separately, \ie, a total of $C_m^2$ 2-dim trees. 

    \item \mypara{Step 2: Consistency on Attributes}
    \myahead achieves consistency on all $m$ attributes among the related 2-dim trees. 
    For example, the attribute $a$ is involved in $(m-1)$ attribute pairs. 
    Assume these $(m-1)$ 2-dim trees are $\{T_1, T_2, T_3, \cdots, T_{m-1}\}$ and each tree has $\ell$ layers except for the root node. 
    For an integer $ k \in [1, B^{\ell/2}]$, we define $f_{T_i}(a, \ell, k)$ to be the sum of frequencies of $T_i$ nodes in level $\ell$, whose specified sub-domain corresponds to $a$ is in $[(k-1)\cdot\frac{|D|}{B^{\ell/2}}+1, k\cdot\frac{|D|}{B^{\ell/2}}]$. 
    To make all $f_{T_i}(a, \ell, k)$ consistent, \myahead calculates their weighted average as $f(a, \ell, k) = \sum_{i=1}^{m-1}\lambda_{i} \cdot f_{T_i}(a, \ell, k)$, where $\lambda_i$ is the weight of $f_{T_i}(a, \ell, k)$ and $\sum_{i=1}^{m-1}\lambda_i = 1$. 
    Then, we should select the values of weight $\lambda$ to minimize the variance of $f(a, \ell, k)$, \ie, $\myvar[f(a, \ell, k)] = \sum_{i=1}^{m-1}\lambda_i^2\cdot \myvar[f_{T_i}(a, \ell, k)]$, where $\myvar[f_{T_i}(a, \ell, k)]$ is the total variance of the nodes used for the calculation of $f_{T_i}(a, \ell, k)$. 
    Based on the analysis in \cite{yang2020answering}, when the weights are inversely proportional to the variance of the estimates, $\myvar[f(a, \ell, k)]$ achieves the minimum. 
    Thus, the optimal weight $\lambda_i = \frac{1}{\myvar[f_{T_i}(a, \ell, k)]}/\sum_{i=1}^{m-1}\frac{1}{\myvar[f_{T_i}(a, \ell, k)]}$. 
    We should update each involved node in $T_i$ by adding the amount of change $(f(a, \ell, k) - f_{T_i}(a, \ell, k))/B^{\ell/2}$ to make each $f_{T_i}(a, \ell, k)$ equal to $f(a, \ell, k)$. 
    Based on the analysis in \cite{qardaji2014priview}, the consistency process makes $\{T_1, T_2, T_3, \cdots, T_{m-1}\}$ agree on attribute $a$ without changing the frequency distributions of other attributes. 
    Thus, following any order of these attributes, \myahead can achieve consistency on all $m$ attributes. 
    Noting that each 2-dim tree has $\ell$ layers, \myahead needs to conduct the above process for all $\ell$ layers. 
    
    \item \mypara{Step 3: Maximum Entropy Optimization}
    The problem we faced is to estimate the frequency of the $m$-dim query with partial information from 2-dim queries. We adopt the principle of Maximum Entropy \cite{qardaji2014priview}. 
    Specifically, for an $m$-dim range query $q$, we can define a set of range queries as $$Q(q)=\left\{\wedge\left(a_{j},\left[\alpha_{j}, \beta_{j}\right] \text { or }\overline{\left[\alpha_{j}, \beta_{j}\right]}\right) \mid a_{j} \in A\right\}\text{,}$$
    where $[\alpha_t, \beta_t]$ represents the query interval and $\overline{\left[\alpha_{t}, \beta_{t}\right]}$ is the complement of it. 
    For $2^m$ queries $g \in Q(q)$, we define $f_q(g)$ as the set of answers for queries in $Q(q)$. 
    Similarly, for 2-dim scenarios, we can obtain the query set 
    $$Q(q^{(j,k)})=\left\{\left(a_{j},\left[\alpha_{j}, \beta_{j}\right] \text { or }\overline{\left[\alpha_{j}, \beta_{j}\right]}\right) \wedge \left(a_{k},\left[\alpha_{k}, \beta_{k}\right] \text { or }\overline{\left[\alpha_{k}, \beta_{k}\right]}\right)\right\}\text{,}$$
    and answer set $f_{q^{(j,k)}}$. 
    For any query $g^{(j,k)} \in Q(q^{(j,k)})$, we use $f_{q^{(j,k)}}(g^{(j,k)})$ to denote its answer. 
    In particular, for a $g^{(j,k)} \in Q(q^{(j,k)})$,  $f_{q}(g^{(j,k)})$ means $g^{(j,k)}$'s answer constructed from $f_q$ by summing up the answers of the associated queries in $Q(q)$. 
    Then we can formulate the following optimization problem: 
$$
\begin{array}{ll}
\text { maximize } & -\sum_{g \in Q(q)} f_{q}(g) \cdot \log \left(f_{q}(g)\right) \\
\text { subject to } & \forall_{g \in Q(q)} f_{q}(g) \geq 0 \\
& \forall_{a_j, a_k \in A} \forall_{g^{(j, k)} \in Q\left(q^{(j, k)}\right)} f_{q^{(j, k)}}(g^{(j, k)})=f_{q}(g^{(j, k)})
\end{array}
$$
    The above optimization problem can be addressed by an off-the-shelf convex optimization tool. 
    To solve the frequency estimation problem more efficiently, \myahead can adopt the Weighted Update \cite{yang2020answering}, which achieves almost the same accuracy as the Maximum Entropy. 
\end{itemize}

\section{Validation of Threshold Choice}
\label{Validation of Threshold Choice}
\begin{figure*}[ht]
    \centering
    \subfigure[Loan, $|D| = 256$, vary $\theta$]{
    \includegraphics[width=0.23\hsize]{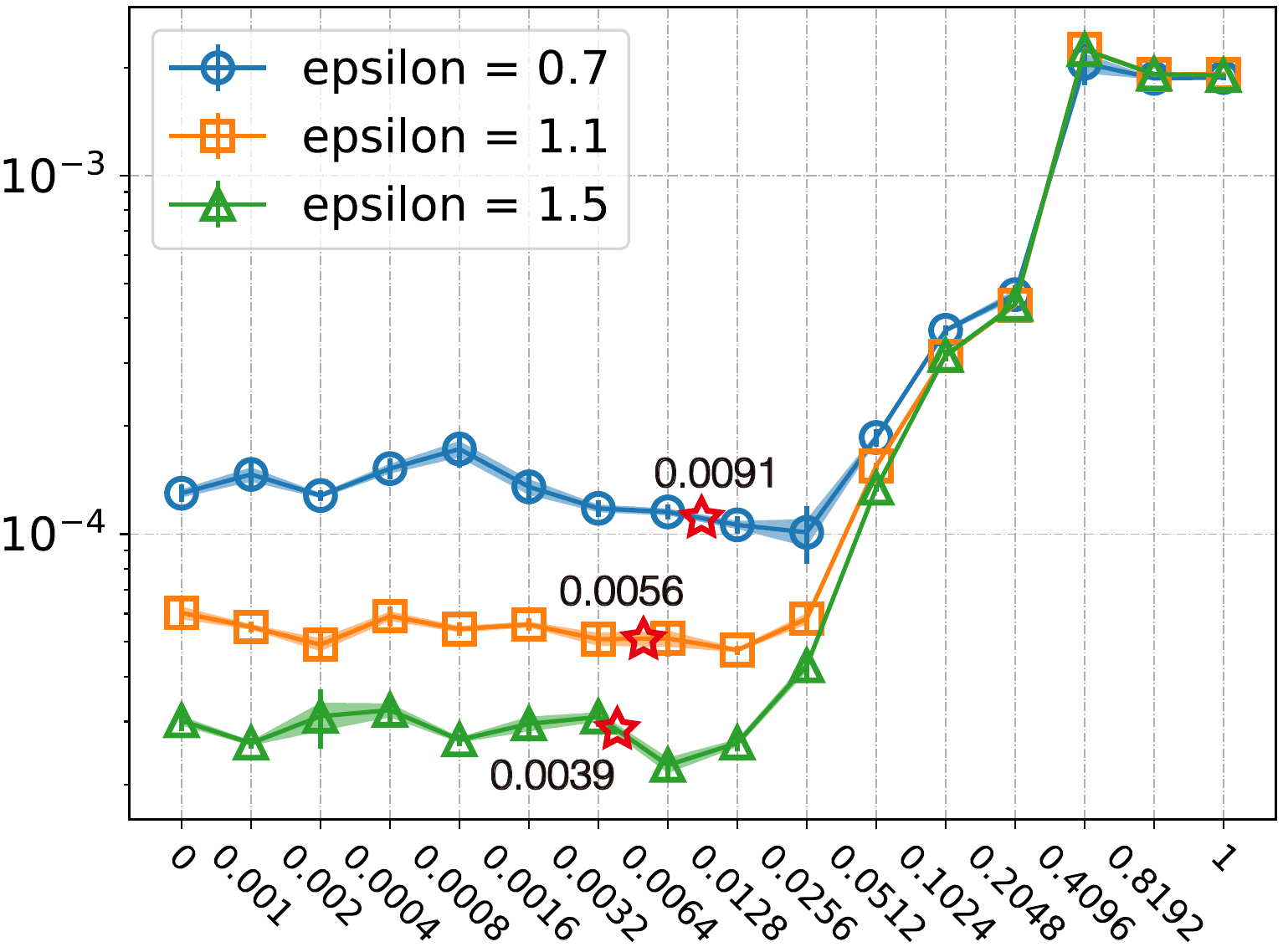}
    }
    \subfigure[Financial, $|D| = 512$, vary $\theta$]{
    \includegraphics[width=0.23\hsize]{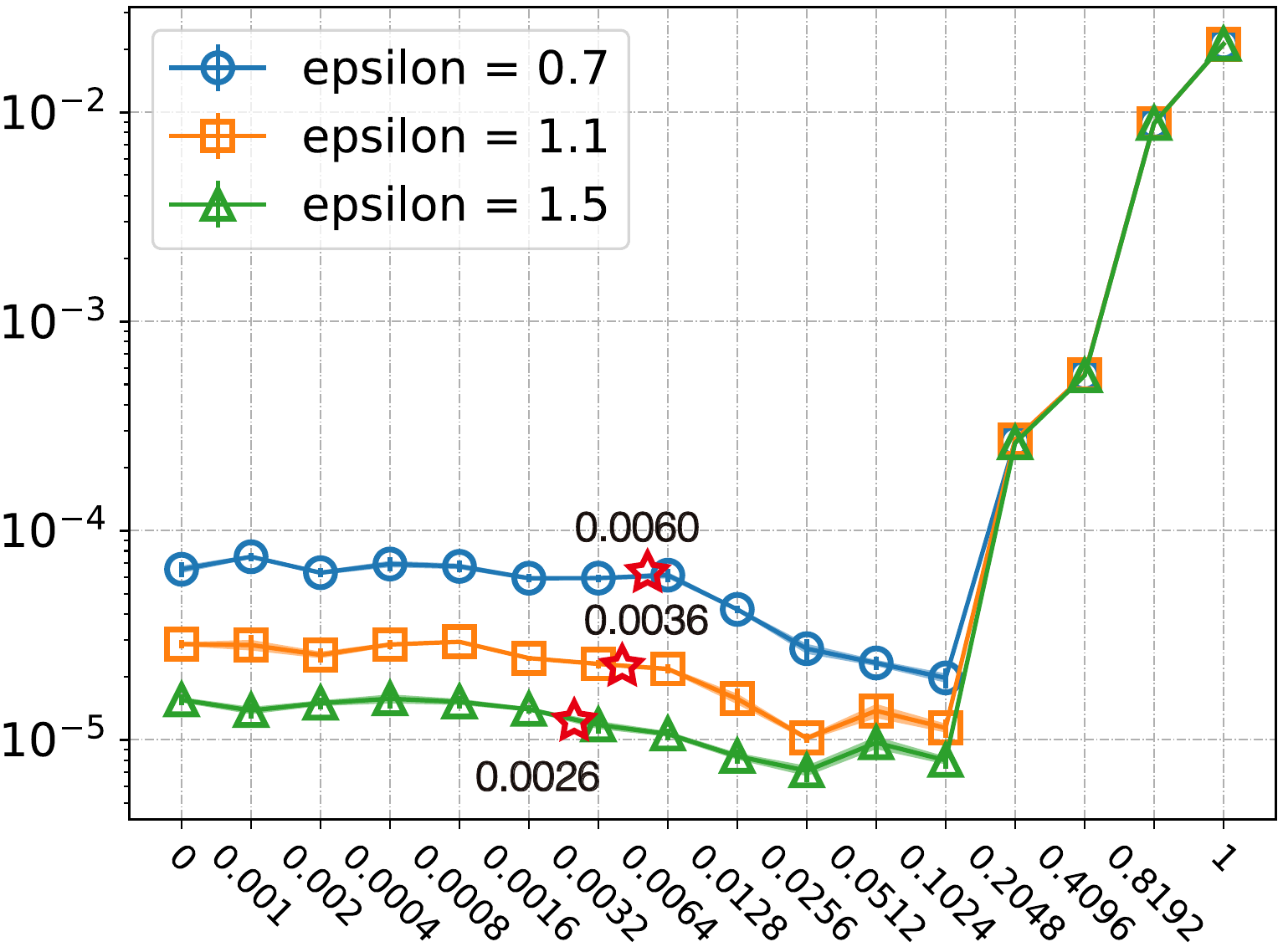}
    }
    \subfigure[BlackFriday, $|D| = 512$, vary $\theta$]{
    \includegraphics[width=0.23\hsize]{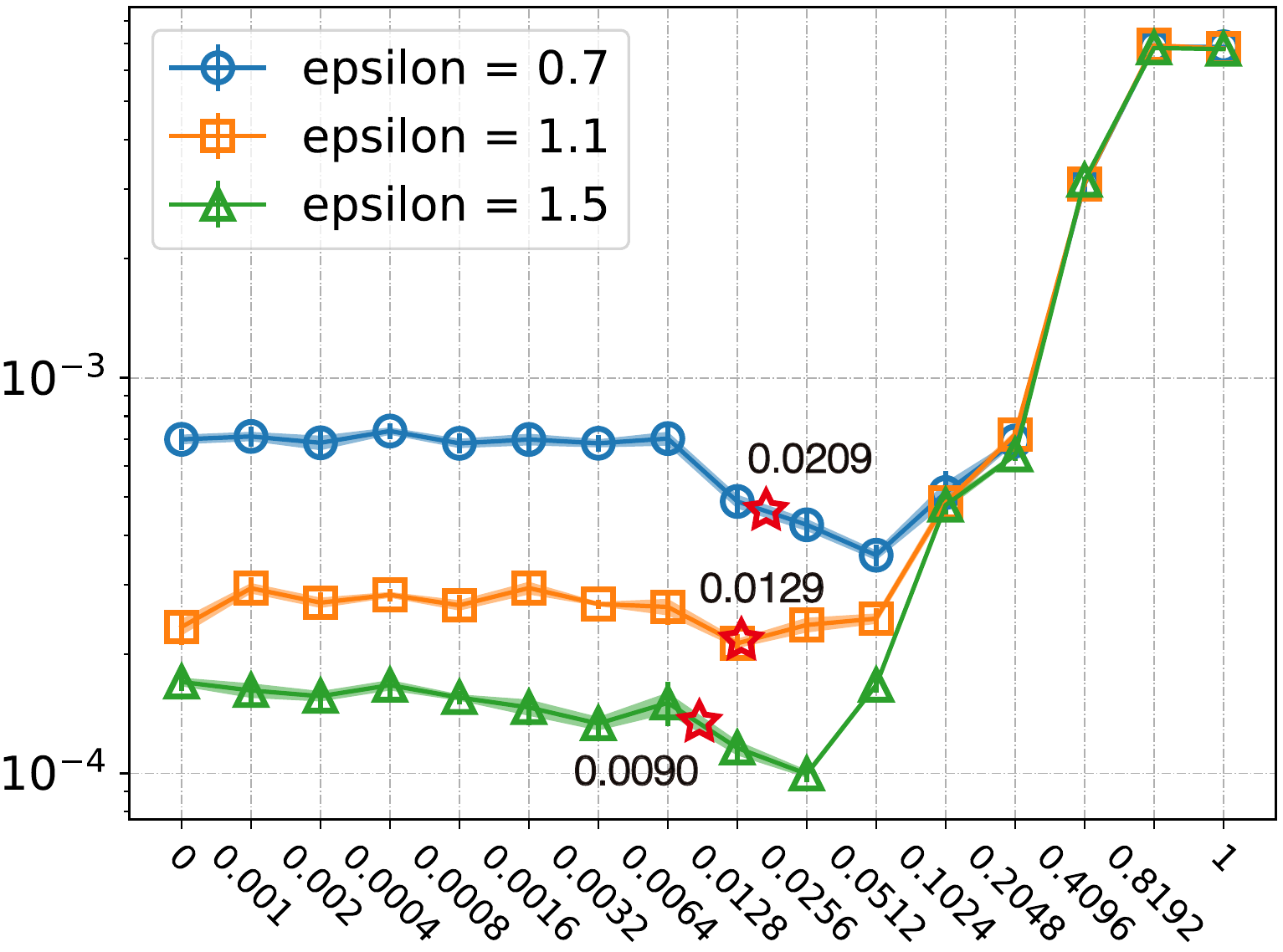}
    }
    \subfigure[Salaries, $|D| = 1024$, vary $\theta$]{
    \includegraphics[width=0.23\hsize]{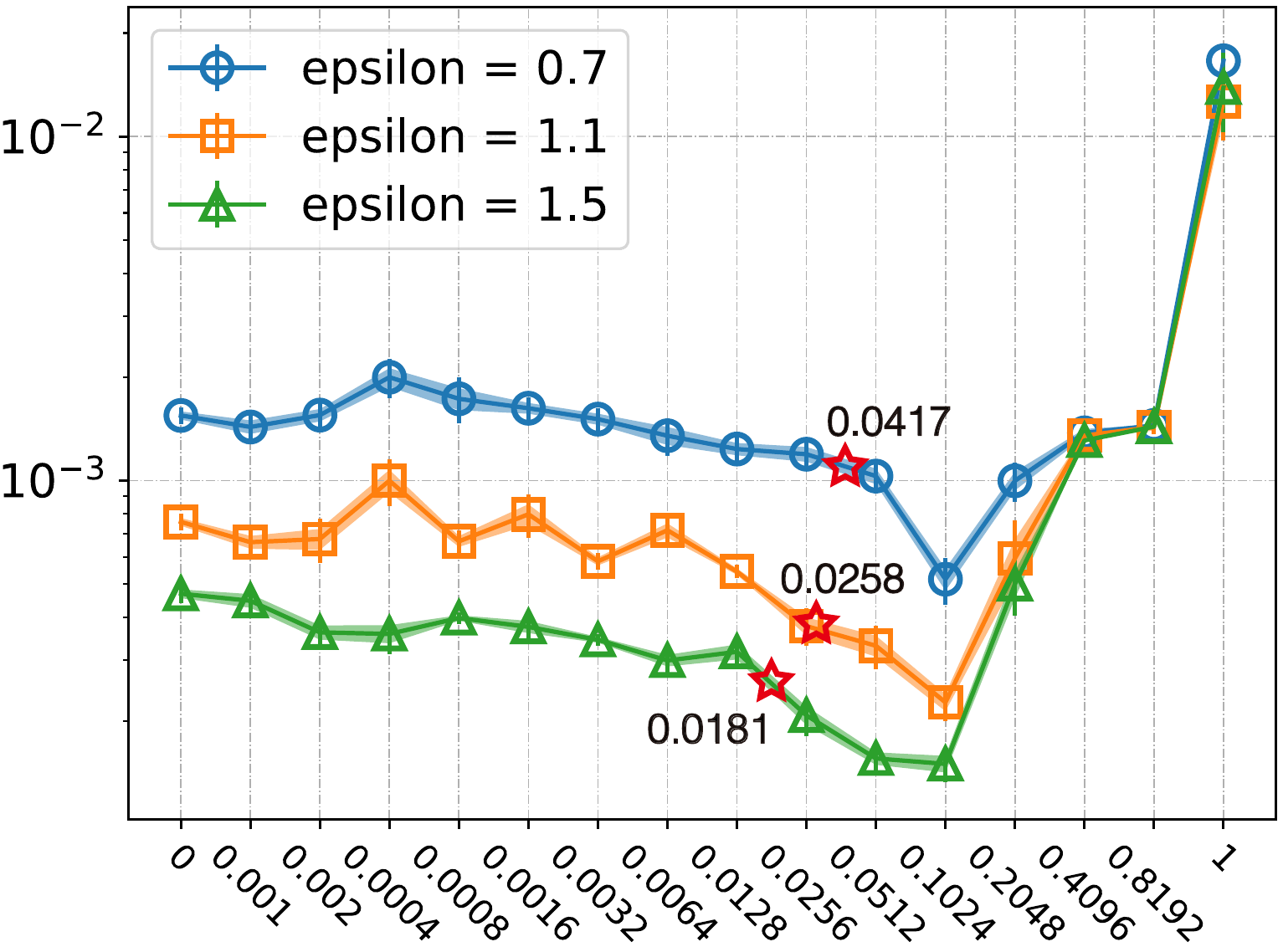}
    }
    \vspace{-0.4cm}
    \caption{Impact of threshold setting under various privacy budgets. The red star corresponds to our theoretical threshold obtained from \autoref{theta setting1}. The results are shown in log scale. }
    \vspace{-0.3cm}
    \label{SNRvalid}
\end{figure*}

\begin{figure*}[ht]
    \centering
    \subfigure[Loan, $|D| = 256$, $\theta = 0.006$]{\label{fig:leaf_Distribution-Loan}
    \includegraphics[width=0.235\hsize]{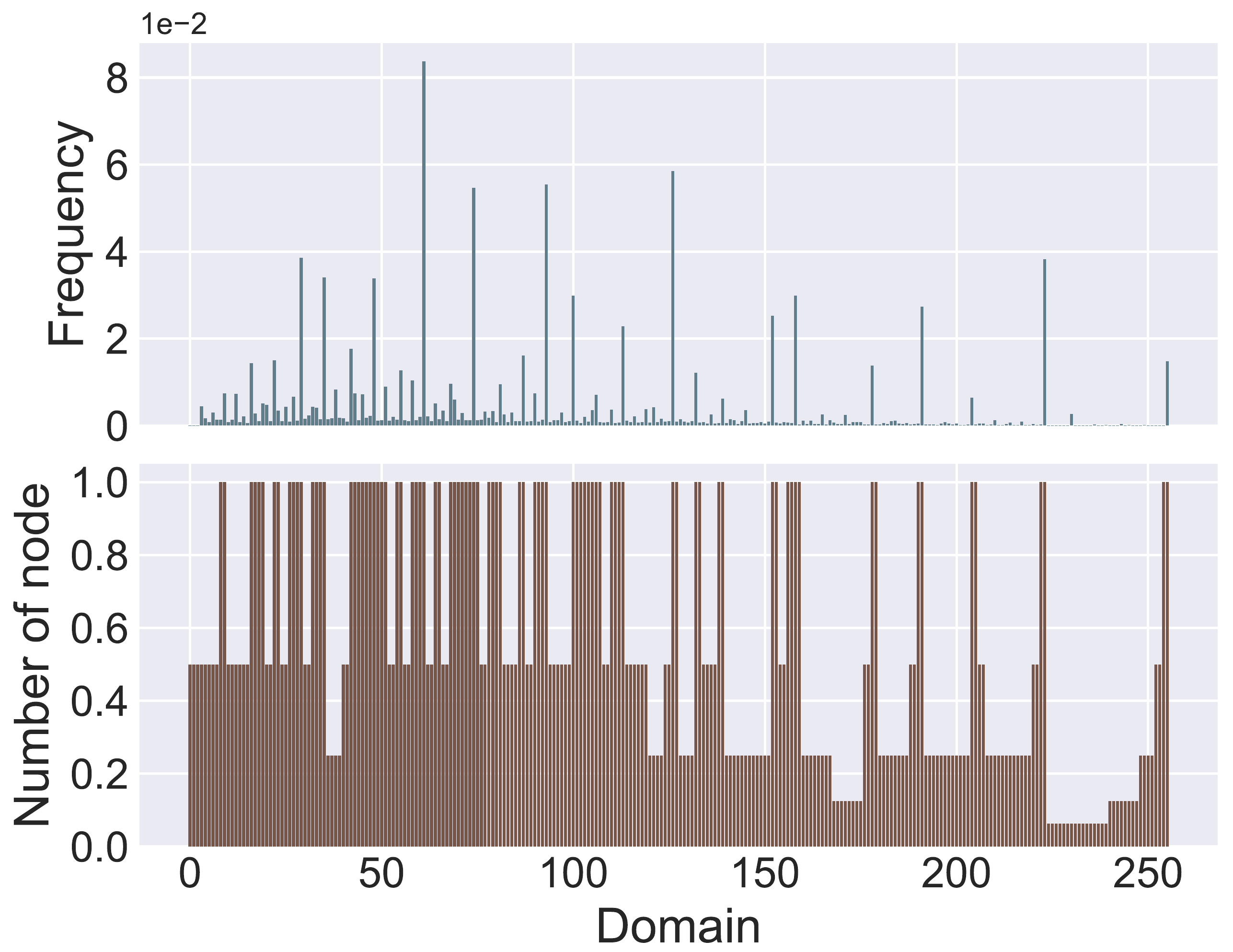}
    }
    \subfigure[Financial, $|D| = 512$, $\theta = 0.004$]{\label{fig:leaf_Distribution-Financial}
    \includegraphics[width=0.235\hsize]{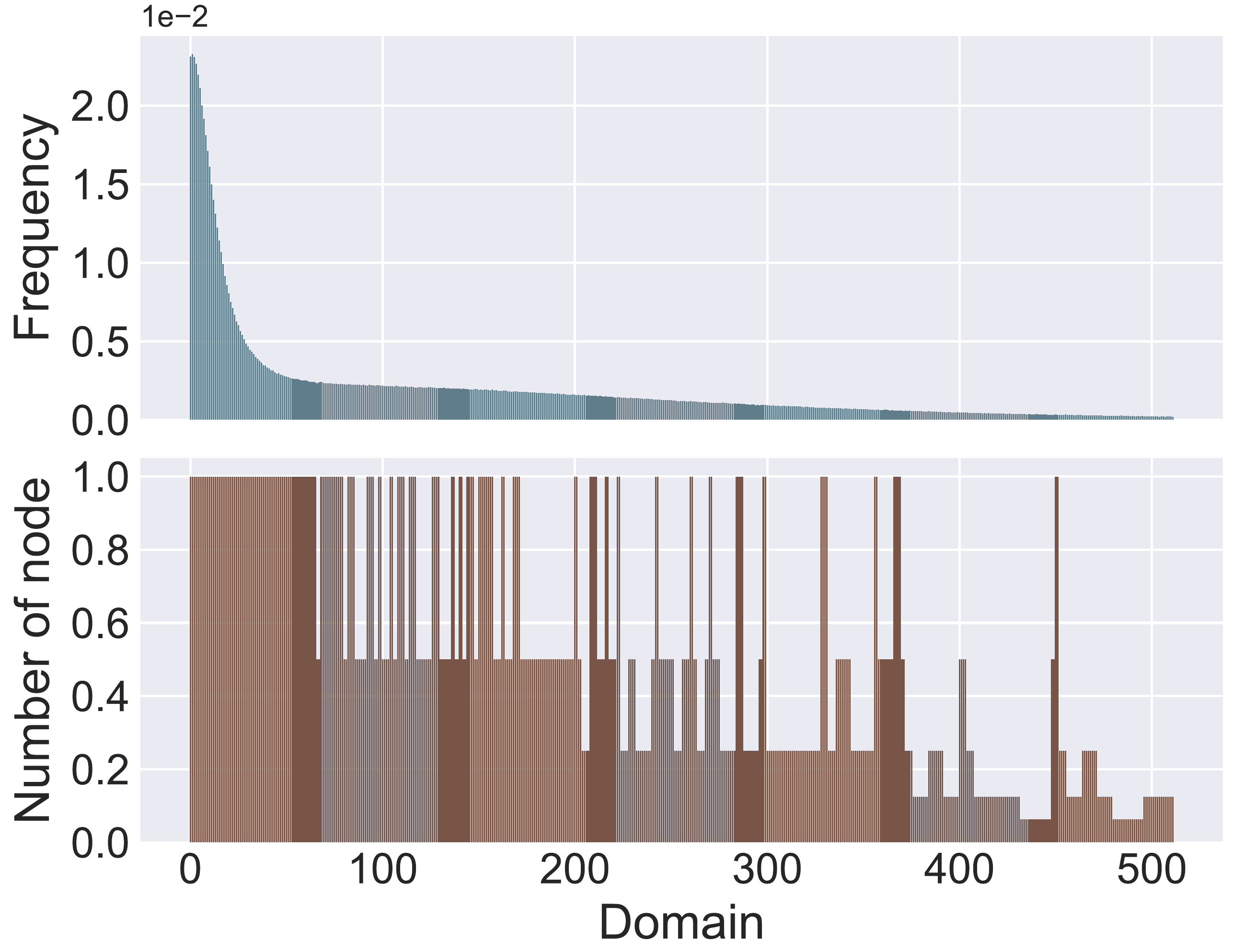}
    }
    \subfigure[BlackFriday, $|D| = 1024$, $\theta = 0.013$]{\label{fig:leaf_Distribution-BlackFriday}
    \includegraphics[width=0.235\hsize]{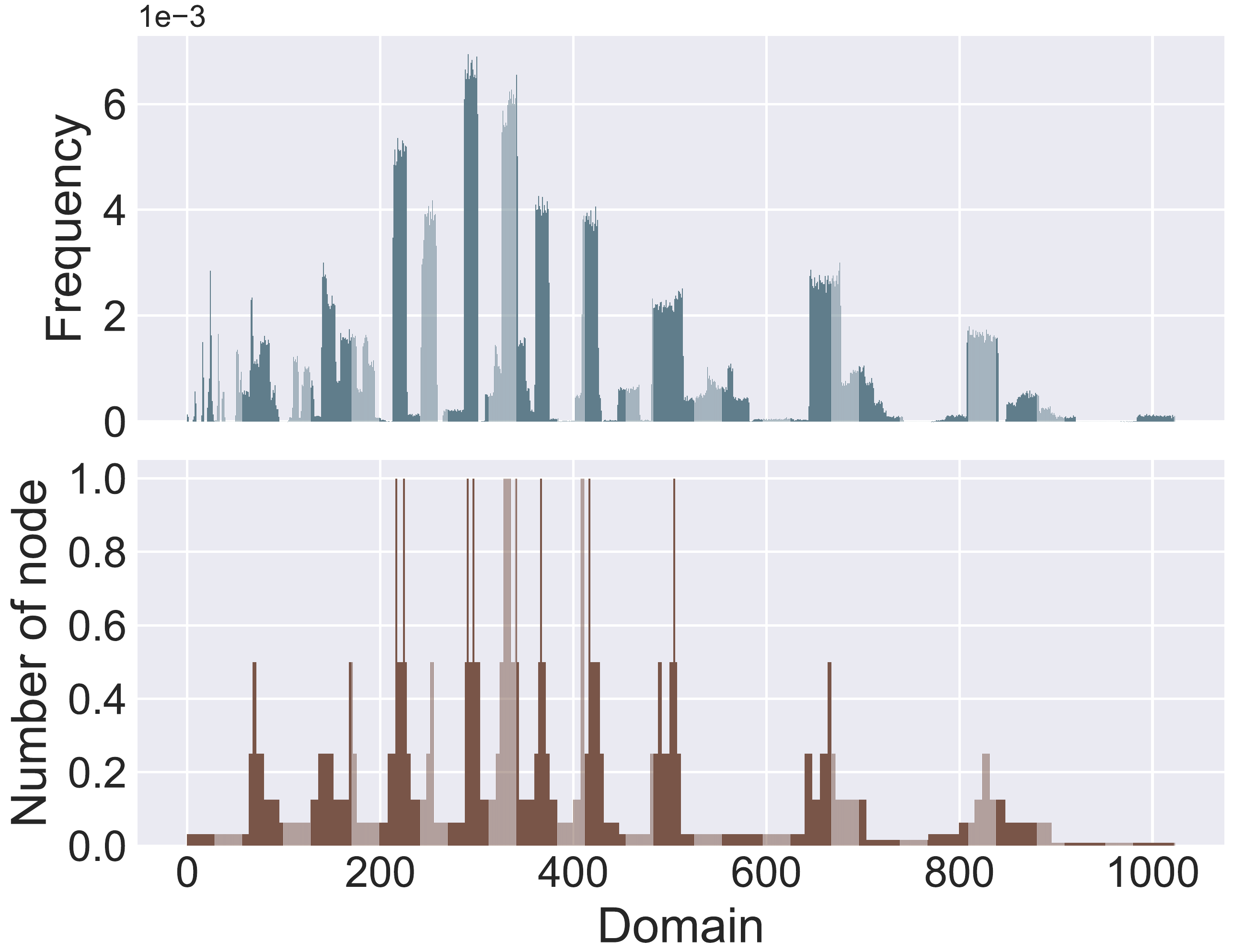}
    }
    \subfigure[Salaries, $|D| = 2048$, $\theta = 0.026$]{\label{fig:leaf_Distribution-Salaries}
    \includegraphics[width=0.235\hsize]{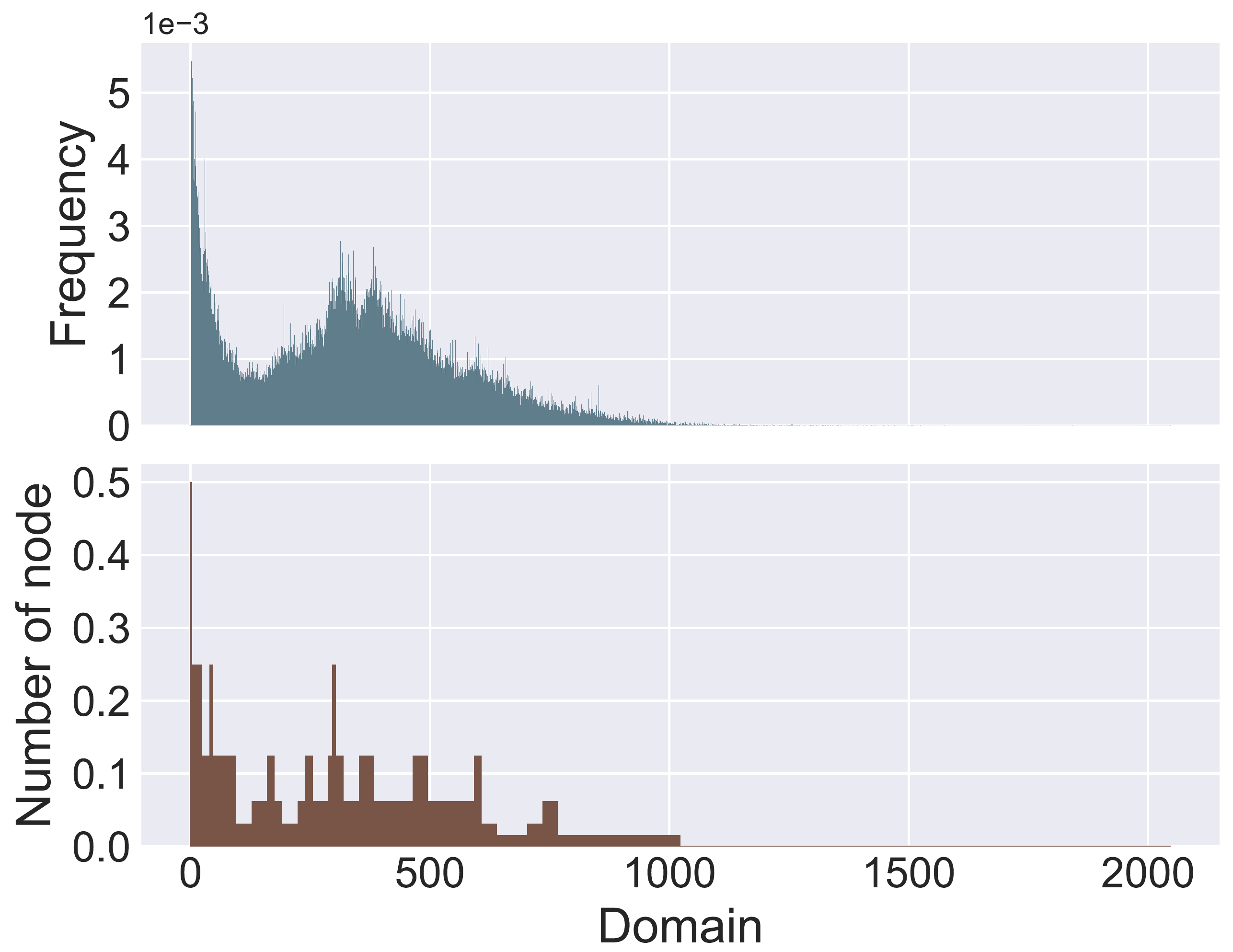}
    }
    \vspace{-0.4cm}
    \caption{The distribution of the number of leaves in \myahead in 1-dim scenes. The top of each picture shows the true frequency distribution, and the bottom shows the corresponding distribution of the number of leaves with $\epsilon = 1.1$. }
    \vspace{-0.2cm}
    \label{1-dim node distribution}
\end{figure*}

\begin{figure*}[ht]
    \centering
    \subfigure[\Laplacian, $|D| = 256^2$, $\theta = 0.003$]{\label{fig:leaf_Distribution-Laplacian256}
    \includegraphics[width=0.235\hsize]{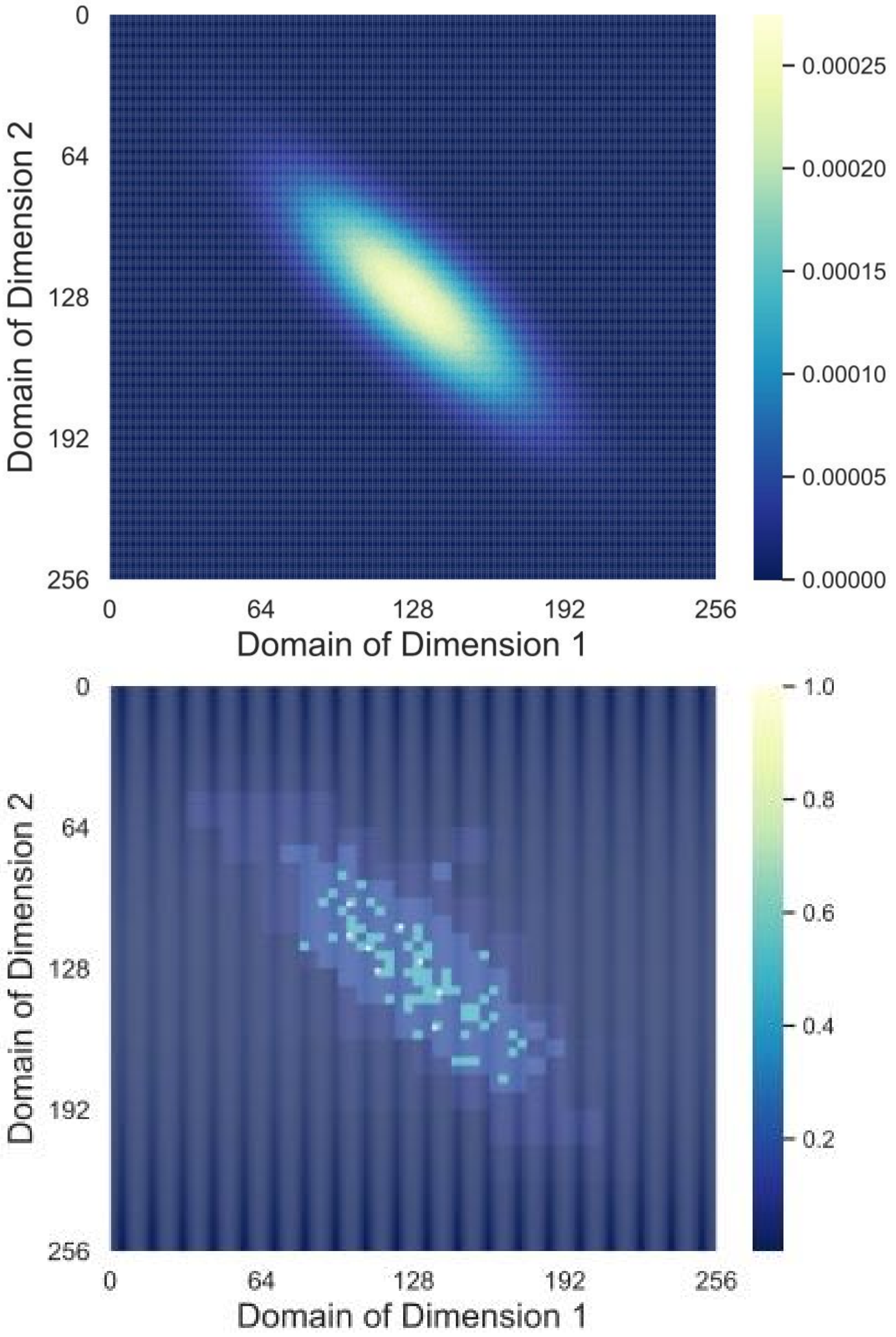}
    }
    \subfigure[\Laplacian, $|D| = 1024^2$, $\theta = 0.004$]{\label{fig:leaf_Distribution-Laplacian1024}
    \includegraphics[width=0.235\hsize]{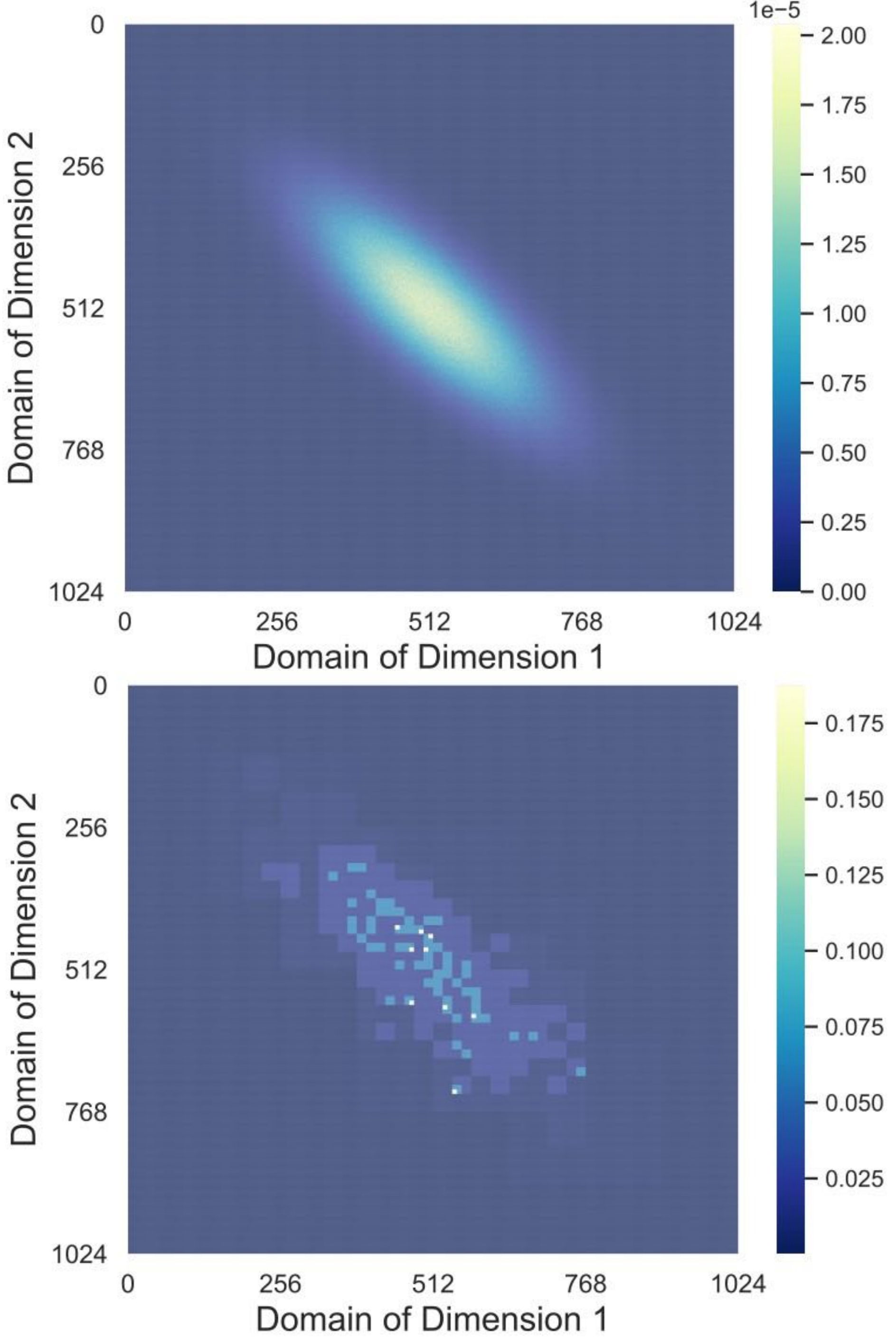}
    }
    \subfigure[\Gaussian, $|D| = 256^2$, $\theta = 0.003$]{\label{fig:leaf_Distribution-guassian256}
    \includegraphics[width=0.235\hsize]{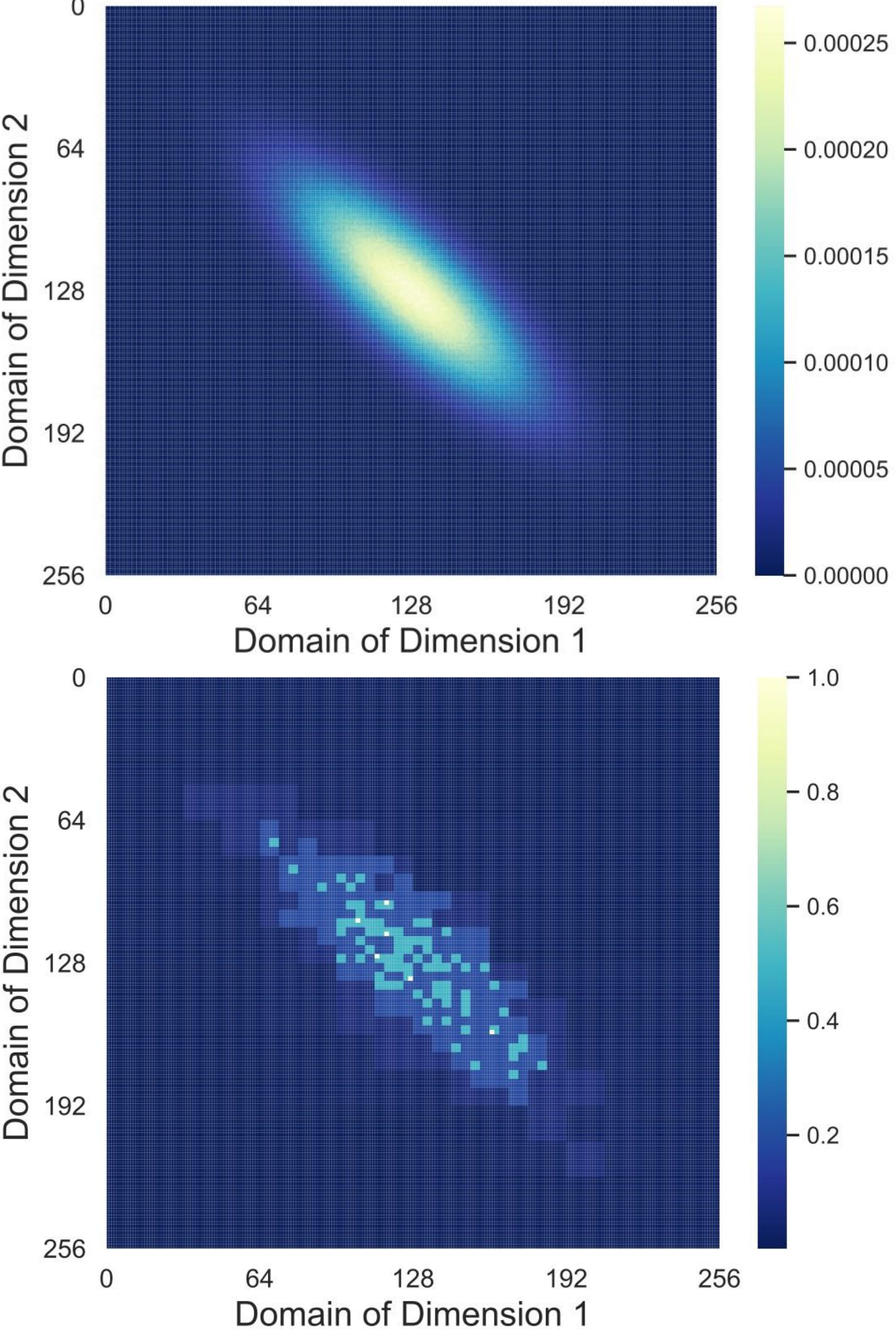}
    }
    \subfigure[\Gaussian, $|D| = 1024^2$, $\theta = 0.004$]{\label{fig:leaf_Distribution-guassian1024}
    \includegraphics[width=0.235\hsize]{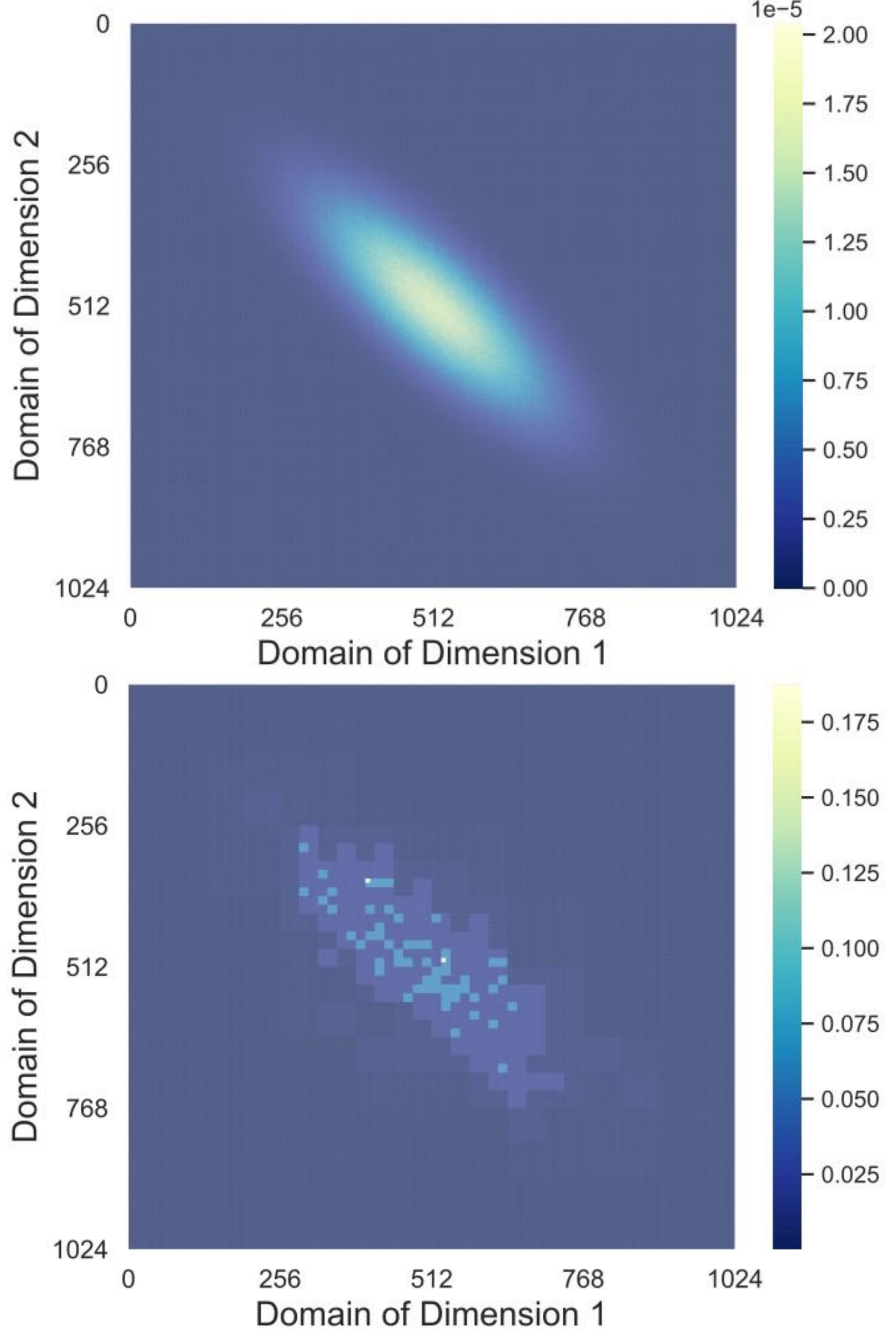}
    }
    \vspace{-0.4cm}
    \caption{The distribution of the number of leaves in \myahead in 2-dim scenes. The top of each picture shows the true frequency distribution, and the bottom shows the corresponding distribution of the number of leaves with $\epsilon = 1.1$. }
    \vspace{-0.2cm}
    \label{2-dim node distribution}
\end{figure*}

Recalling the analysis in \autoref{Selection parameters}, by setting a threshold, we do not divide the sub-domains whose frequencies are smaller than the threshold. 
A reasonable threshold setting can balance the noise error and non-uniform error thus minimizing the overall estimation error.

Here, we omit the post-processing step in \myahead to highlight the effect of the threshold values. 
\autoref{SNRvalid} shows the impact of different threshold settings in \myahead. 
The horizontal axis of each plot represents the threshold from 0 to 1. 
The red star on each line is the threshold value obtained by \autoref{theta setting1} in \autoref{Selection parameters}.

From \autoref{SNRvalid}, we have the following observations consistent with our analysis. 
1) A large threshold will cause the MSE to increase significantly. 
The reason is that a larger threshold would make the estimated frequency distribution closer to the uniform distribution. In this way, the non-uniform error will dominate the estimation error (recall \autoref{method_2}) thus degrading the overall accuracy. 
On the Loan dataset, when the threshold value is larger than 0.4096, the MSE becomes more than $2 \times 10^{-3}$, which is ten times of the minimum MSE. 
2) A small threshold has little impact on MSE, \eg, on the BlackFriday dataset, the MSE hardly changes when the threshold is less than 0.0064. 
In this case, the MSE is mainly caused by the noise error, which is related to the privacy budget. 
3) On different datasets, the optimal experimental threshold values are also various. 
For example, on the BlackFriday dataset, the optimal experimental threshold is 0.0256 for $\epsilon = 1.5$. 
While on the Salaries dataset, the optimal experimental threshold is 0.1024. 
4) It is worth noting that the theoretical $\theta$ is not exactly the empirically observed optimum in some cases. 
In the derivation of $\theta$, we use the true frequency value $f$. 
However, in the practical implementation of \myahead, we only have access to the estimated frequency value $\hat{f}$, \ie, the true frequency value $f$ with a random noise variable $X$. 
Since \oue is an unbiased protocol, the expected value of $\hat{f}$ is equal to $f$, making the theoretical $\theta$ close to the empirically observed optimum.

We also provide the distribution of the number of leaves in \autoref{1-dim node distribution} and \autoref{2-dim node distribution}. 
The upper part of each sub-graph is the true frequency distribution, and the following is the corresponding node number distribution of \myahead with $\epsilon = 1.1$. 
From \autoref{1-dim node distribution} and \autoref{2-dim node distribution}, the node number distributions are almost the same with the true frequency distributions, which confirms the rationality of the threshold selection in \autoref{Selection parameters}. 
More specifically, when the frequency of the sub-domain is large, \myahead further divides the sub-domain, \ie, more nodes for estimation, to reduce the non-uniform error. 
Otherwise, \myahead does not decompose the sub-domain, \ie, less nodes for estimation, to suppress the noise error. 
In addition, from \autoref{overall compare Salaries} and \autoref{fig:leaf_Distribution-Salaries},  
when the frequencies are more uniformly distributed across nodes in most subtrees (the nodes fell into the interval $[1024, 2048]$), \myahead will behave better owing to smaller non-uniform errors (refer to more details in \autoref{Privacy and Utility Analysis}).

\section{Dataset Description} 
\label{Dataset Description}

\begin{table}[!ht]
\caption{Correlation between attributes of Salaries. }
\begin{tabular}{|c|ccccc|}
\hline
                 & $a_1$              & $a_2$                       & $a_3$                      & $a_4$                 & $a_5$  \\ \hline
$a_1$ & \multicolumn{1}{c|}{1}       & \multicolumn{1}{c|}{-0.6638} & \multicolumn{1}{c|}{-0.3631} & \multicolumn{1}{c|}{0.2344} & 0.2344   \\ \cline{2-6} 
$a_2$          & \multicolumn{1}{c|}{-0.6638} & \multicolumn{1}{c|}{1}       & \multicolumn{1}{c|}{0.0208}  & \multicolumn{1}{c|}{0.2442} & 0.2442   \\ \cline{2-6} 
$a_3$         & \multicolumn{1}{c|}{-0.3631} & \multicolumn{1}{c|}{0.0208}  & \multicolumn{1}{c|}{1}       & \multicolumn{1}{c|}{0.4665} & 0.4665   \\ \cline{2-6} 
$a_4$      & \multicolumn{1}{c|}{0.2344}  & \multicolumn{1}{c|}{0.2442}  & \multicolumn{1}{c|}{0.4665}  & \multicolumn{1}{c|}{1}      & 1        \\ \cline{2-6} 
$a_5$         & \multicolumn{1}{c|}{0.2344}  & \multicolumn{1}{c|}{0.2442}  & \multicolumn{1}{c|}{0.4665}  & \multicolumn{1}{c|}{1}      & 1        \\ \hline
\end{tabular}

\label{table: Correlation between Attributes of Salaries.}
\vspace{-0.6cm}
\end{table}
\begin{table}[!ht]
\caption{Correlation between attributes of BlackFriday. }
\begin{tabular}{|c|ccccc|}
\hline
                     & $a_1$              & $a_2$                       & $a_3$                      & $a_4$                 & $a_5$ \\ \hline
$a_1$            & \multicolumn{1}{c|}{1}       & \multicolumn{1}{c|}{-0.0508} & \multicolumn{1}{c|}{-0.0030} & \multicolumn{1}{c|}{-0.0062} & 0.0083     \\ \cline{2-6} 
$a_2$ & \multicolumn{1}{c|}{-0.0508} & \multicolumn{1}{c|}{1}       & \multicolumn{1}{c|}{-0.0552} & \multicolumn{1}{c|}{-0.4010} & -0.3749    \\ \cline{2-6} 
$a_3$ & \multicolumn{1}{c|}{-0.0030} & \multicolumn{1}{c|}{-0.0552} & \multicolumn{1}{c|}{1}       & \multicolumn{1}{c|}{0.0774}  & 0.0566     \\ \cline{2-6} 
$a_4$ & \multicolumn{1}{c|}{-0.0062} & \multicolumn{1}{c|}{-0.4010} & \multicolumn{1}{c|}{0.0774}  & \multicolumn{1}{c|}{1}       & 0.3239     \\ \cline{2-6} 
$a_5$          & \multicolumn{1}{c|}{0.0083}  & \multicolumn{1}{c|}{-0.3749} & \multicolumn{1}{c|}{0.0566}  & \multicolumn{1}{c|}{0.3239}  & 1          \\ \hline
\end{tabular}
\label{table: Correlation between Attributes of BlackFriday.}
\end{table}
\begin{table}[!ht]
\caption{Correlation between attributes of Loan. }
\begin{tabular}{|c|ccccc|}
\hline
                  & $a_1$              & $a_2$                       & $a_3$                      & $a_4$                 & $a_5$ \\ \hline
$a_1$ & \multicolumn{1}{c|}{1}      & \multicolumn{1}{c|}{0.0018} & \multicolumn{1}{c|}{0.9431} & \multicolumn{1}{c|}{0.2398} & 0.1519            \\ \cline{2-6} 
$a_2$         & \multicolumn{1}{c|}{0.0018} & \multicolumn{1}{c|}{1}      & \multicolumn{1}{c|}{0.0460} & \multicolumn{1}{c|}{0.0436} & 0.0481            \\ \cline{2-6} 
$a_3$       & \multicolumn{1}{c|}{0.9431} & \multicolumn{1}{c|}{0.0460} & \multicolumn{1}{c|}{1}      & \multicolumn{1}{c|}{0.2633} & 0.1651            \\ \cline{2-6} 
$a_4$      & \multicolumn{1}{c|}{0.2398} & \multicolumn{1}{c|}{0.0436} & \multicolumn{1}{c|}{0.2633} & \multicolumn{1}{c|}{1}      & 0.9612            \\ \cline{2-6} 
$a_5$ & \multicolumn{1}{c|}{0.1519} & \multicolumn{1}{c|}{0.0481} & \multicolumn{1}{c|}{0.1651} & \multicolumn{1}{c|}{0.9612} & 1                 \\ \hline
\end{tabular}

\label{table: Correlation between Attributes of Loan.}
\vspace{-0.6cm}
\end{table}
Here, we provide a detailed description of the datasets used in our evaluation. 
\begin{itemize}
    \item Salaries\footnote{https://www.kaggle.com/kaggle/sf-salaries}: This dataset is about San Francisco city employee salary data on an annual basis from 2011 to 2014. 
    It contains 148,654 records and 13 attributes, where we select 5 attributes as shown in \autoref{table: Correlation between Attributes of Salaries.}. 
    \item BlackFriday\footnote{https://www.kaggle.com/roshansharma/black-friday}: This dataset is a sample of the transaction records in a retail store, who wants to know  the customer purchase behavior against different products. 
    It contains 537,577 records and 12 attributes, where we select 5 attributes as shown in \autoref{table: Correlation between Attributes of BlackFriday.} 
    \item Loan\footnote{https://www.kaggle.com/wordsforthewise/lending-club}: This dataset provides the complete loan data of Lending Club for all loans issued through 2007-2018. 
    It contains 2,260,668 records and 150 attributes, where we select 5 attributes as shown in \autoref{table: Correlation between Attributes of Loan.} 
    \item Financial\footnote{https://www.kaggle.com/ntnu-testimon/paysim1}: This synthetic dataset is generated by the PaySim mobile money simulator \cite{lopez2016paysim}. 
    It contains 6,362,620 records and 11 attributes, including transaction type, customer ID and transaction amount. 
    Since the distribution of transaction amount is quite skewed, 
    we truncate the data greater than 500,000. 
    After the truncation, the processed dataset has 6,022,336 records, \ie, $94.6\%$ of the original data remaining. 
    Then, we divide the range of transaction amount [0, 500000] into slots of a fixed length and bucketize the records with a domain size of 512. 
\end{itemize}

\section{High-dimensional Range Query on Real datasets}
\label{High-dimensional Range Query on Real datasets}
In this section, we evaluate \myahead on high-dimensional real-world datasets, where the correlations between the selected attributes are shown in \autoref{table: Correlation between Attributes of Salaries.}, \autoref{table: Correlation between Attributes of BlackFriday.} and \autoref{table: Correlation between Attributes of Loan.}. 
With domain size $|D|=64$ for each dimension, we show the MSE of \myahead for 2-dim, 3-dim and 5-dim range query, respectively. 
Under each setting, we set 8 different privacy budgets. 

Base on the results in \autoref{fig: high-dim real dataset}, we have the following observations which are consistent with our analysis in \autoref{Extension to Multi-dimensional Settings} and \autoref{Effectiveness of ahead for high-dimensional Range Query}. 
1) \myahead outperforms \myHDG throughout most cases. 
2) \myahead with \lle obtains lower MSEs than \de. 
3) The data utility of \myHDG changes significantly with the correlation of attributes, and becomes worse with a stronger correlation. 
For instance, as shown in \autoref{fig:Rand_ep-Loan-Ori-Domain6_Attribute2}, \autoref{fig:Rand_ep-Loan-Ori-Domain6_Attribute3} and \autoref{fig:Rand_ep-Loan-Ori-Domain6_Attribute5}, the superiority of \myahead will increase with the stronger correlation between attributes, \ie, `installment' in 3-dim and `last\_pymnt\_amnt' in 5-dim (recall \autoref{table: Correlation between Attributes of Loan.}).

\begin{figure*}[!t]
    \centering
    \subfigure[2-dim, Salaries, vary $\epsilon$]{\label{fig:Rand_ep-Salaries-Ori-Domain6_Attribute2}\includegraphics[width=0.25\hsize]{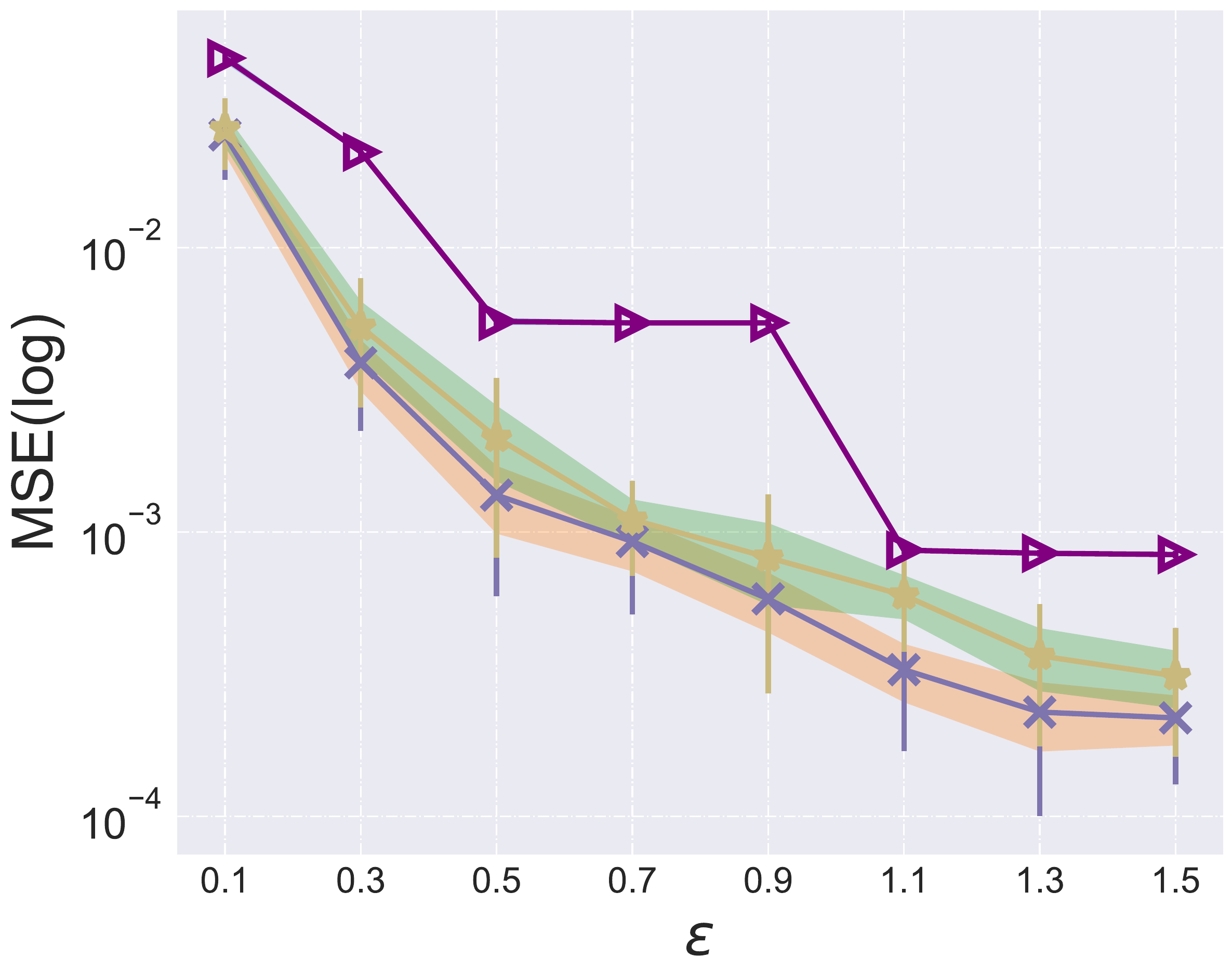}}
    \subfigure[3-dim, Salaries, vary $\epsilon$]{\label{fig:Rand_ep-Salaries-Ori-Domain6_Attribute3}\includegraphics[width=0.25\hsize]{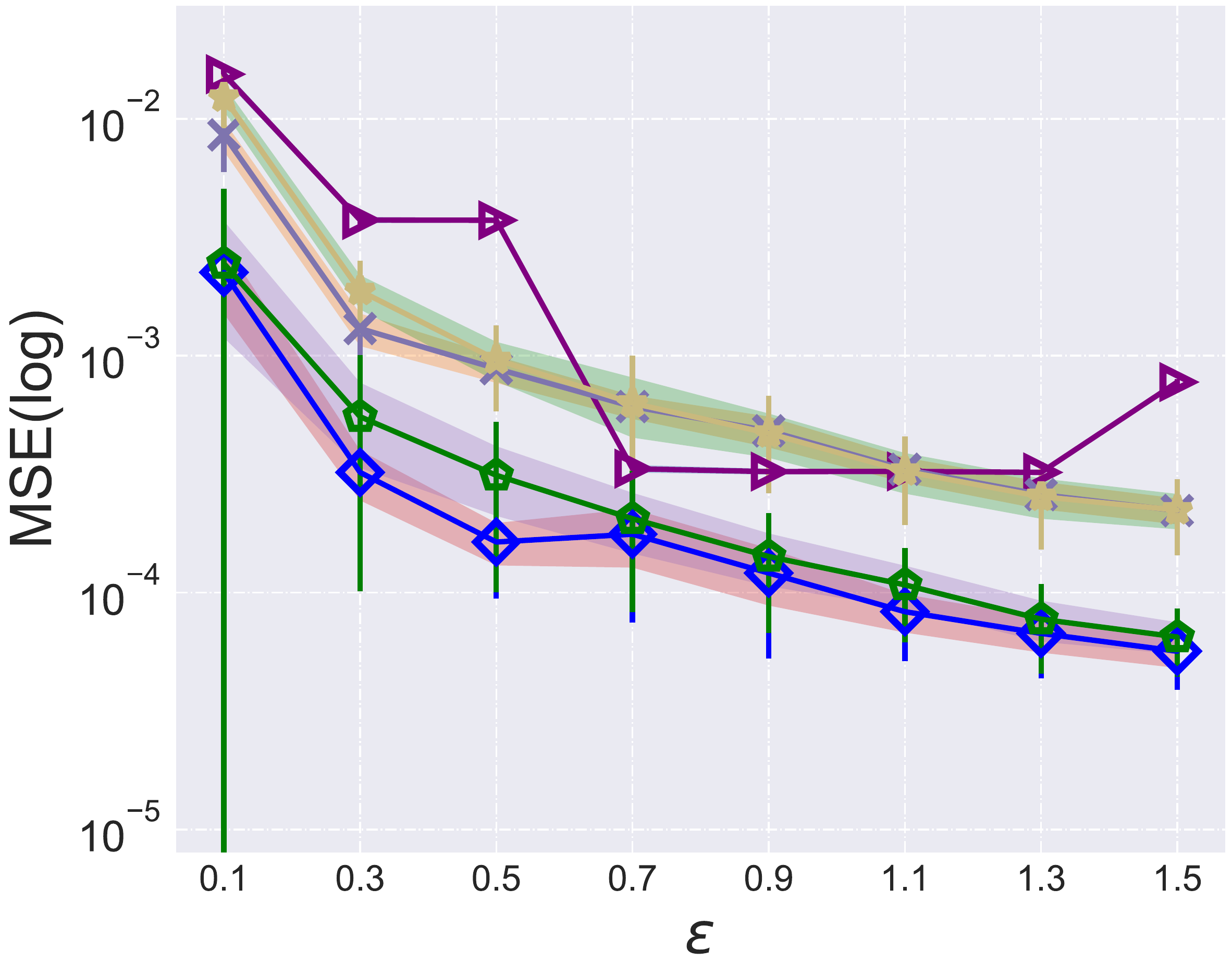}}
    \subfigure[5-dim, Salaries, vary $\epsilon$]{\label{fig:Rand_ep-Salaries-Ori-Domain6_Attribute5}\includegraphics[width=0.25\hsize]{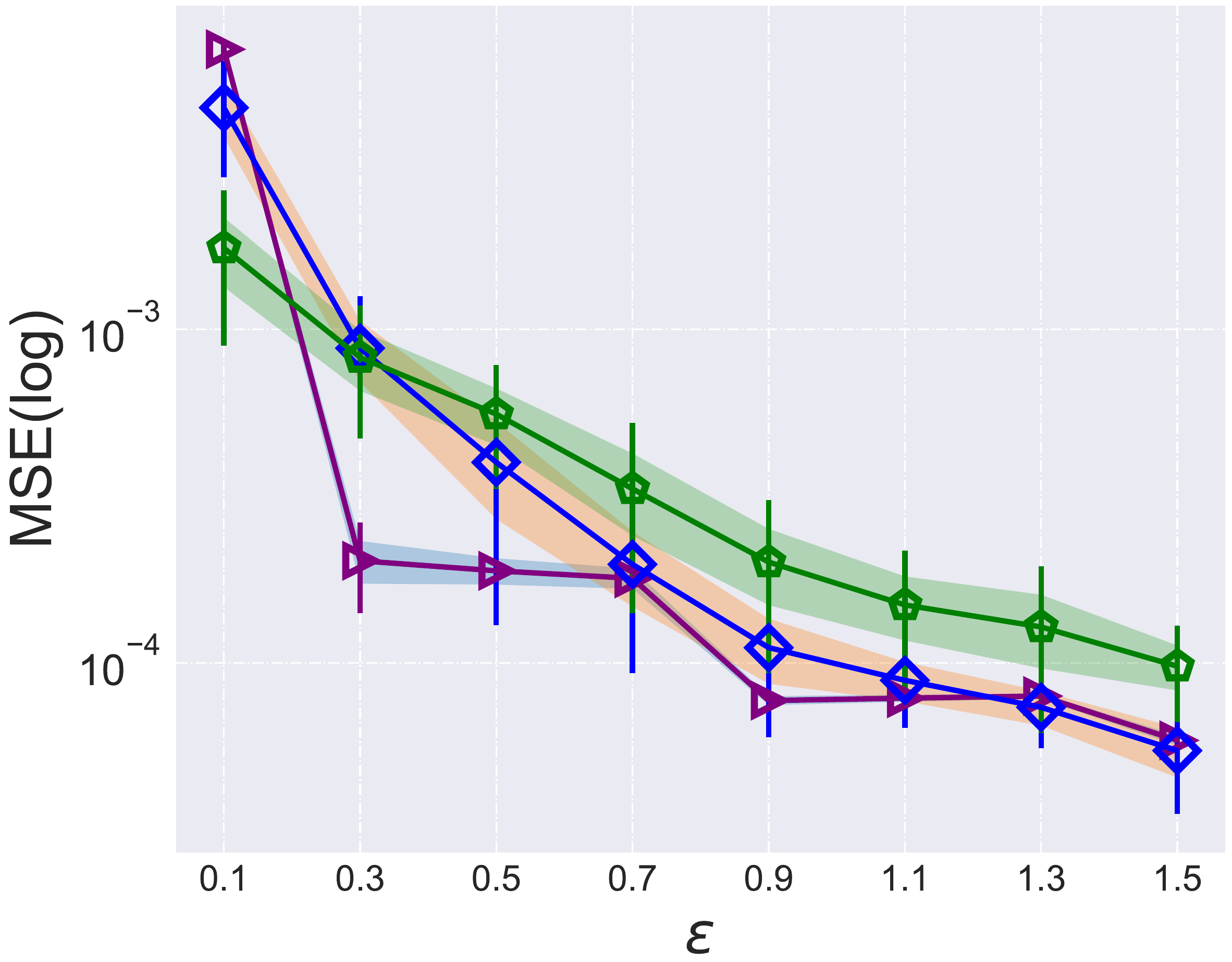}}\\[-2ex]
    \subfigure[2-dim, BlackFriday, vary $\epsilon$]{\label{fig:Rand_ep-BlackFriday-Ori-Domain6_Attribute2}\includegraphics[width=0.25\hsize]{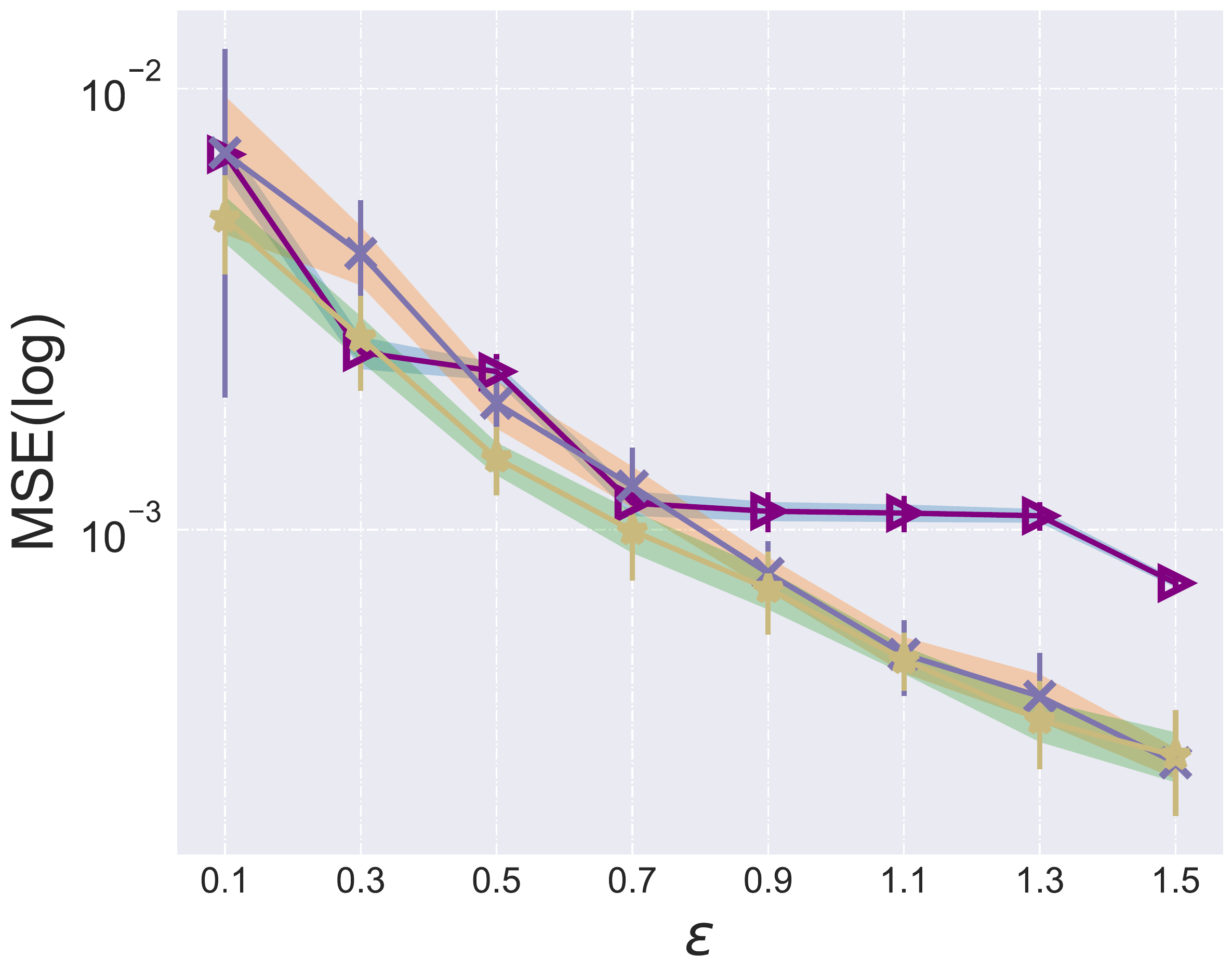}}
    \subfigure[3-dim, BlackFriday, vary $\epsilon$]{\label{fig:Rand_ep-BlackFriday-Ori-Domain6_Attribute3}\includegraphics[width=0.25\hsize]{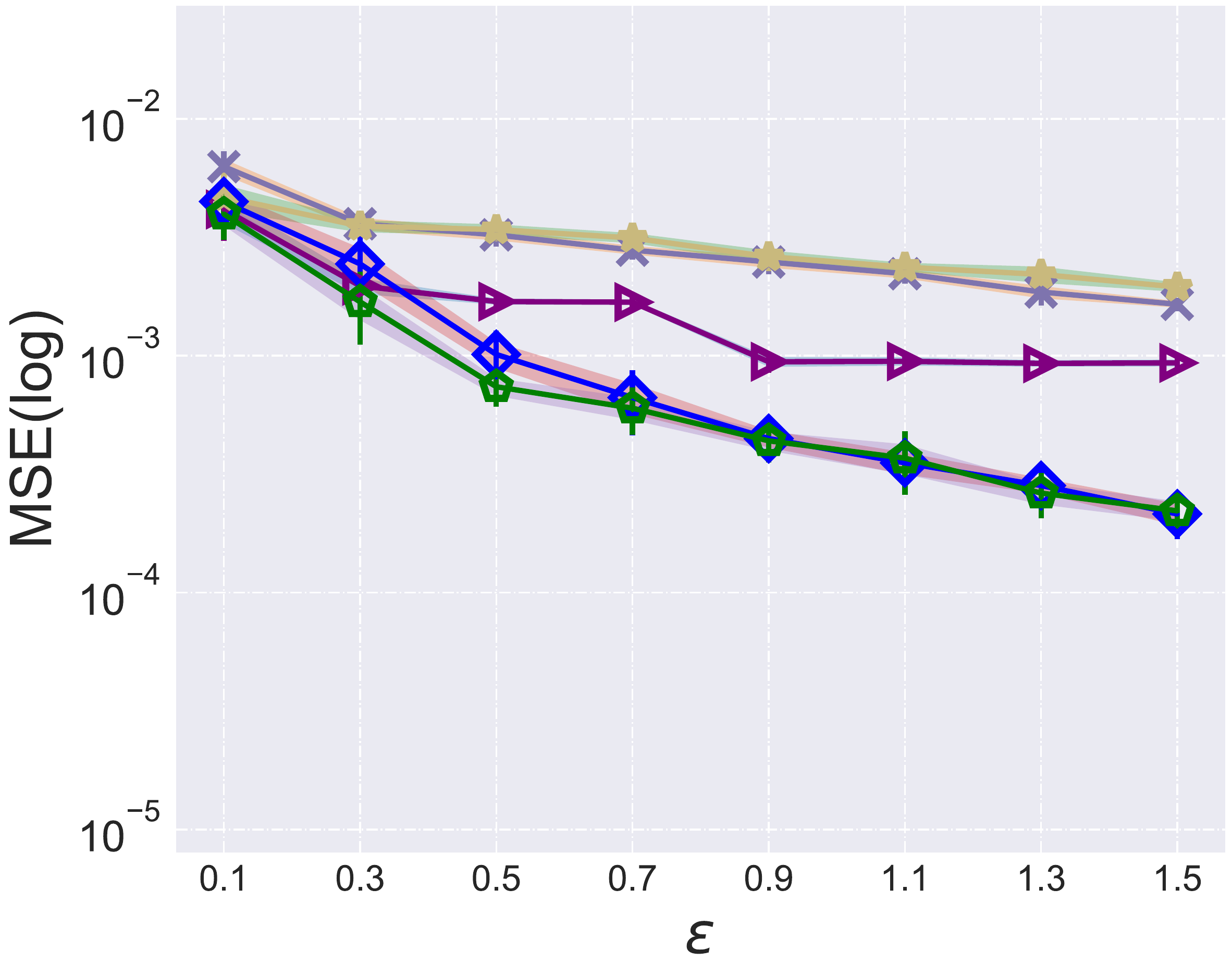}}
    \subfigure[5-dim, BlackFriday, vary $\epsilon$]{\label{fig:Rand_ep-BlackFriday-Ori-Domain6_Attribute5}\includegraphics[width=0.25\hsize]{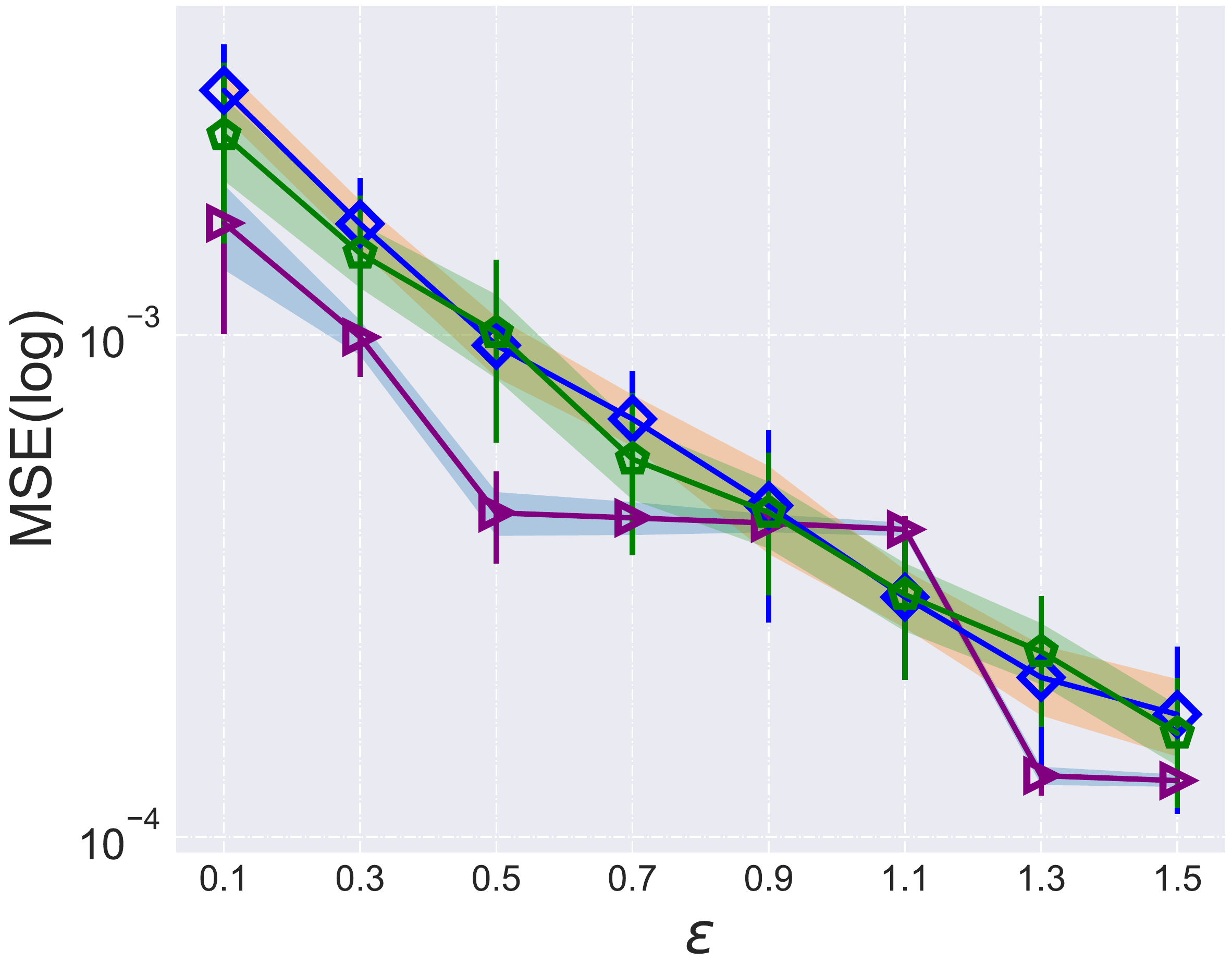}}\\[-2ex]
    \subfigure[2-dim, Loan, vary $\epsilon$]{\label{fig:Rand_ep-Loan-Ori-Domain6_Attribute2}\includegraphics[width=0.25\hsize]{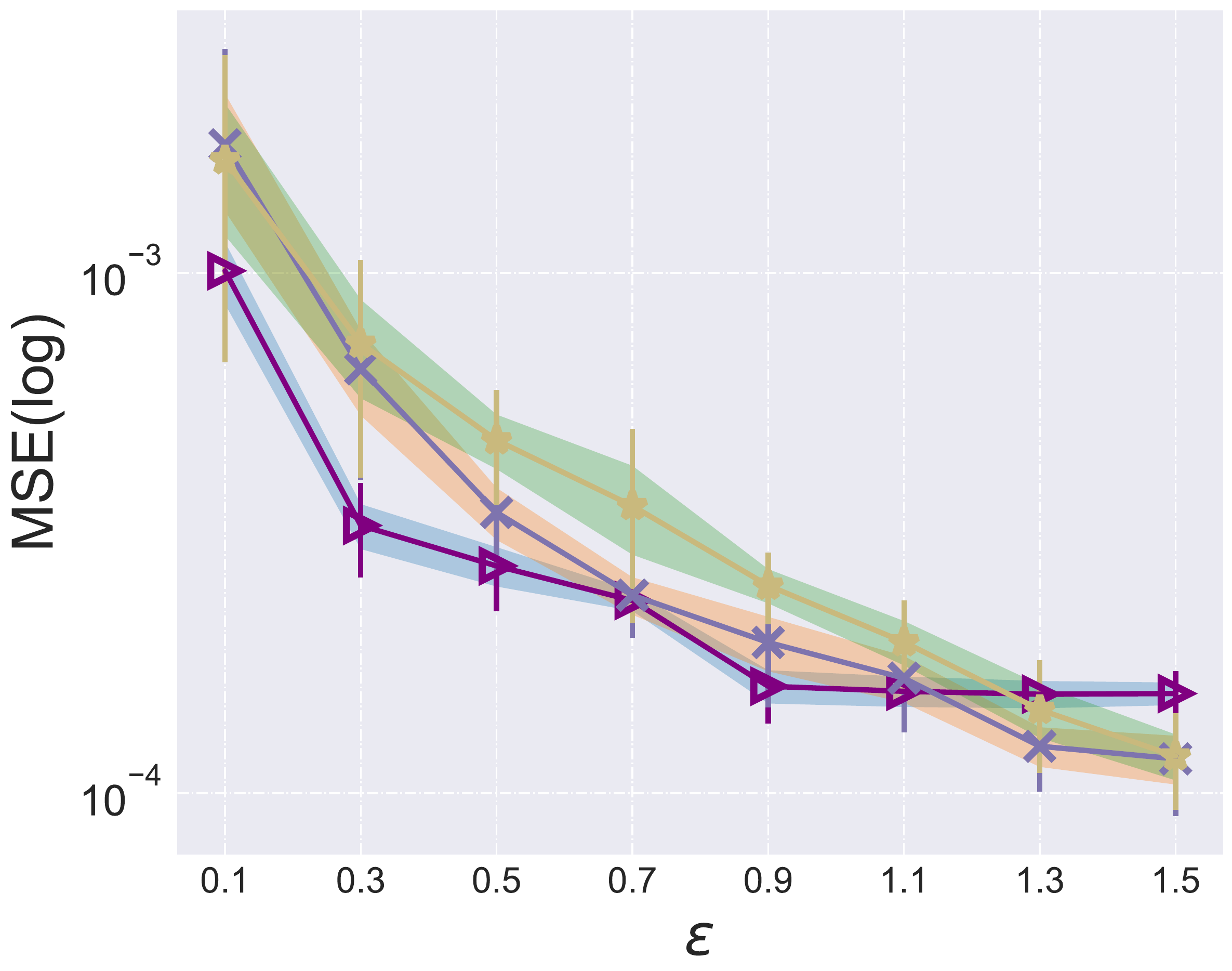}}
    \subfigure[3-dim, Loan, vary $\epsilon$]{\label{fig:Rand_ep-Loan-Ori-Domain6_Attribute3}\includegraphics[width=0.25\hsize]{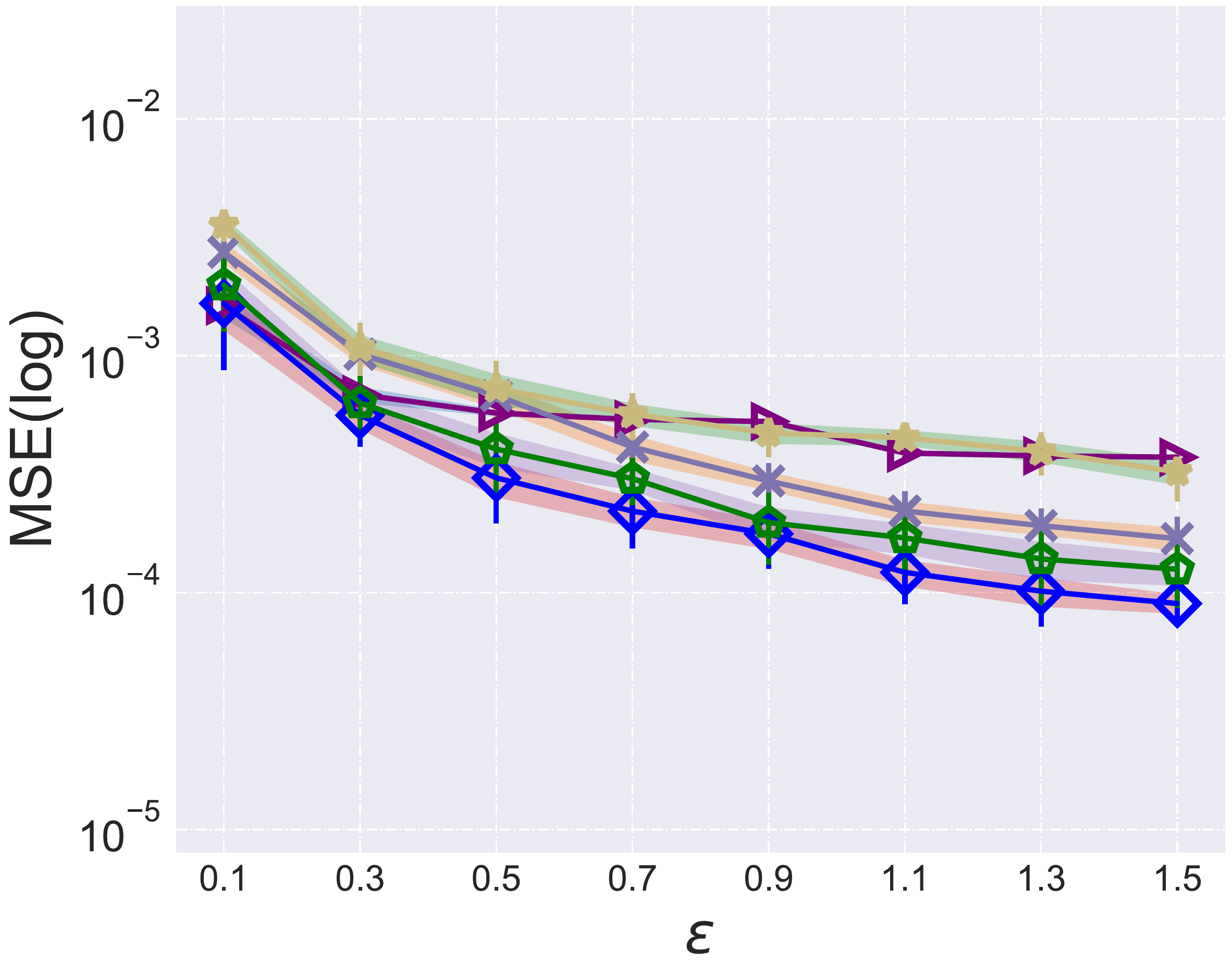}}
    \subfigure[5-dim, Loan, vary $\epsilon$]{\label{fig:Rand_ep-Loan-Ori-Domain6_Attribute5}\includegraphics[width=0.25\hsize]{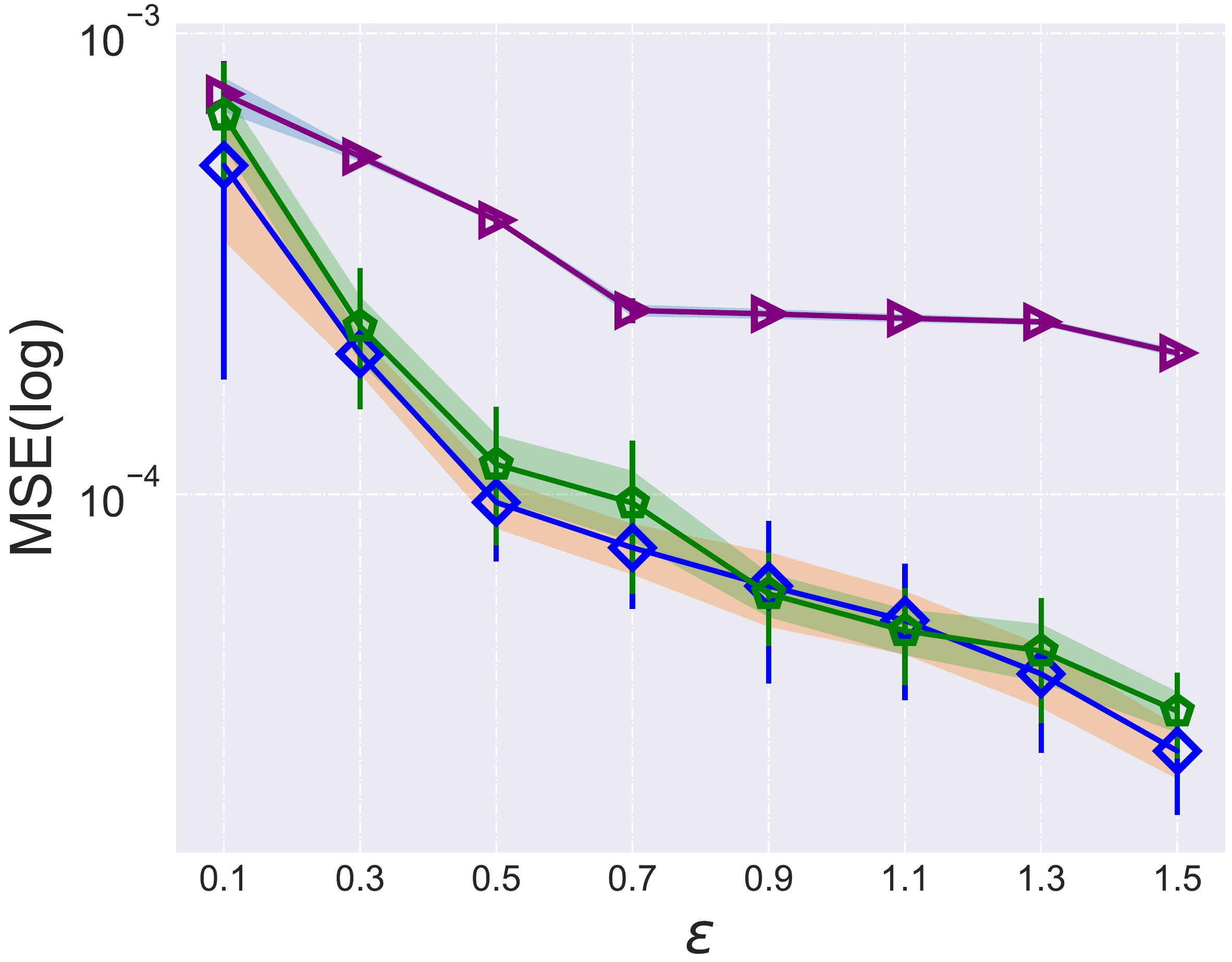}}\\[-1ex]
    \subfigure{\includegraphics[width=0.8\textwidth]{./figures_others/comparison_3D_legend5.pdf.pdf}}  
    \vspace{-0.5cm}
    \caption{Comparison of different methods on high-dimensional real datasets under various privacy budgets. \de and \lle respectively represent two high-dimensional expansion methods, \ie, ``direct estimation'' and ``leveraging low-dimensional estimation''. \myHDG is a baseline method. The results are shown in log scale.}
    \label{fig: high-dim real dataset}
\end{figure*}

\section{Practical Deployment of \myahead}
\label{Practical Deployment of AHEAD}

\label{comprehensive experiment}
In this section, we systematically analyze the impact of privacy budget, user scale, domain size, data skewness, data dimension and attribute correlation on the performance of \myahead and the state-of-the-art methods. 
When focusing on one variable, we should ensure the others are consistent. 
Therefore, the experiments are mainly conducted on the four synthetic datasets. 
The observations in our following analysis can help adopt \myahead in practice 
as well as assess existing LDP-based range query frameworks. 
\subsection{Impact of User Scale}

\begin{figure*}[!ht]
    \centering
    \subfigure[\myahead, 1-dim \Zipf, $|D|$ = 256]{
        \includegraphics[width=0.3\hsize]{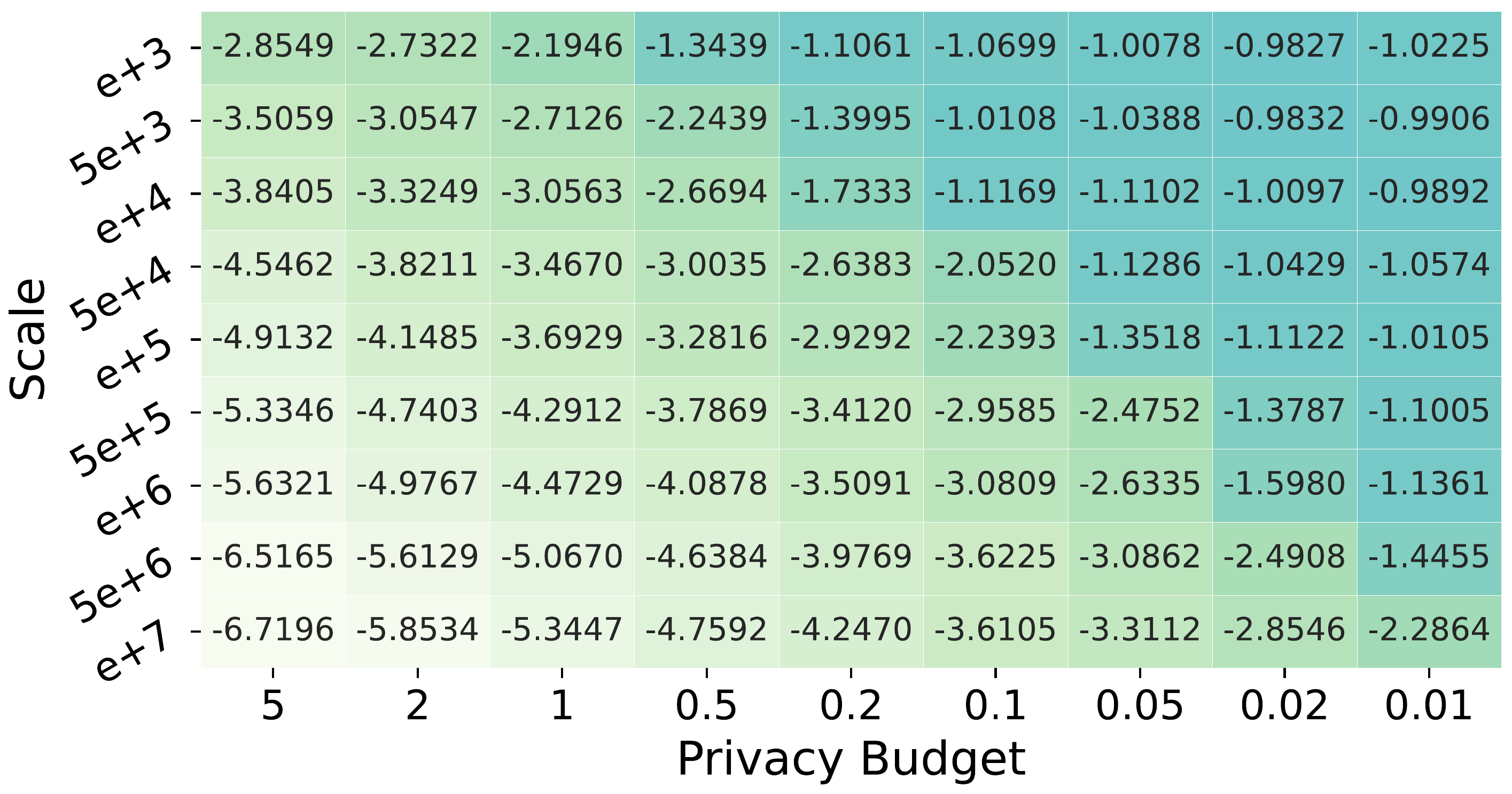}
        \label{ahead scale eps zpif}
    }
    \subfigure[\myhio, 1-dim \Zipf, $|D|$ = 256]{
        \includegraphics[width=0.3\hsize]{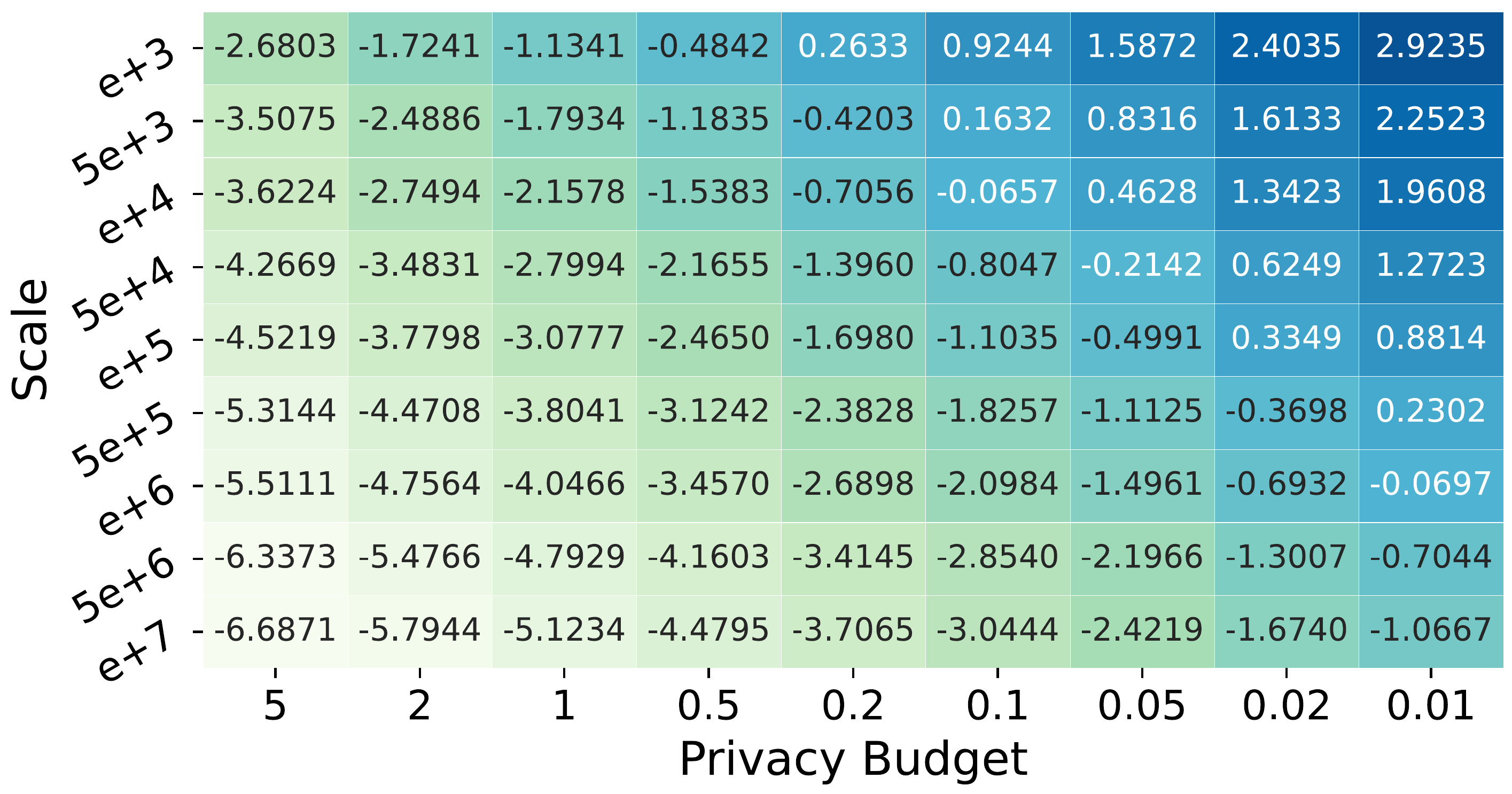}
        \label{hio scale eps zpif}
    }
    \subfigure[\mydht, 1-dim \Zipf, $|D|$ = 256]{
        \includegraphics[width=0.3\hsize]{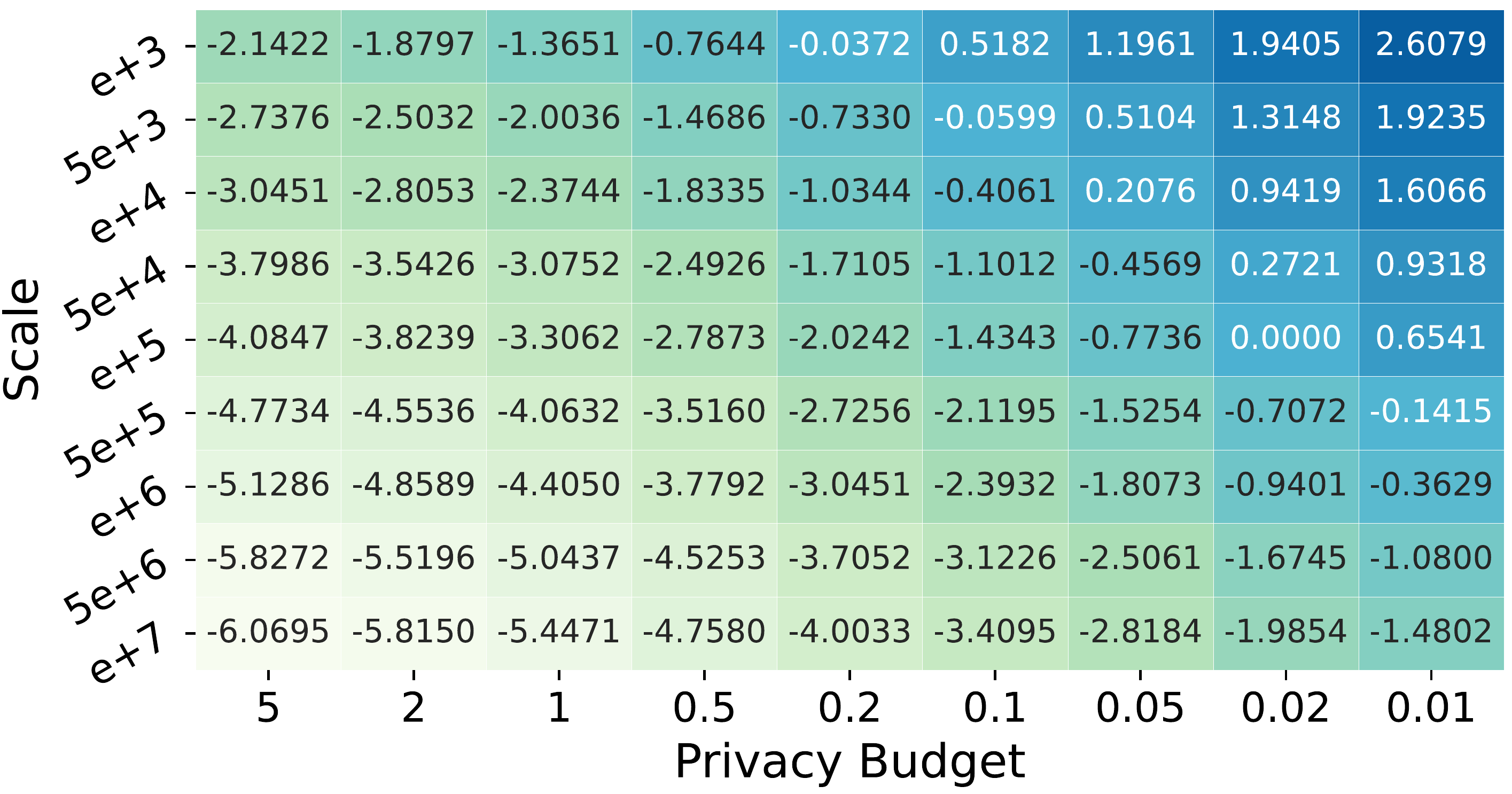}
        \label{dht scale eps zpif}
    }
    \\[-2ex]
    \subfigure[\myahead, 1-dim \Cauchy, $|D|$ = 1024]{
        \includegraphics[width=0.3\hsize]{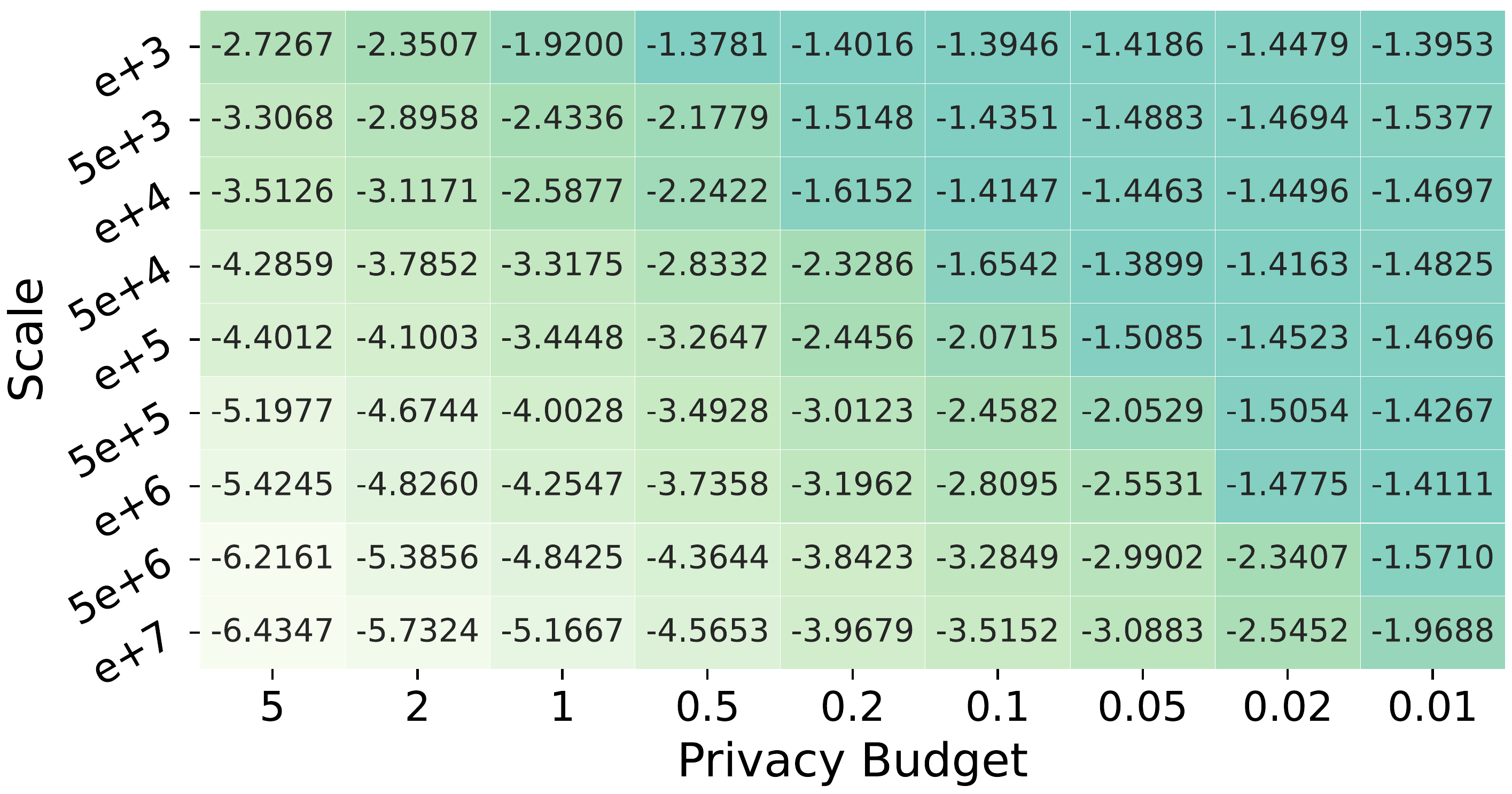}
        \label{ahead scale eps cauchy}
    }
    \subfigure[\myhio, 1-dim \Cauchy, $|D|$ = 1024]{
        \includegraphics[width=0.3\hsize]{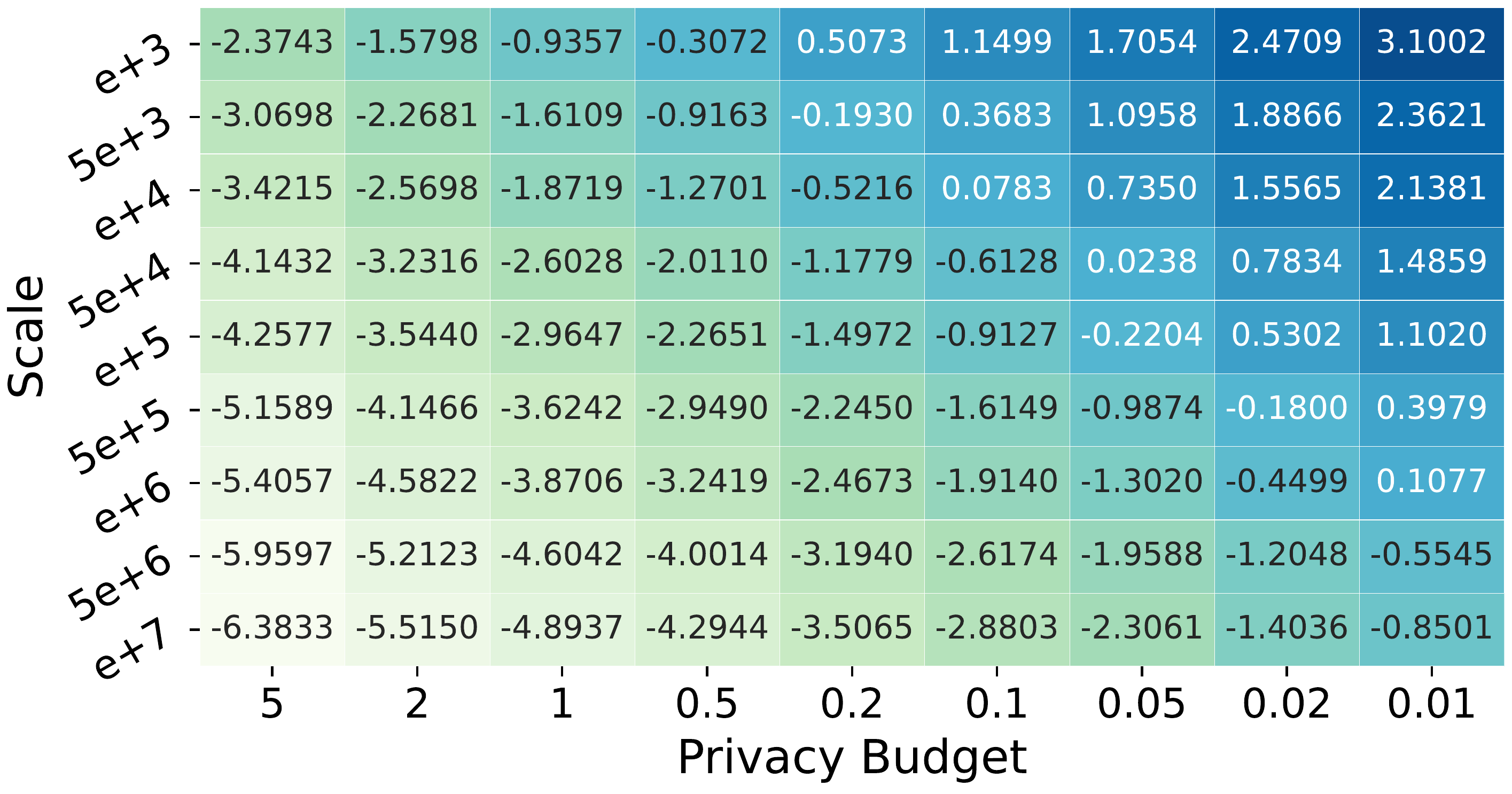}
        \label{hio scale eps cauchy}
    }
    \subfigure[\mydht, 1-dim \Cauchy, $|D|$ = 1024]{
        \includegraphics[width=0.3\hsize]{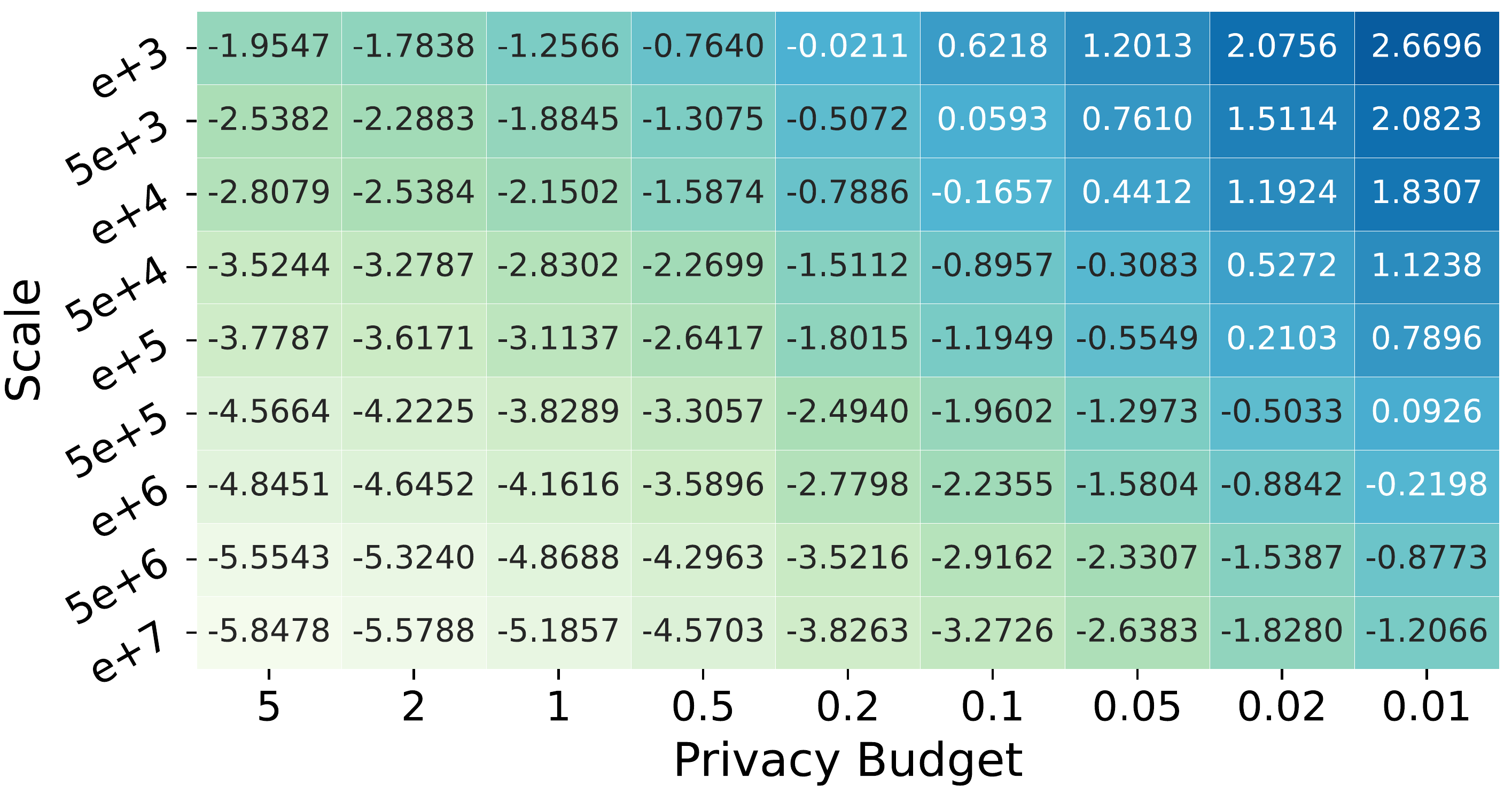}
        \label{dht scale eps cauchy}
    }
    \vspace{-0.5cm}
    \caption{The MSE of different methods when varying user scales and privacy budgets. The results are shown in log scale. }
    \vspace{-0.4cm}
    \label{scale and privacy budget}
\end{figure*}

\mypara{Setup}
We compare the algorithms' performance at increasing user scales, \ie, from $10^3$ to $10^7$. 
Under each user scale $N$, we set 8 different privacy budgets. 
As shown in \autoref{scale and privacy budget}, we use heatmaps to illustrate the impact of user scale and privacy budget on the MSE. 

In practical use, 
the privacy budget is usually specified by users for satisfying their privacy requirements. 
Then, the aggregator needs to ensure query accuracy under the fixed privacy budget. 
For instance, when privacy budget $\epsilon = 1$, the aggregator desires MSE not higher than $10^{-2}$ to ensure the accuracy of the query results. 
\autoref{ahead scale eps zpif} and \autoref{ahead scale eps cauchy} show the MSE of \myahead over different coupling of user scale and privacy budget. 
To meet the accuracy demand, the aggregator needs to collect at least $5 \cdot 10^3$ user records with \myahead. 
In comparison, previous works require more user records (more than $10^4$) for satisfying the same level of accuracy. 
For stronger privacy protection, \ie, smaller privacy budget, all methods require more user records. 
Based on the above observations, we summarize the following observation for selecting a proper user scale. 

\begin{observation}[Necessity of choosing proper user scale]
    \label{Necessity of choosing proper user scale}
    When the privacy protection strength is fixed, \ie, a determined privacy budget $\epsilon$, one needs to use an appropriate user scale to ensure algorithm performance.
\end{observation}
Next, we take a step further to analyze the relation between user scale and privacy budget in order to provide easy methods for selecting user scale. 
As shown in \autoref{ahead scale eps zpif} and \autoref{ahead scale eps cauchy}, \myahead has similar MSEs under different combinations of user scale and privacy budget. 
When user scale $N_1 = 10^6$ and privacy budget $\epsilon_1 = 0.1$, the MSE of \myahead is $10^{-3.0809}$ on the \Zipf dataset. 
In addition, \myahead obtains similar MSE ($10^{-3.0563}$) with $N_1 = 10^4$ and $\epsilon_2 = 1$. 
Based on the results in \autoref{hio scale eps zpif}, \autoref{hio scale eps cauchy}, \autoref{dht scale eps zpif} and \autoref{dht scale eps cauchy}, \myhio and \mydht also have the user scale \& privacy budget exchangeability. 
All the three evaluated methods leverage the \fo protocol when collecting user private data. 
From Equation \ref{OUEVAR}, 
the variance of \oue is $\frac{4e^{\epsilon}}{N\left(e^{\epsilon}-1\right)^{2}}$. 
The \myvar can be converted to $\frac{4}{N\left(e^{\epsilon}+e^{-\epsilon}-2\right)}$, which is approximate to $\frac{4}{N\epsilon^2}$ when $\epsilon$ is not large ($\epsilon \leq 2$). 
Then, the variance of the estimated frequency is approximately the same if the product $N\epsilon^2$ of two combinations is equal. 
Therefore, we have the following observation. 
\begin{observation}[The exchangeability between user scale and privacy budget]
    \label{User scale-Privacy budget exchangeability}
    For any two pairs of user scale and privacy budget combination ($N_1$, $\epsilon_1$) and ($N_2$, $\epsilon_2$), when $\epsilon$ is not large and $N_1\epsilon^2_1 \approx N_2\epsilon^2_2$ is satisfied, \myahead can achieve a similar MSE under two pairs ($N_1$, $\epsilon_1$) and ($N_2$, $\epsilon_2$). 
\end{observation}
\subsection{Impact of Domain Size}
\label{Impact of Domain Size}

\begin{figure*}[t]
    \centering
    \subfigure[\myahead, 1-dim \Zipf, $N = 10^5$]{
        \includegraphics[width=0.3\hsize]{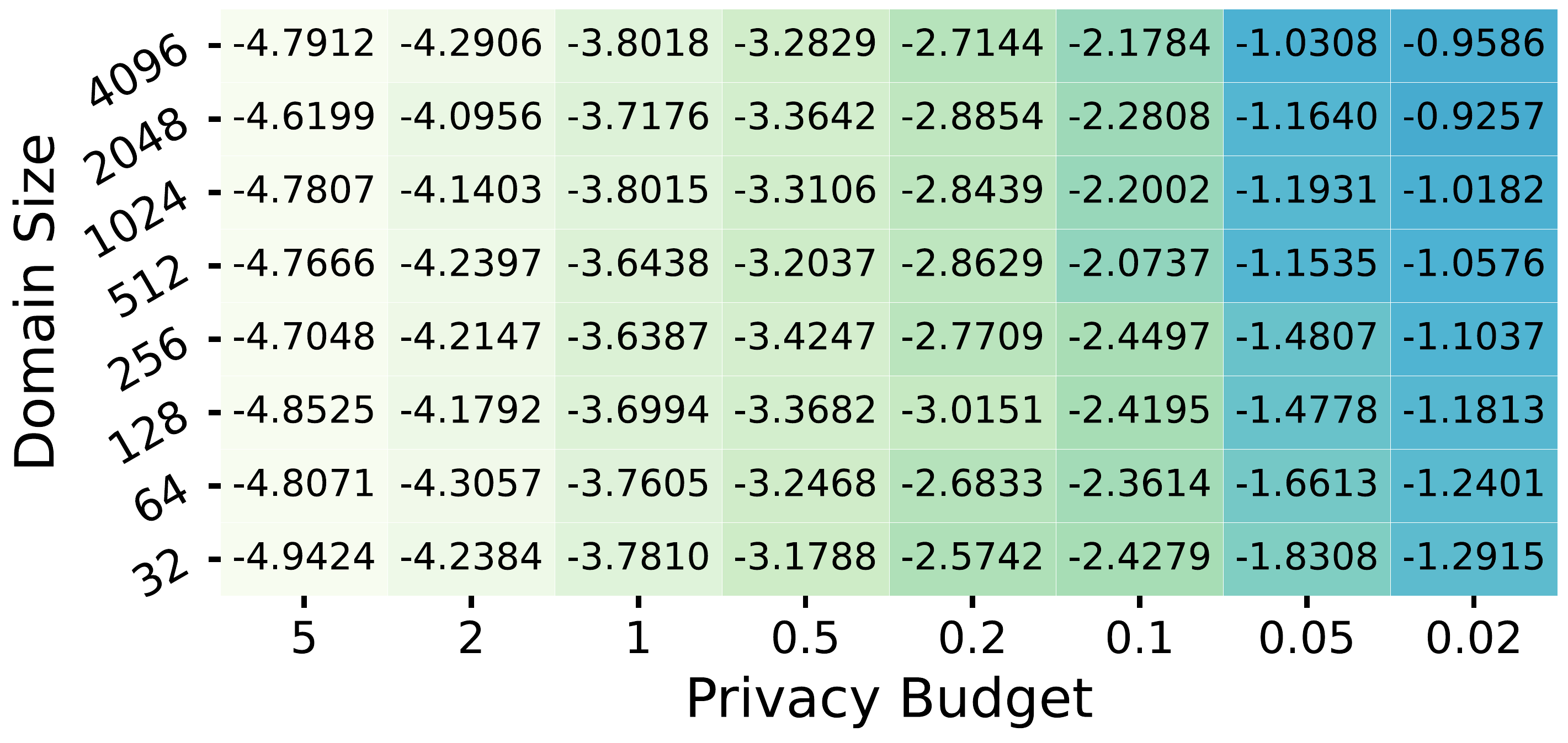}
        \label{ahead domainsize eps zpif}
    }
    \subfigure[\myhio, 1-dim \Zipf, $N = 10^5$]{
        \includegraphics[width=0.3\hsize]{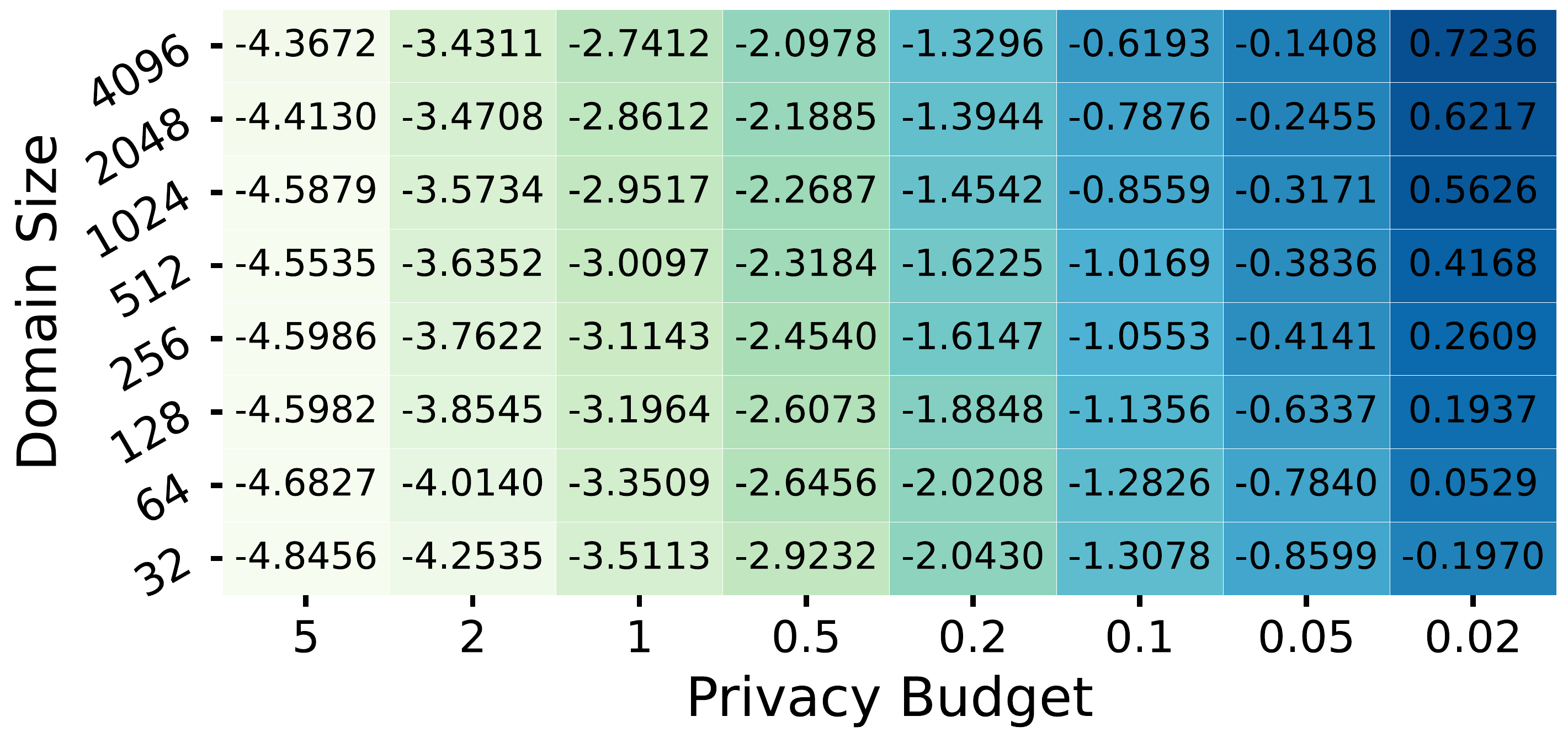}
        \label{hio domainsize eps zpif}
    }
    \subfigure[\mydht, 1-dim \Zipf, $N = 10^5$]{
        \includegraphics[width=0.3\hsize]{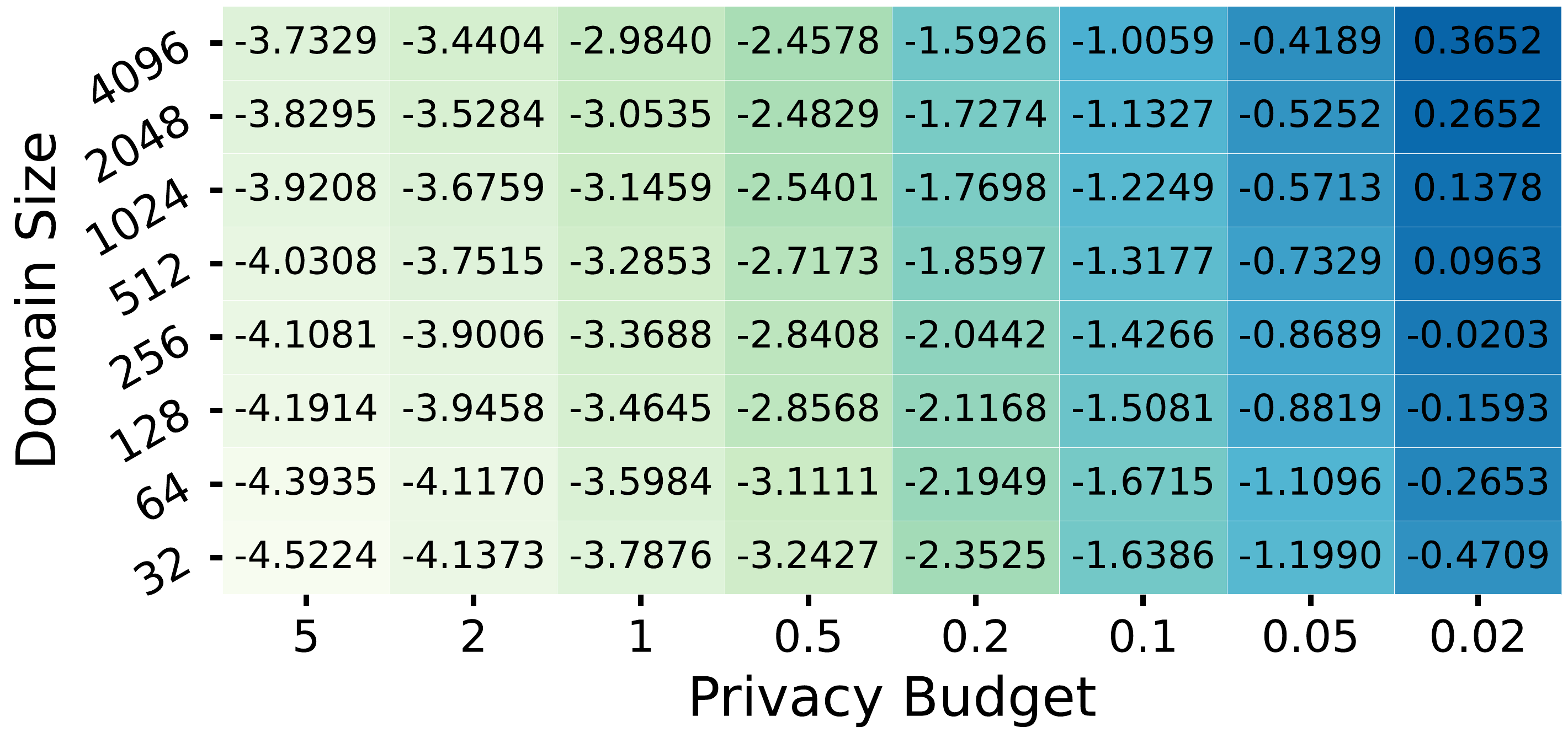}
        \label{dht domainsize eps zpif}
    }
    \\[-2ex]
    \subfigure[\myahead, 1-dim \Cauchy, $N = 10^5$]{
        \includegraphics[width=0.3\hsize]{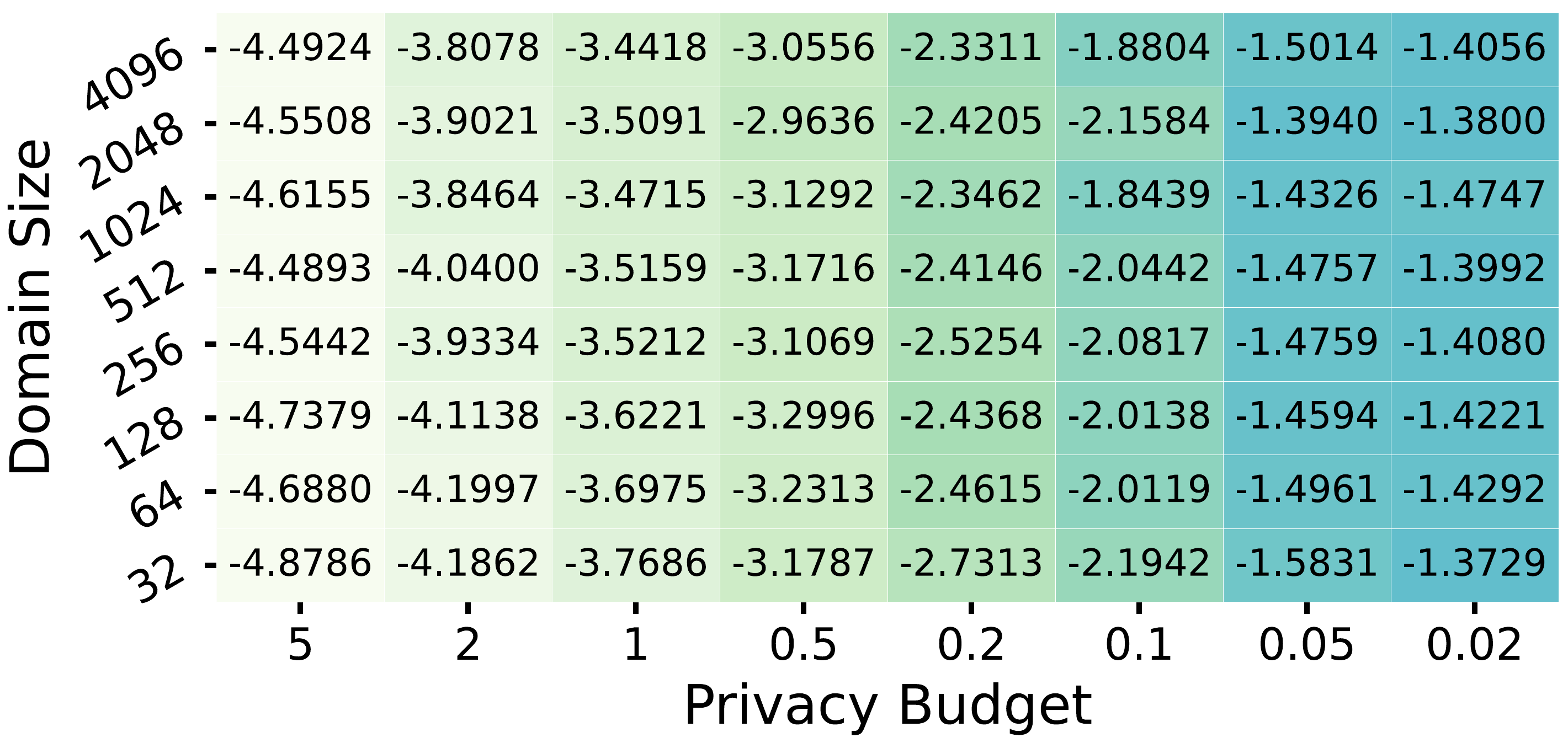}
        \label{ahead domainsize eps cauchy}
    }
    \subfigure[\myhio, 1-dim \Cauchy, $N = 10^5$]{
        \includegraphics[width=0.3\hsize]{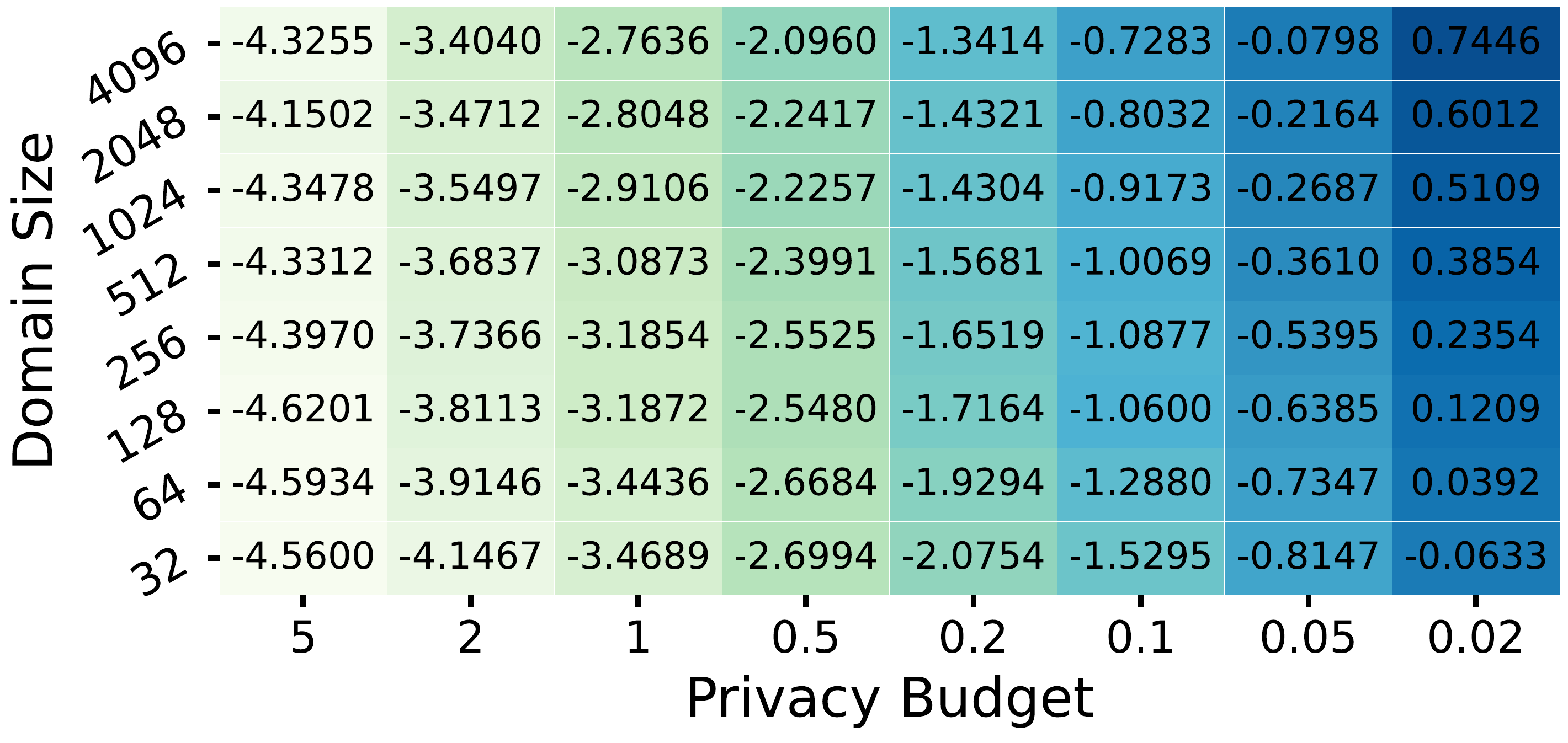}
        \label{hio domainsize eps cauchy}
    }
    \subfigure[\mydht, 1-dim \Cauchy, $N = 10^5$]{
        \includegraphics[width=0.3\hsize]{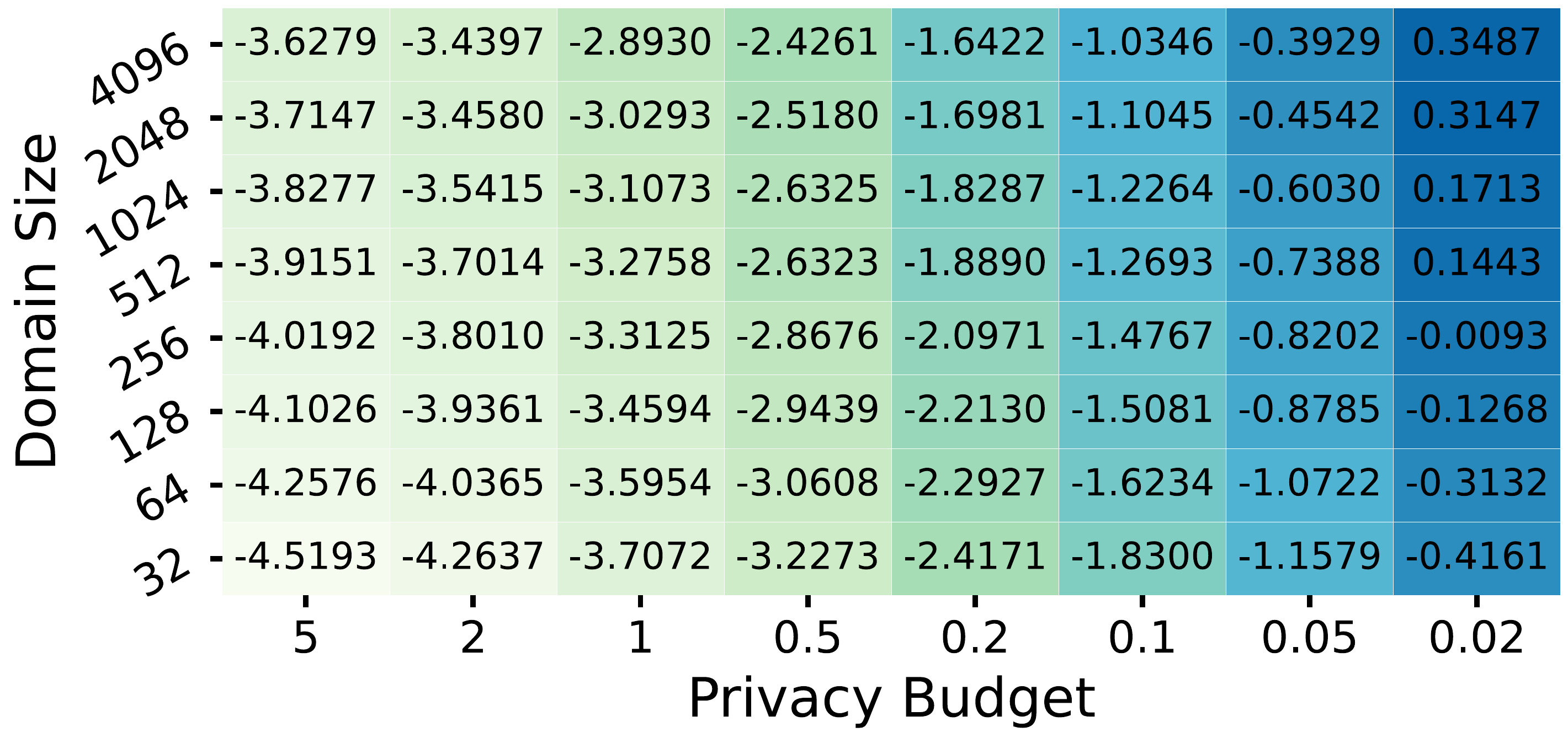}
        \label{dht domainsize eps cauchy}
    }
    \vspace{-0.5cm}
    \caption{The MSE of different methods when varying domain sizes and privacy budgets. The results are shown in log scale. }
    \vspace{-0.4cm}
    \label{domainSizeVSPrivacyBudget}
\end{figure*}
\mypara{Setup}
There are 8 different domain sizes 
used in the experiments. 
In addition, we set 7 different privacy budgets 
to explore the coupling effect of domain size $|D|$ and privacy budget $\epsilon$ on MSE. 

1) From \autoref{hio domainsize eps zpif}, \autoref{dht domainsize eps zpif}, \autoref{hio domainsize eps cauchy} and \autoref{dht domainsize eps cauchy}, 
the MSEs of \myhio and \mydht have a tendency to increase with the increase of domain size. 
Since all the three algorithms adopt the user partition strategy and the number of groups grows with the domain size, 
the user scale in each group becomes smaller for a larger domain size. 
For example, when domain size rises from $32$ to $4096$, the number of user groups enlarges from 5 to 12 for \mydht. 
Based on \autoref{GRRVAR} and \autoref{OUEVAR}, the variance of the \fo algorithm is inversely proportional to user scale. 
2) As shown in \autoref{ahead domainsize eps zpif} and \autoref{ahead domainsize eps cauchy}, \myahead is less affected by the changes in domain size. 
Due to the threshold setting, \myahead makes the actual domain size smaller by amalgamating some lower frequency sub-domains. 
In addition, \myahead averages the estimated frequency values of repeated sub-domains, \eg, sub-domains $[0, 3]$ and $[6, 7]$ in \autoref{AHEAD method}, to reduce the influence of noise. 
Based on the results in \autoref{domainSizeVSPrivacyBudget} and the above analysis, we have the following observation on the impact of domain size. 

\begin{observation}[Robustness under various domain sizes] 
    \label{Robustness under various domain size}
    The impact of varying domain size on \myahead is different from \myhio and \mydht. \myahead reacts more robust to domain size changes. 
    \vspace{-0.2cm}
\end{observation}

\subsection{Impact of Data Skewness}
\mypara{Setup}
We compare the algorithms' performance at various data skewnesses, where the data records are sampled from low skewness (0 for both \Gaussian and \Laplacian) to high skewness (0.9 for \Gaussian and 2.0 for \Laplacian). 
As shown in \autoref{skewness and privacy budget}, 
under each user scale, we set 8 different privacy budgets and use heatmaps to illustrate the impact of skewness and privacy budget on the MSE. 

1) From \autoref{ahead skewness eps 105 gaussian} and \autoref{ahead skewness eps 105 laplace}, when privacy budget $\epsilon \leq 0.1$, the MSEs of \myahead have a tendency to decrease with the increase of data skewness. 
Since there may exist more nodes suitable for pruning for data with higher skewness, \ie, more sparse sub-domains, \myahead can significantly suppress the injected noises. 
2) With the increase of privacy budget or user scale, the impact of skewness becomes insignificant on MSE of \myahead. 
Recalling \autoref{theta setting1} in \autoref{Selection parameters}, the threshold $\theta$ becomes small with a large privacy budget or user scale. 
For instance, $\theta = 0.11$ when $\epsilon=0.1$ and $N = 10^6$, and $\theta=0.003$ when $\epsilon=1$ and $N = 10^7$. 
For large privacy budgets or user scales, \myahead reduces the intensity of tree pruning, where the impact of data skewness on the tree construction will weaken. 

\begin{observation}[The benefit from high skewness]
    \label{The benefit from high skewness}
    \myahead tends to have smaller MSE on highly skew data, where the impact of skewness will weaken as the privacy budget or user scale increases. 
\end{observation}

\begin{figure*}[!ht]
    \centering
    \subfigure[\myahead, 1-dim \Gaussian, $N = 10^5$]{
        \includegraphics[width=0.31\hsize]{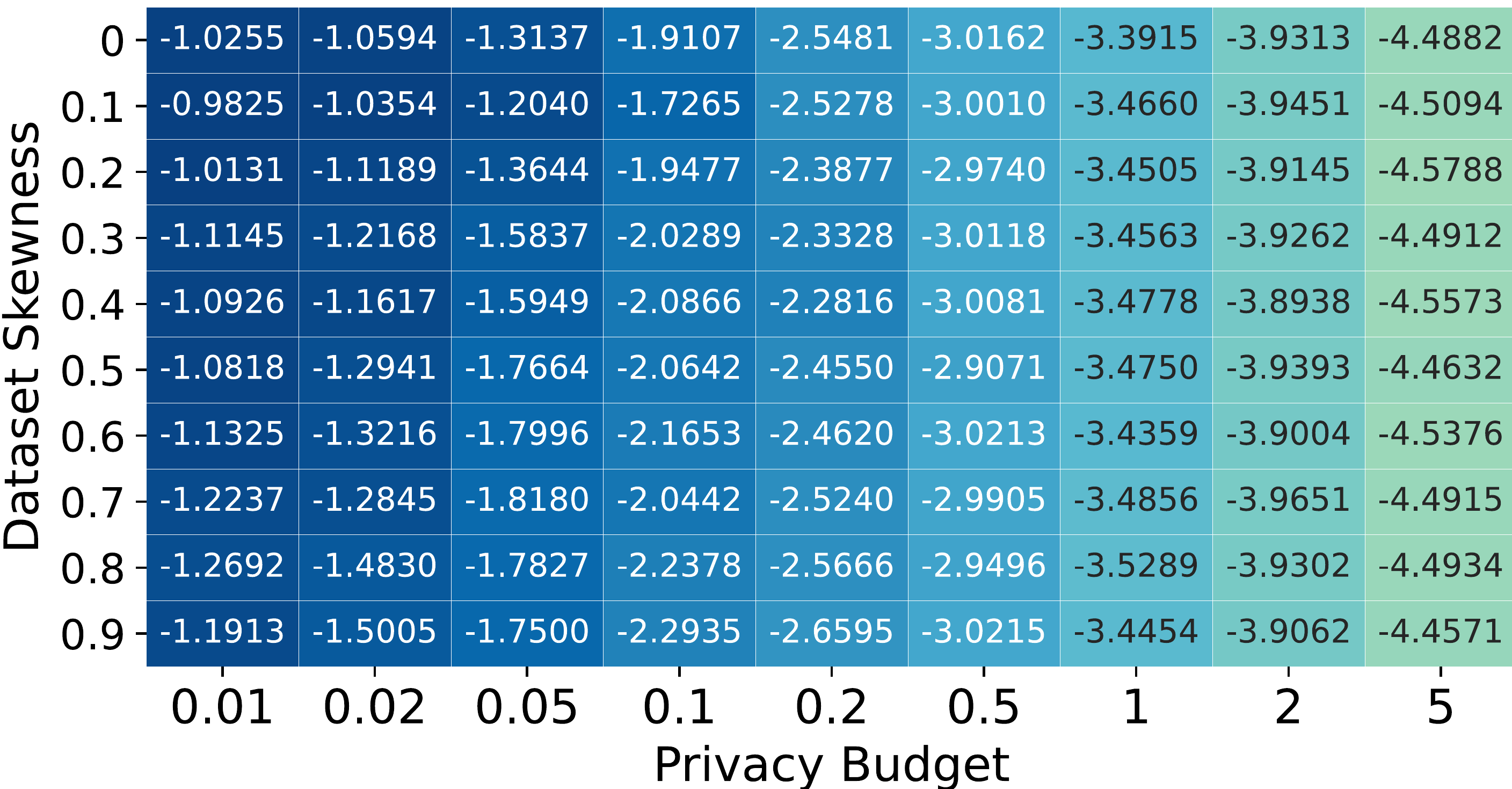}
        \label{ahead skewness eps 105 gaussian}
    }
    \subfigure[\myahead, 1-dim \Gaussian, $N = 10^6$]{
        \includegraphics[width=0.31\hsize]{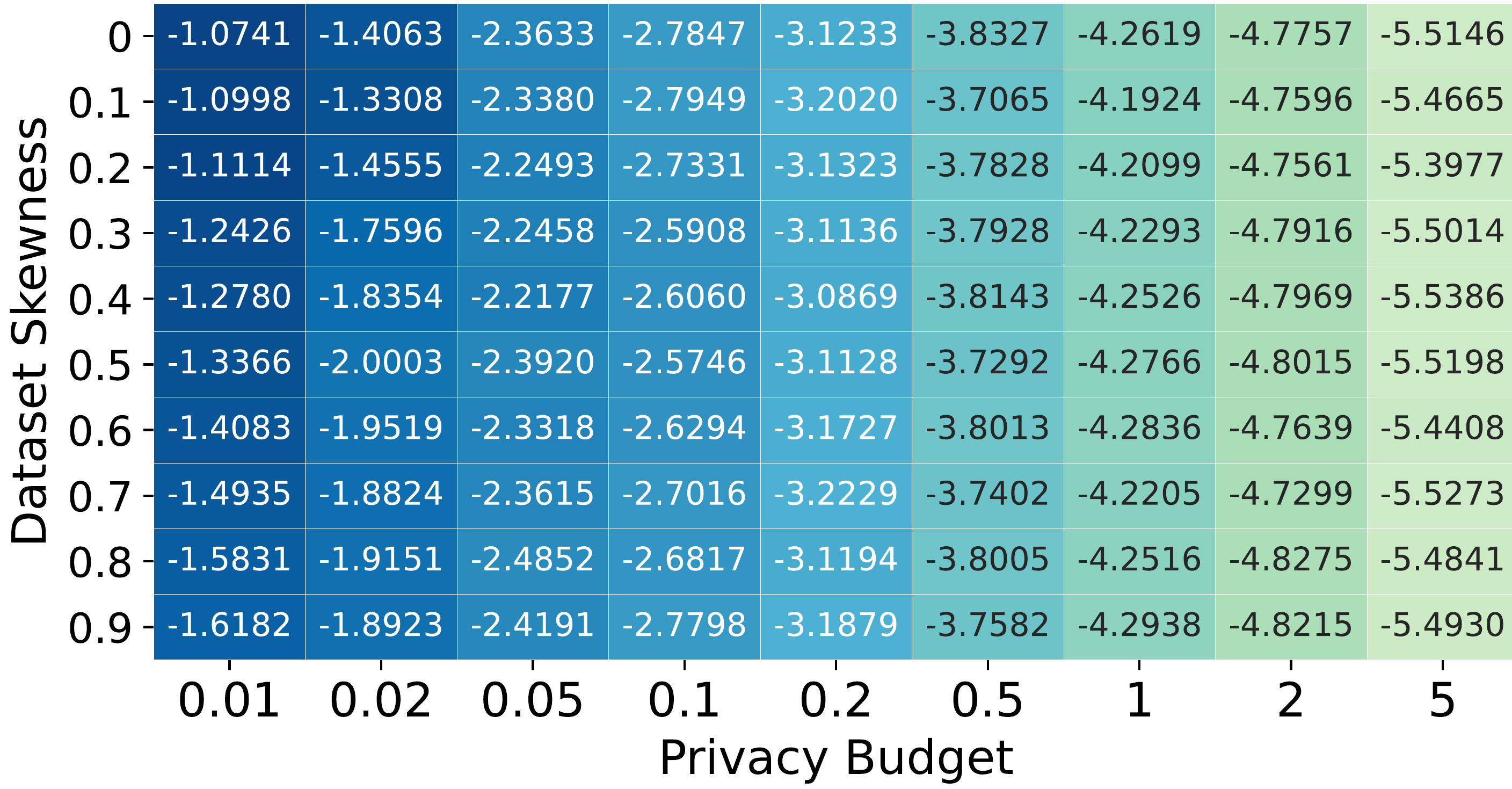}
        \label{ahead skewness eps 106 gaussian}
    }
    \subfigure[\myahead, 1-dim \Gaussian, $N = 10^7$]{
        \includegraphics[width=0.31\hsize]{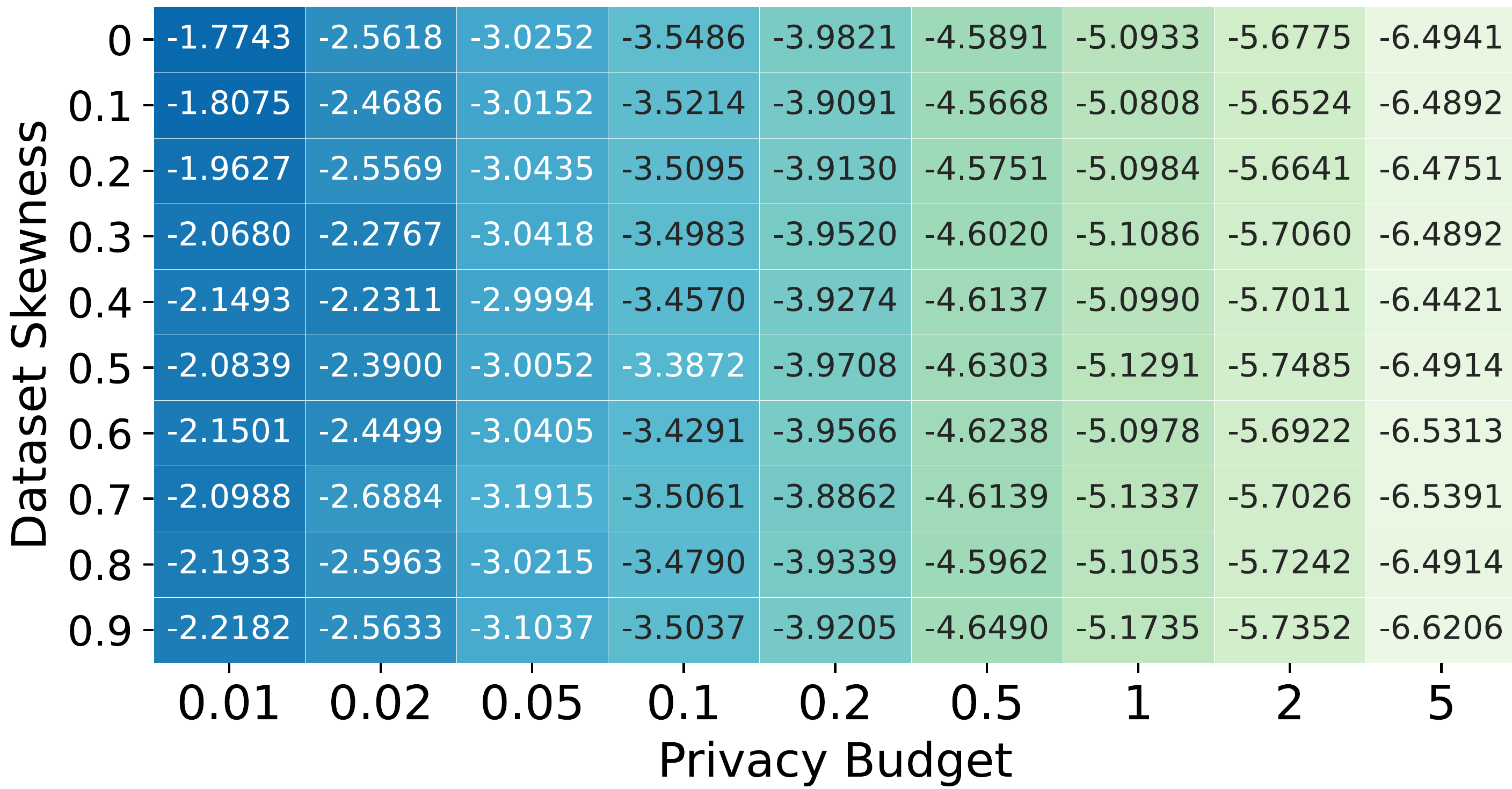}
        \label{ahead skewness eps 107 gaussian}
    }
    \\[-2ex]
    \subfigure[\myahead, 1-dim \Laplacian, $N = 10^5$]{
        \includegraphics[width=0.31\hsize]{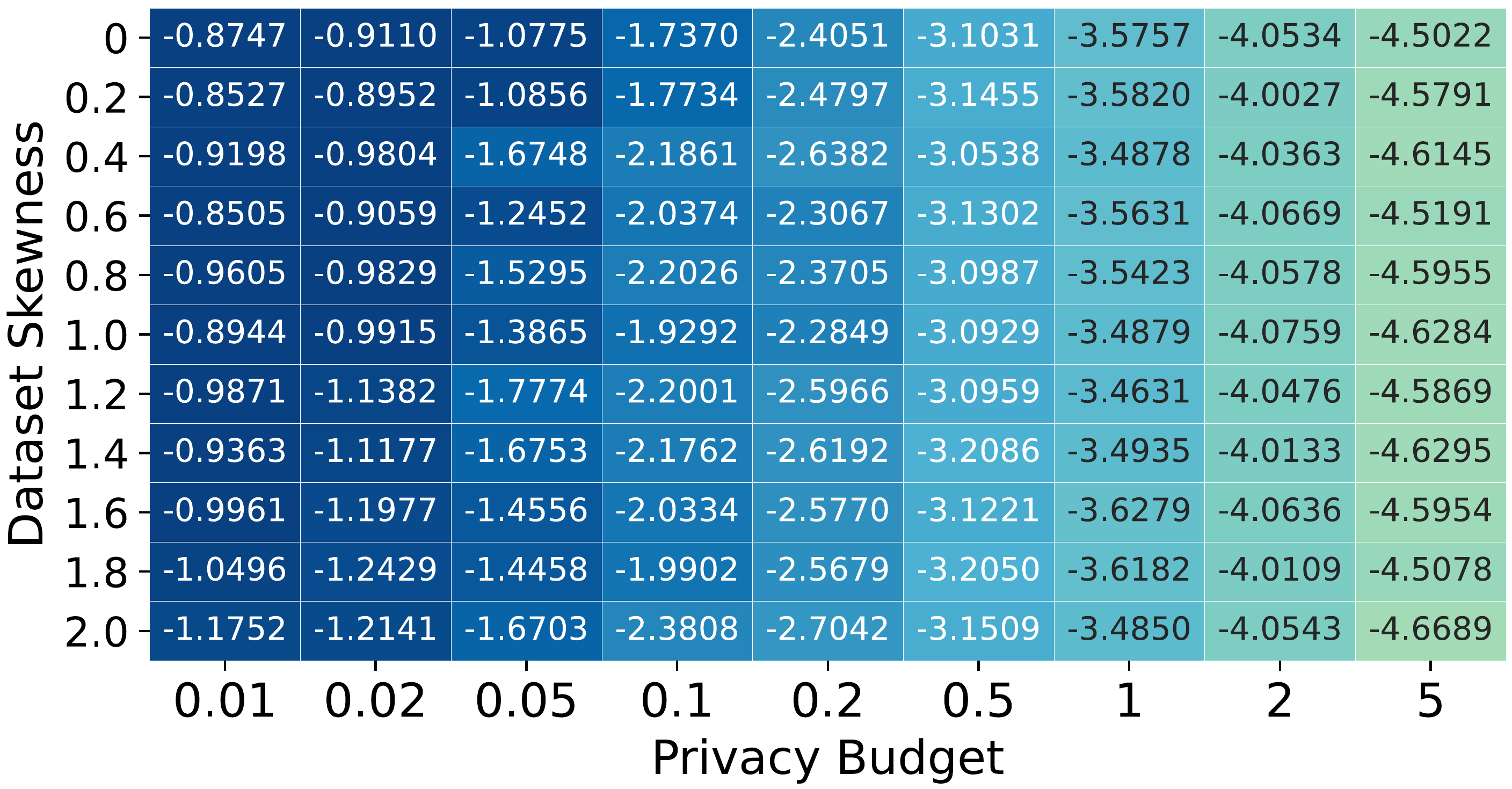}
        \label{ahead skewness eps 105 laplace}
    }
    \subfigure[\myahead, 1-dim \Laplacian, $N = 10^6$]{
        \includegraphics[width=0.31\hsize]{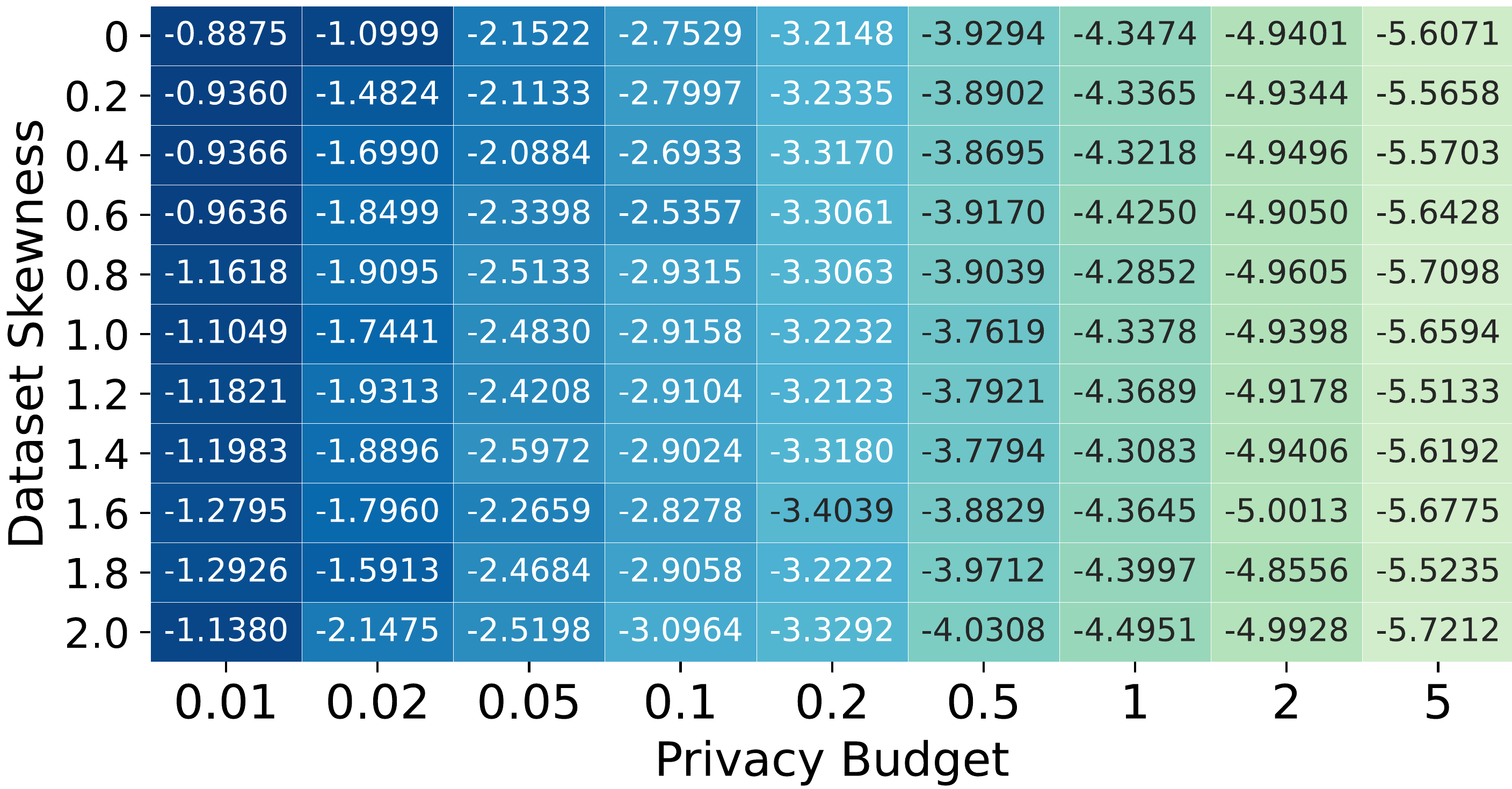}
        \label{ahead skewness eps 106 laplace}
    }
    \subfigure[\myahead, 1-dim \Laplacian, $N = 10^7$]{
        \includegraphics[width=0.31\hsize]{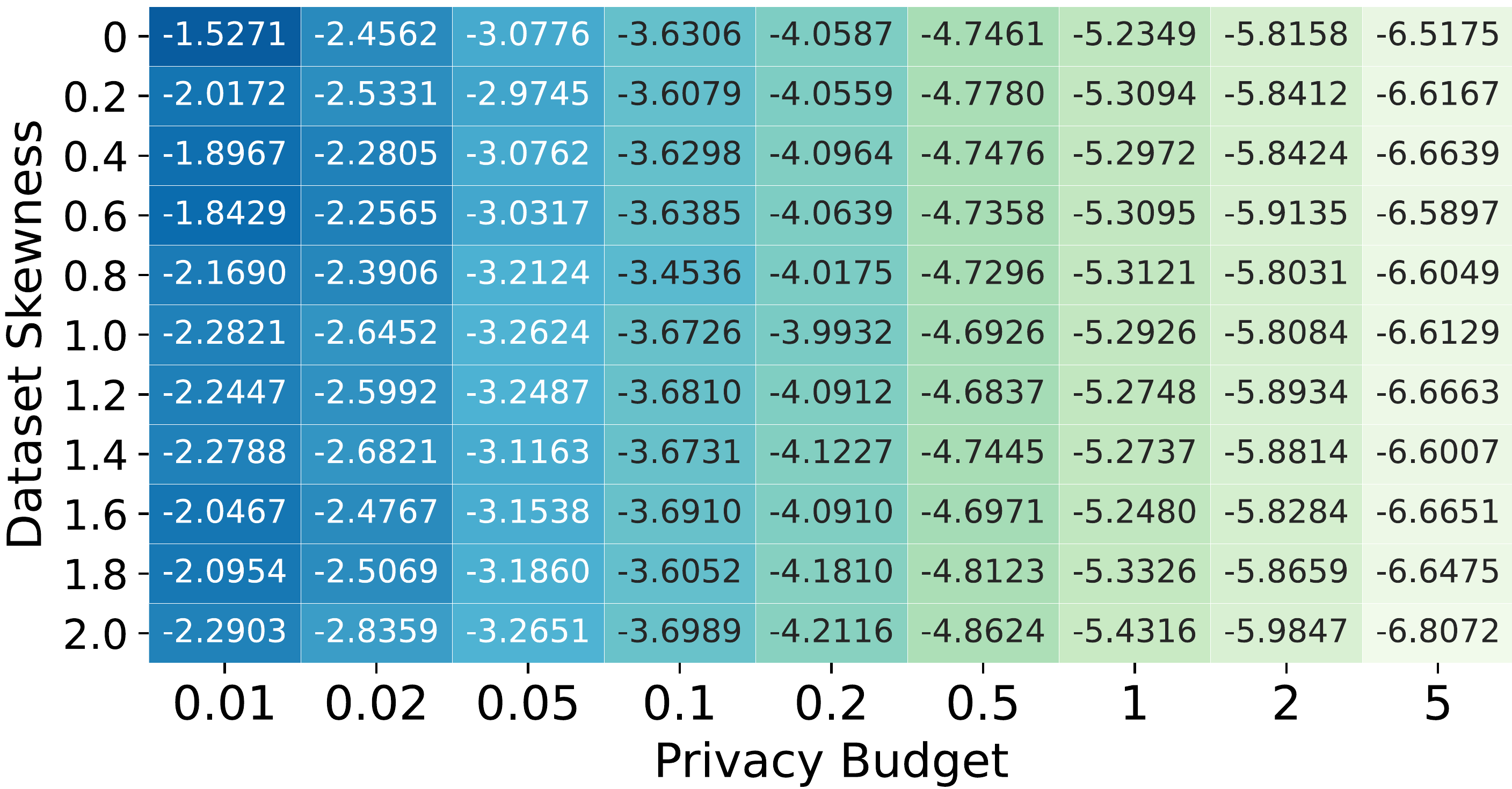}
        \label{ahead skewness eps 107 laplace}
    }
    \vspace{-0.5cm}
    \caption{The MSE of \myahead when varying data skewnesses and privacy budgets. The results are shown in log scale. }
    \vspace{-0.4cm}
    \label{skewness and privacy budget}
\end{figure*}

\subsection{Impact of Attribute Correlation}
\label{Impact of Attribute Correlation}
\mypara{Setup}
Next, we evaluate the impact of different attribute correlations on query errors as shown in \autoref{3-dim correlation} and \autoref{5-dim correlation}. 
The experimental settings are similar to 2-dim scenes, \ie, varying the correlation coefficient $r$ from 0.1 (\emph{representing weakly correlated}) to 0.9 (\emph{representing strongly correlated}) with a fixed privacy budget $\epsilon = 1.1$. 

From the results, we have the following findings. 
1) The \lle method can protect the correlation between attributes in high-dimensional situations. 
Instead of decomposing all dimensions simultaneously like \de, \lle uses the 2-dim \myahead tree to estimate the answers of high-dimensional range queries. The MSE of \myahead behaves consistently across different correlations, thus the correlation of the data is well preserved by \lle. 
2) For high-dimensional scenarios, the impact of attribute correlation becomes less on \myHDG. 
With the increase of dataset dimension, the number of 1-dim grids ($C_m^2$) increases faster than the 1-dim grids ($C_m^1$). 
Therefore, as shown in \autoref{5-dim correlation}, the negative impact of the 1-dim grids on the correlation is reduced. 

\begin{observation}[Robustness under various attribute correlation] 
    \label{Robustness under various attribute correlation}
    The MSE of \myahead behaves consistently on attribute correlation changes. 
    As the data dimension increases, the impact of correlation on \myHDG decreases. 
    \vspace{-0.2cm}
\end{observation}
\begin{figure*}[!t]
    \centering
    \subfigure[3-dim \Laplacian, $N=10^6$, vary $r$]{\label{fig:Three_LaplaceCorr-Set_10_6-Domain_6_6_6}\includegraphics[width=0.24\hsize]{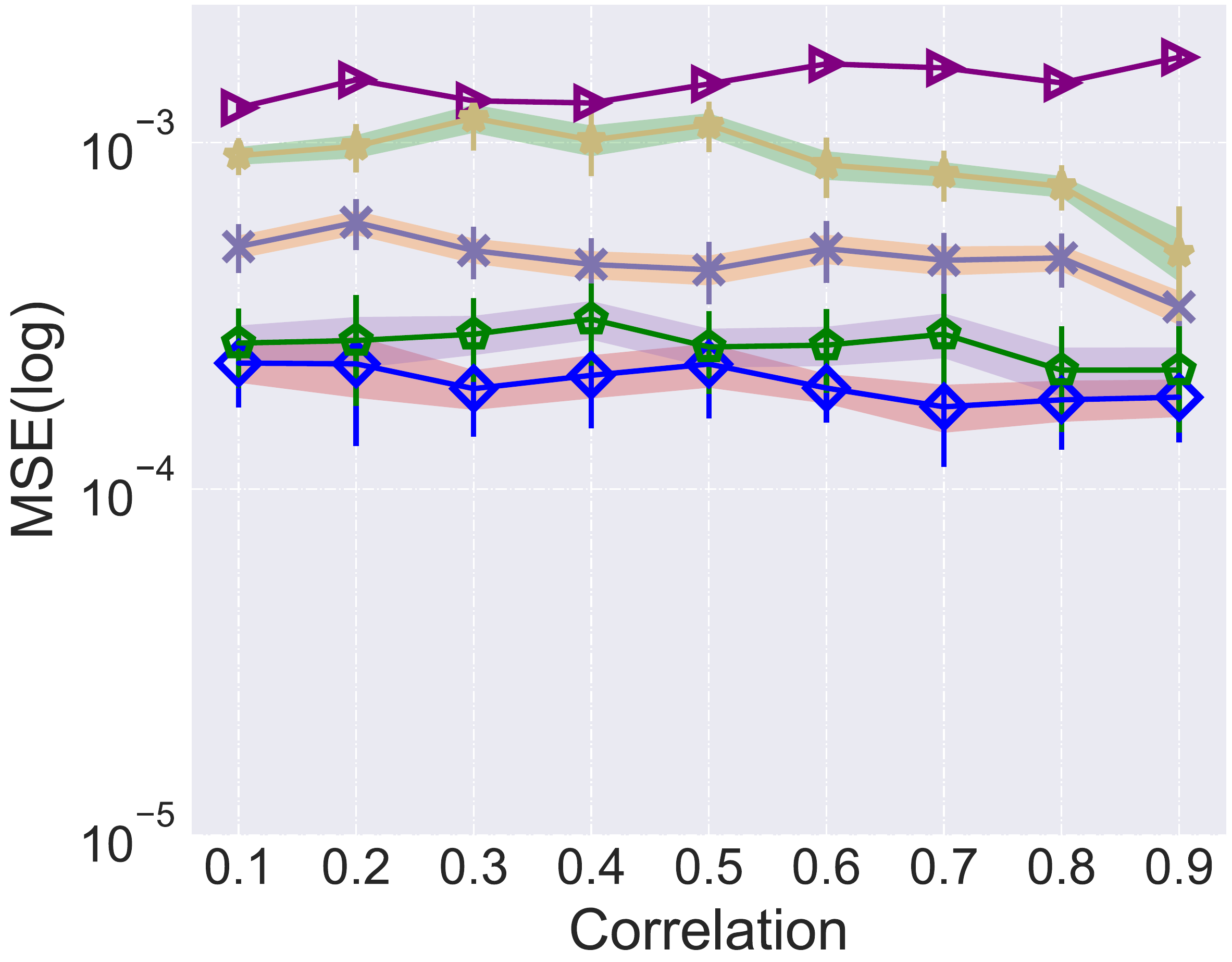}}
    \subfigure[3-dim \Laplacian, $N=10^7$, vary $r$]{\label{fig:Three_LaplaceCorr-Set_10_7-Domain_6_6_6}\includegraphics[width=0.24\hsize]{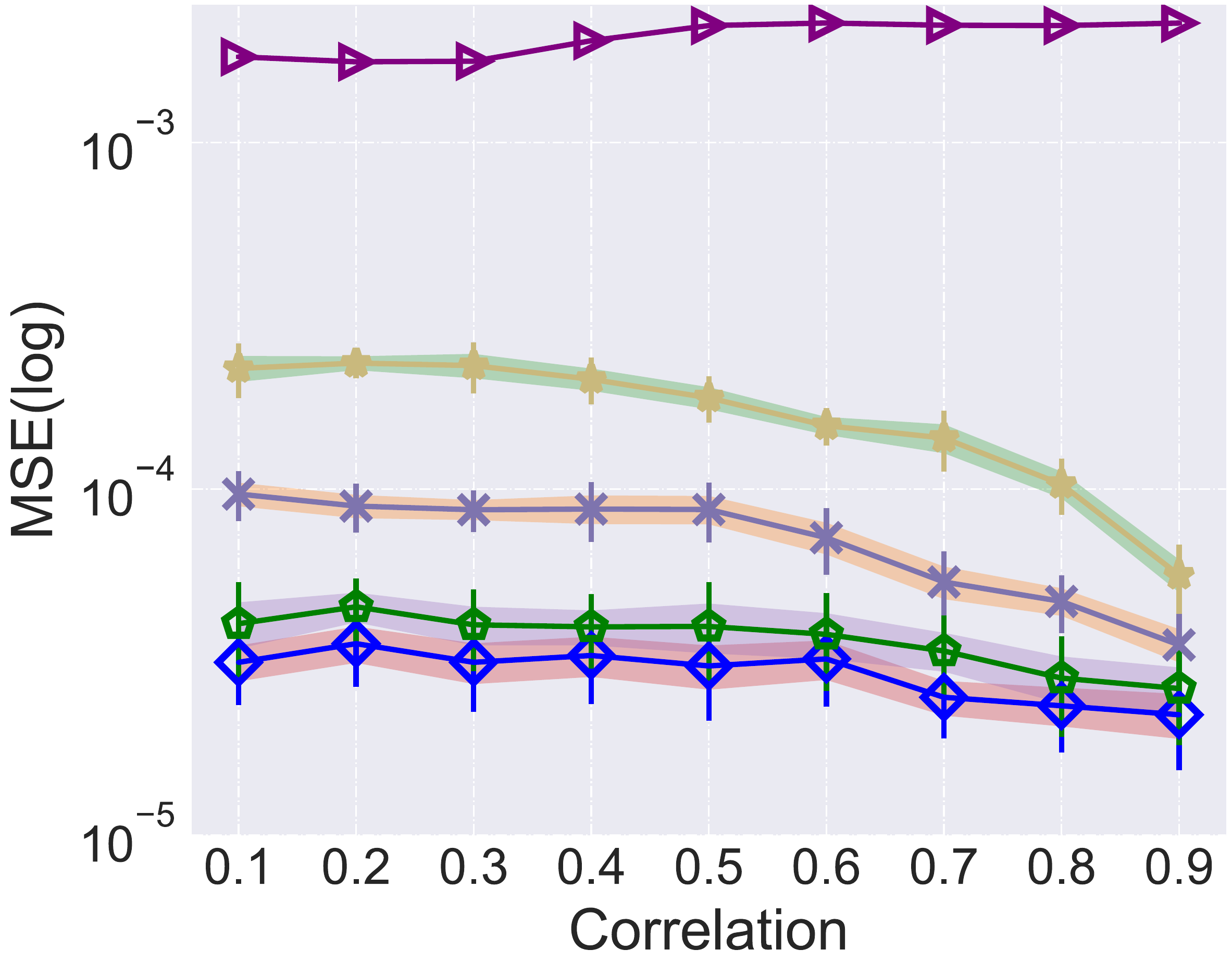}}
    \subfigure[3-dim \Gaussian, $N=10^6$, vary $r$]{\label{fig:Three_NormalCorr-Set_10_6-Domain_6_6_6}\includegraphics[width=0.24\hsize]{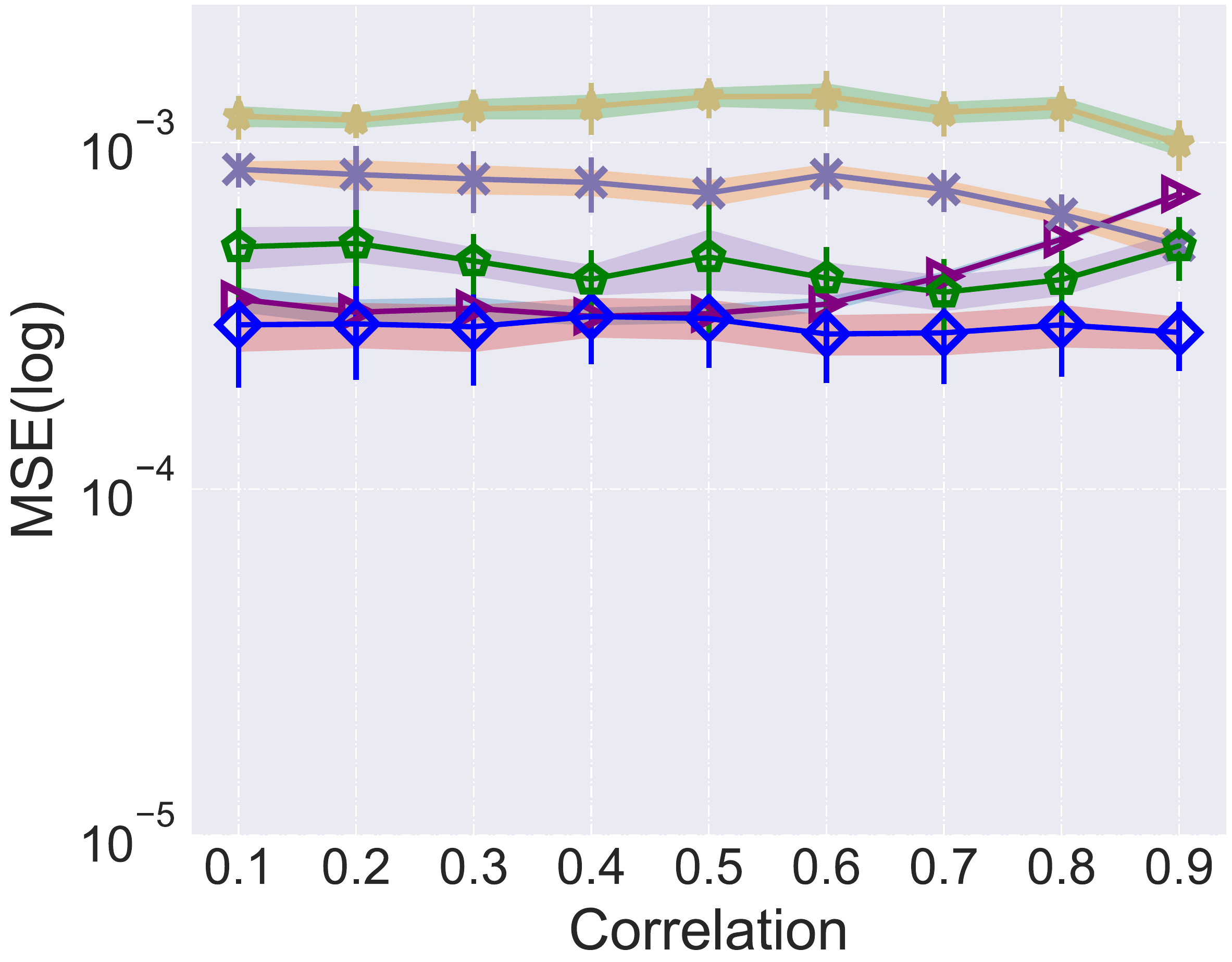}}
    \subfigure[3-dim \Gaussian, $N=10^7$, vary $r$]{\label{fig:Three_NormalCorr-Set_10_7-Domain_6_6_6}\includegraphics[width=0.24\hsize]{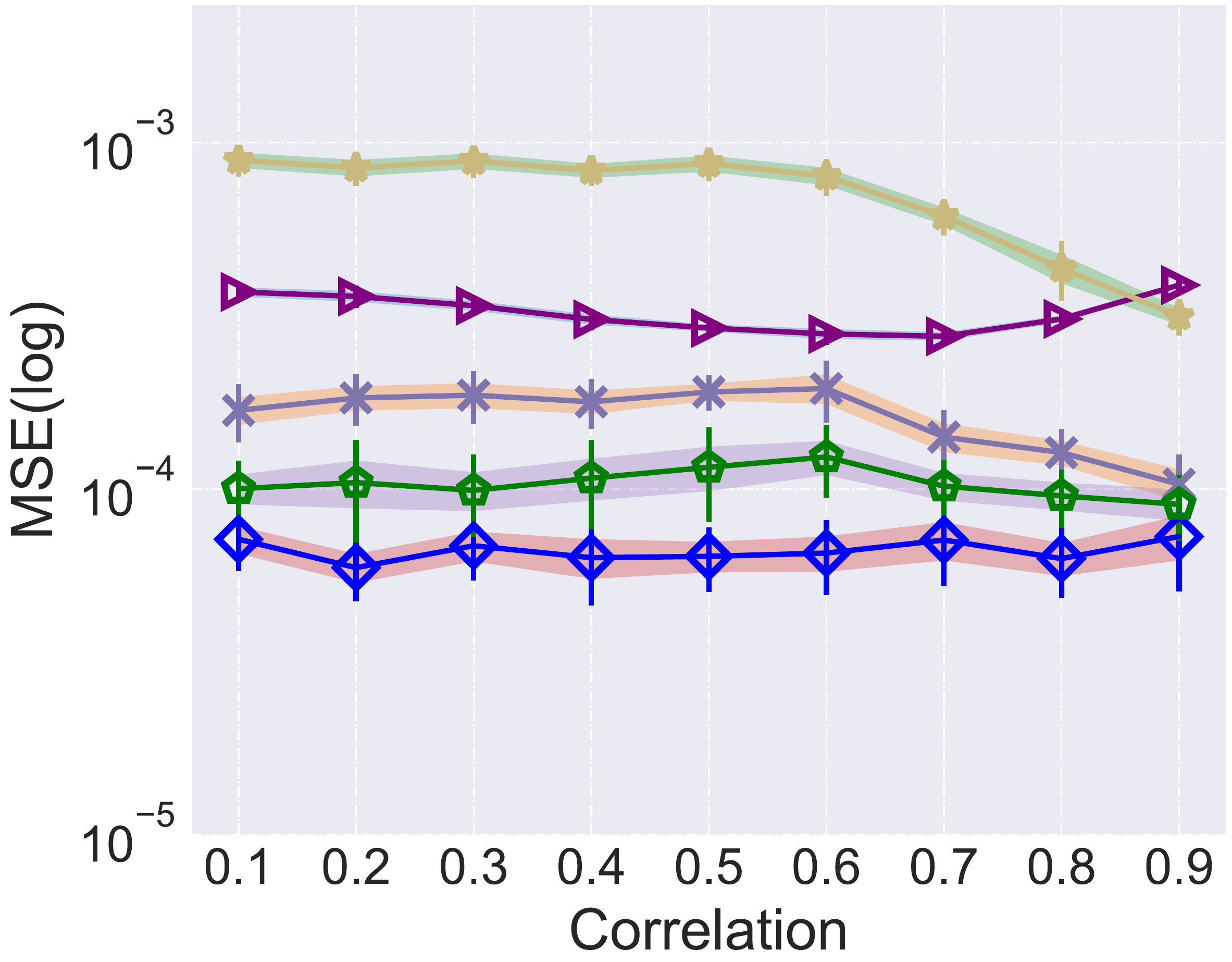}}\\[-1ex]
    \subfigure{\includegraphics[width=0.8\textwidth]{./figures_others/comparison_3D_legend5.pdf.pdf}}  
    \vspace{-0.4cm}
    \caption{Comparison of different methods on 3-dim \Laplacian and \Gaussian datasets under various attribute correlations. \de and \lle respectively represent two high-dimensional expansion methods, \ie, ``direct estimation'' and ``leveraging low-dimensional estimation''. \myHDG is a baseline method. The results are shown in log scale.}
    \vspace{-0.4cm}
    \label{3-dim correlation}
\end{figure*}

\begin{figure*}[!t]
    \centering
    \subfigure[5-dim \Laplacian, $N=10^6$, vary $r$]{\label{fig:Five_LaplaceCorr-Set_10_6-Domain_6_6_6_6_6}\includegraphics[width=0.24\hsize]{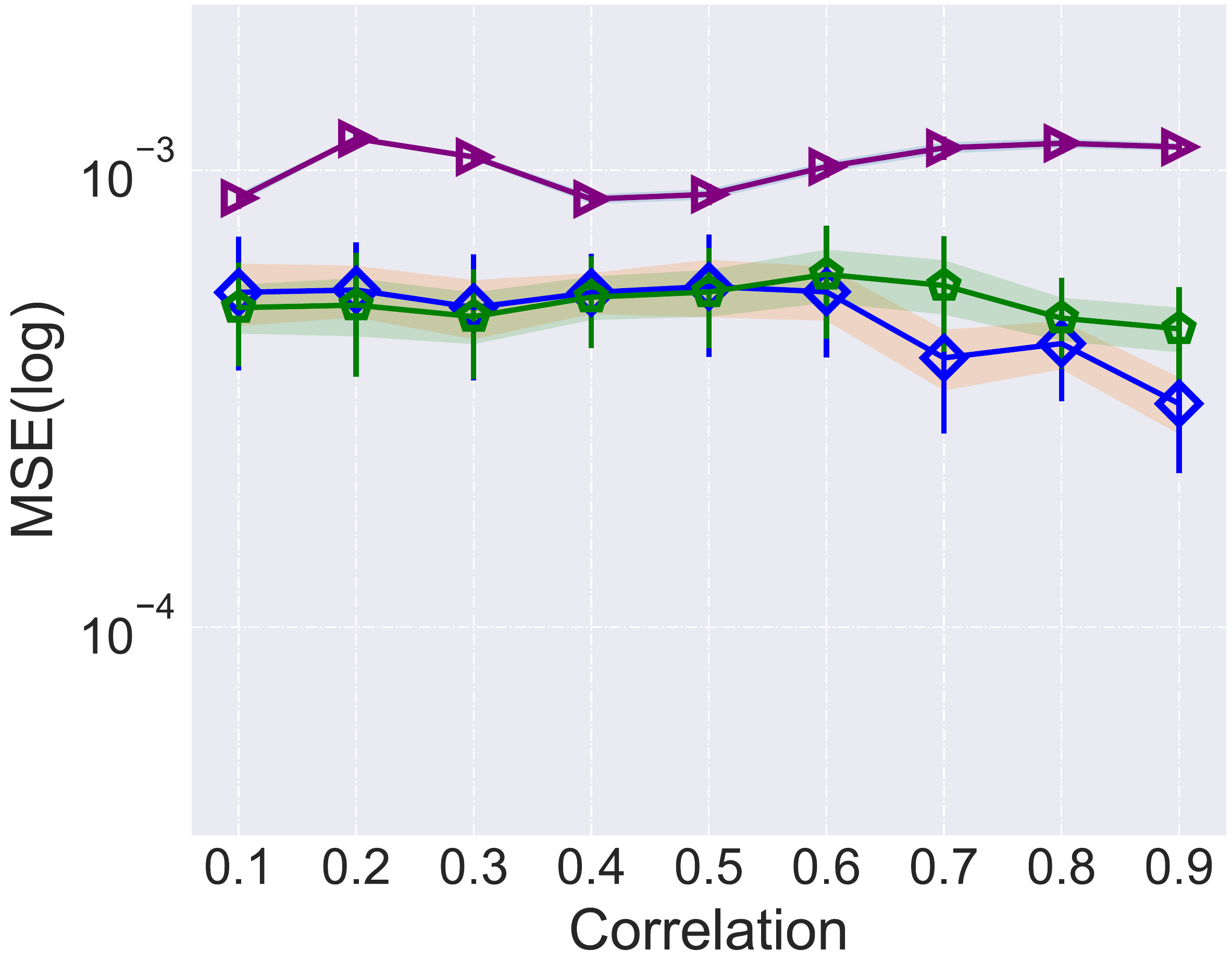}}
    \subfigure[5-dim \Laplacian, $N=10^7$, vary $r$]{\label{fig:Five_LaplaceCorr-Set_10_7-Domain_6_6_6_6_6}\includegraphics[width=0.24\hsize]{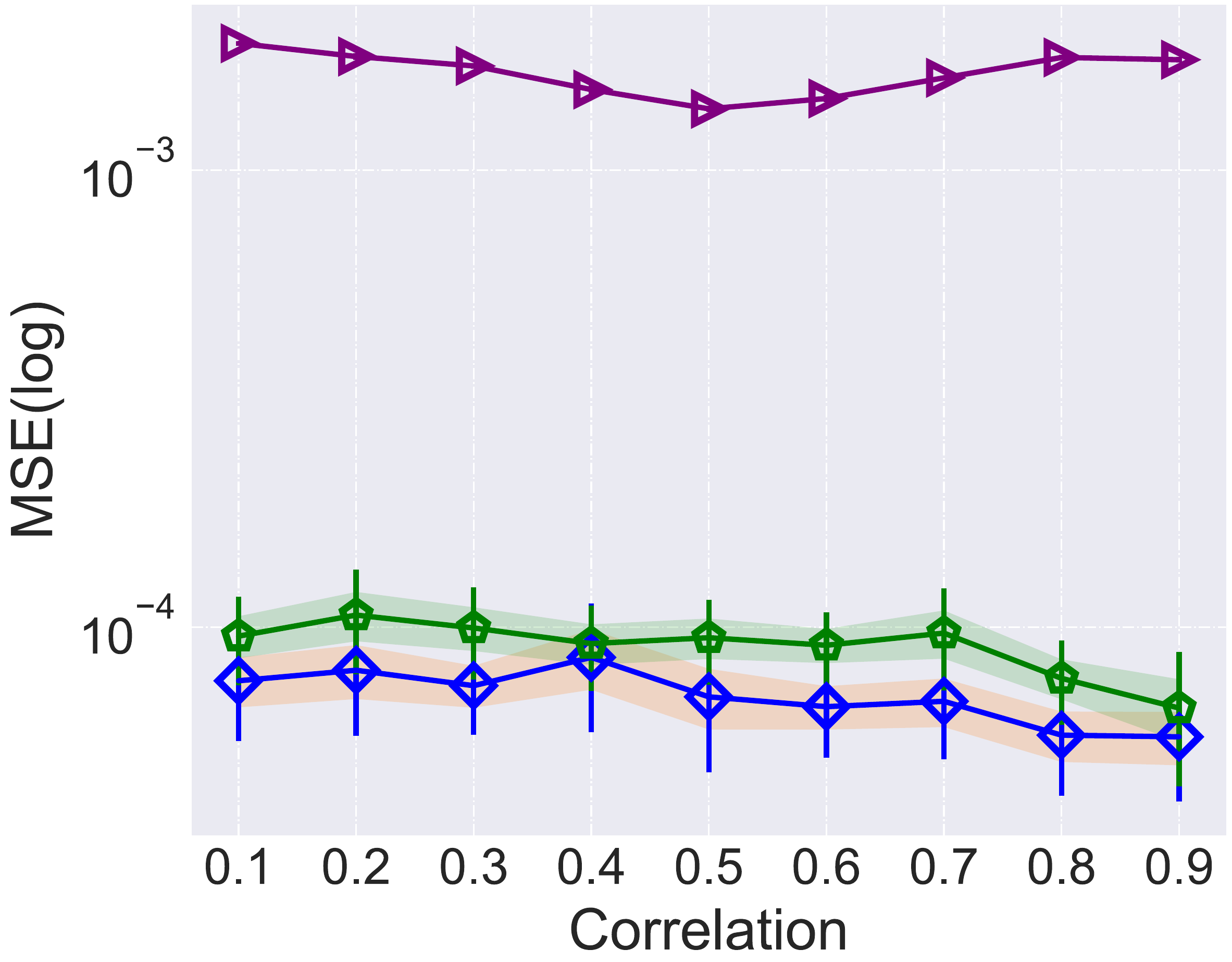}}
    \subfigure[5-dim \Gaussian, $N=10^6$, vary $r$]{\label{fig:Five_NormalCorr-Set_10_6-Domain_6_6_6_6_6}\includegraphics[width=0.24\hsize]{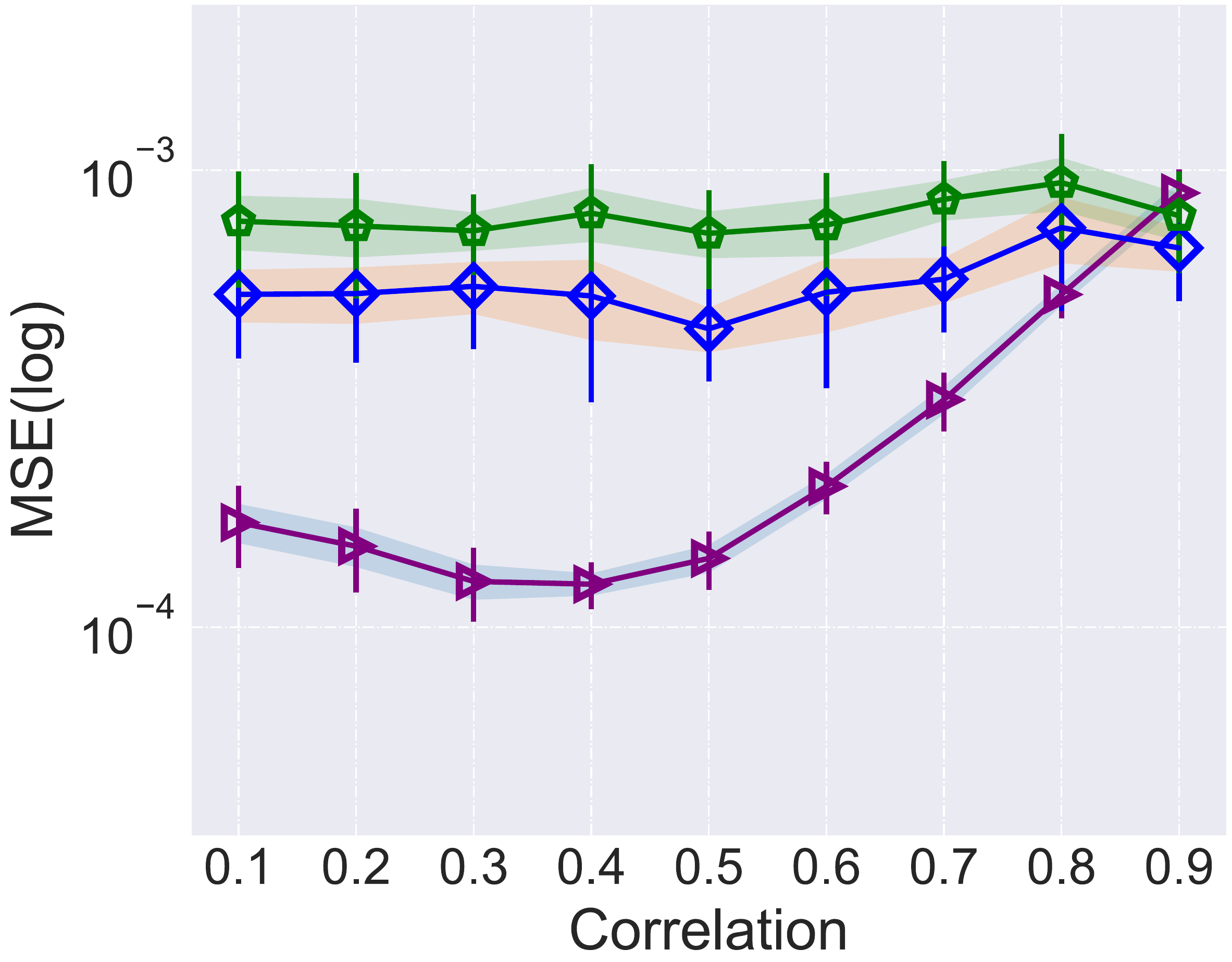}}
    \subfigure[5-dim \Gaussian, $N=10^7$, vary $r$]{\label{fig:Five_NormalCorr-Set_10_7-Domain_6_6_6_6_6}\includegraphics[width=0.24\hsize]{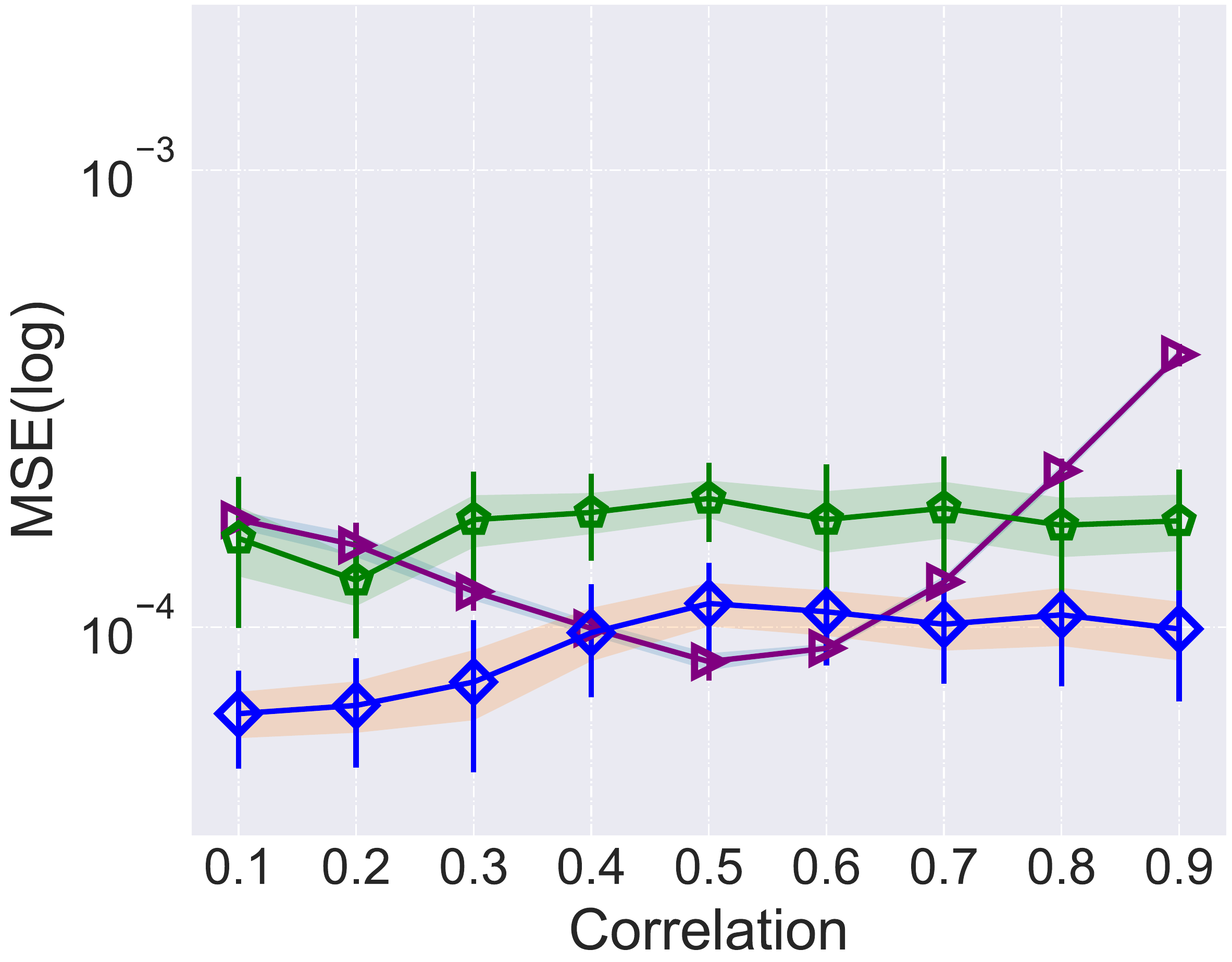}}\\[-1ex]
    \subfigure{\includegraphics[width=0.5\textwidth]{./figures_others/comparison_3D_legend3.pdf.pdf}}  
    \vspace{-0.5cm}
    \caption{Comparison of different methods on 5-dim \Laplacian and \Gaussian datasets under various attribute correlations. }
    \vspace{-0.4cm}
    \label{5-dim correlation}
\end{figure*}

\subsection{Remarks}
The six observations mentioned above reflect the impact of various factors on the performance of \myahead. 
\autoref{Necessity of choosing proper user scale} and \autoref{User scale-Privacy budget exchangeability} describe the coupling influence of user scale and privacy budget on algorithm accuracy and demonstrate the transformation relationship between the user scale and privacy budget. 
Observation \autoref{The benefit from high skewness} provides an encouraging of practical adoption of \myahead when facing highly skew data. 
\autoref{Robustness under various domain size} and \autoref{Robustness under various attribute correlation} demonstrate the advantageous robustness of \myahead under different domain sizes and attribute correlations.

In detail, 1) when the privacy protection strength (the privacy budget $\epsilon$) is fixed, according to \autoref{Necessity of choosing proper user scale}, 
\myahead obtains more accurate query answers with larger user scales. 
2) If the user scale is insufficient, the aggregator tends to pay the cost, such as compensation fee, 
to trade with users to reduce the strength of privacy protection (increasing the privacy budget $\epsilon$). 
From \autoref{User scale-Privacy budget exchangeability}, 
the aggregator can readily calculate the appropriate $\epsilon$ based on the user scale. 
3) For the continuous attribute, 
the aggregator divides the total domain into slots of a fixed length 
and bucketize the records at a suitable granularity for queries 
\cite{wang2019answering,cormode2019answering}. 
Therefore, 
when the domain size is 1024, the bucketize granularity is two times that of the domain size 512. 
From \autoref{Robustness under various domain size}, since \myahead is more robust to domain size changes, the aggregator can choose a larger domain size to obtain a higher granularity. 
4) When private data tends to have a high skewness, such as data from income, social networks and Web surfing, practitioners should give priority to \myahead from \autoref{The benefit from high skewness}. 
5) From \autoref{Robustness under various attribute correlation},
in practice, \myahead can handle the impact of attribute correlation varying.

\end{document}